\documentclass[11pt]{article}

\usepackage[margin=1in]{geometry}
\usepackage{amsmath,amssymb,amsthm,mathtools}
\usepackage{bm}
\usepackage{enumitem}
\usepackage{algorithm}
\usepackage{algpseudocode}
\usepackage{graphicx} 
\usepackage{caption}
\usepackage{booktabs}
\usepackage{xcolor}
\usepackage{float}
\usepackage{tikz}
\usetikzlibrary{decorations.pathreplacing}
\usepackage[authoryear,round]{natbib}
\usepackage{hyperref}
\hypersetup{colorlinks=true,linkcolor=blue,citecolor=blue,urlcolor=blue}
\usepackage{cleveref}

\makeatletter
\providecommand{\theHALG@line}{}
\renewcommand{\theHALG@line}{\thealgorithm.\arabic{ALG@line}}
\makeatother

\newtheorem{theorem}{Theorem}
\newtheorem{proposition}{Proposition}[section]
\newtheorem{lemma}{Lemma}[section]

\newtheorem{assumption}{Assumption}
\crefname{assumption}{assumption}{assumptions}
\Crefname{assumption}{Assumption}{Assumptions}
\theoremstyle{plain}
\newtheorem{remark}{Remark}

\newcommand{\R}{\mathbb R}
\newcommand{\E}{\mathbb E}
\newcommand{\Pp}{\mathbb P}
\newcommand{\Sph}{\mathbb S}
\newcommand{\norm}[1]{\left\lVert #1\right\rVert}
\newcommand{\abs}[1]{\left\lvert #1\right\rvert}
\newcommand{\ip}[2]{\left\langle #1,#2\right\rangle}
\newcommand{\op}{\mathrm{op}}
\newcommand{\F}{\mathrm F}
\newcommand{\argmaxsmall}{\mathop{\mathrm{arg\,max}}}

\title{Optimal detection of general moment changes: Simultaneous mean and covariance change detection and beyond}
\author{Xiaokai Luo$^{1,*}$ \and Chenghao Xu$^{2,*}$ \and Haotian Xu$^3$
\and Carlos Misael Madrid Padilla$^4$ \and Daren Wang$^2$}
\date{$^1$Department of Applied and Computational Mathematics and Statistics, University of Notre Dame\\
$^2$Department of Mathematics, University of California, San Diego\\
$^3$Department of Mathematics and Statistics, Auburn University\\
$^4$Department of Statistics and Data Science, Washington University in St. Louis\\[0.3em]
$^*$Equal Contribution.\\[0.8em]
September 26, 2026}

\begin{document}
\maketitle

\begin{abstract}
We study multiple change-point detection in multivariate time series whose
distributions change in a piecewise constant manner. Distributional changes
can manifest across different moment orders, from shifts in the mean and
covariance to changes in higher-order moments. 
Higher-order moments capture increasingly rich distributional features but
become difficult to estimate in high dimensions. Our tensor representation unifies moments of different orders within a common linear algebraic framework, enabling a new method  to detect changes in moments of all orders up to
a prescribed fixed order $p$. The resulting procedure accommodates temporal dependence and allows the dimension of the time series to grow with the sample size. Under suitable regularity conditions, the proposed procedure  achieves a  localization error rate that matches a newly developed minimax lower bound.  We further derive limiting distributions
under both nonvanishing and vanishing moment jumps and construct
asymptotically valid confidence intervals in the vanishing-jump regime.
Numerical experiments and real-data analyses demonstrate the method's
effectiveness in detecting moment changes and a range of distributional
shifts.
\end{abstract}

\section{Introduction}

Multivariate time series analysis is central to learning from dynamic systems
whose components evolve jointly. Machine learning applications range
from clinical
monitoring \citep{HylandEtAl2020} and epidemic growth
rates inference \citep{DehningEtAl2020} to global weather forecasting \citep{LamEtAl2023} and electricity-demand and traffic
prediction \citep{LiuEtAl2024iTransformer}.
Such data can span multiple statistical regime transitions. Fitting a single model across
those changes can obscure their important transition structure and may lead to misleading conclusions. Change-point
analysis partitions observations into several homogeneous segments and
quantifies uncertainty about the change points.

We consider the scenario where  \(X_1,\ldots,X_n\in\mathbb{R}^D\) is a time series whose distribution may change at multiple unknown times. We allow the dimension \(D\) to grow with the sample size \(n\).
Under the sub-Gaussian tail condition in \Cref{ass:data-conditions},
the full collection of joint moments uniquely determines the
distribution, so every distributional change must alter at least one
moment. This observation motivates a unified approach to distributional
change detection through moment changes.

Throughout, $p\geq2$ is a fixed integer that does not grow with $n$. 
 We study changes in the moments
$\E X_t^{\otimes r}$, $r=1,\ldots,p$, where $X_t^{\otimes r}$ denotes
the $r$-fold tensor product of $X_t$. We assume that these moment tensors
are piecewise constant, with an unknown number $K\geq1$ of change points
$0<\eta_1<\cdots<\eta_K<n$. Specifically,
\begin{align}
\E X_t^{\otimes r}\ne\E X_{t+1}^{\otimes r}
\quad\text{for at least one }r\in\{1,\ldots,p\}
\qquad\text{if and only if }\quad
t\in\{\eta_1,\ldots,\eta_K\}
\label{eq:piecewise-x-moment-model}
\end{align}

To detect changes in moments of different orders simultaneously, we
introduce
\begin{align}
Y_t&=\begin{pmatrix}1\\X_t\end{pmatrix}\in\R^{D+1},\qquad t=1,\ldots,n,
\nonumber
\end{align}
and define its order-$p$ moment tensor as
\begin{align}
\mathcal M_t^{(p)}&=\E Y_t^{\otimes p}.
\nonumber
\end{align}
\begin{remark}[Hierarchy of moment tensors]
For $1\leq r\leq p$, write $\mathcal M_t^{(r)}=\E Y_t^{\otimes r}$.
Since $Y_{t,1}=1$,
\begin{align}
\bigl(\mathcal M_t^{(r+1)}\bigr)_{1,i_1,\ldots,i_r}
&=\bigl(\mathcal M_t^{(r)}\bigr)_{i_1,\ldots,i_r},
\qquad 1\leq r<p,
\end{align}
for $i_1,\ldots,i_r\in\{1,\ldots,D+1\}$. Thus, $\mathcal M_t^{(p)}$
contains all lower-order moment tensors of $Y_t$ and all joint moments of
$X_t$ through order $p$. For example, when $p=3$, the first three moment orders appear as
\begin{align}
\text{Mean (first-order moments):}\qquad
\bigl(\mathcal M_t^{(3)}\bigr)_{1,1,i+1}
&=\E X_{t,i},
\nonumber\\[4pt]
\text{Second-order moments:}\qquad
\bigl(\mathcal M_t^{(3)}\bigr)_{1,i+1,j+1}
&=\E(X_{t,i}X_{t,j}),
\nonumber\\[4pt]
\text{Third-order moments:}\qquad
\bigl(\mathcal M_t^{(3)}\bigr)_{i+1,j+1,k+1}
&=\E(X_{t,i}X_{t,j}X_{t,k}),
\nonumber
\end{align}
where $i,j,k\in\{1,\ldots,D\}$.
\end{remark}

The moment-change model \eqref{eq:piecewise-x-moment-model} is therefore
equivalent to
\begin{align}
\mathcal M_t^{(p)}\ne\mathcal M_{t+1}^{(p)}
\quad\text{if and only if}\quad
t\in\{\eta_1,\ldots,\eta_K\},\qquad t=1,\ldots,n-1.
\label{eq:piecewise-moment-model}
\end{align}
The difficulty of detecting a change depends on its magnitude and its
separation from neighboring changes. For $k=1,\ldots,K$, define the jump
tensor
\begin{align}
\Theta_k&=\mathcal M_{\eta_k}^{(p)}-\mathcal M_{\eta_k+1}^{(p)}
\label{eq:moment-jump}
\end{align}
and its magnitude
\begin{align}
\kappa_k&=\norm{\Theta_k}_{\op},\qquad k=1,\ldots,K,
\label{eq:jump-size}
\end{align}
and let $\kappa_{\min}=\min_{1\leq k\leq K}\kappa_k>0$.
Set $\eta_0=0$ and $\eta_{K+1}=n$. The segment lengths are
\begin{align}
\Delta_k&=\eta_{k+1}-\eta_k,\qquad k=0,\ldots,K,
\label{eq:segment-spacing}
\end{align}
and the minimum spacing is $\Delta_{\min}=\min_{0\leq k\leq K}\Delta_k$. Our goal is to estimate $K$ and the
locations $\eta_1,\ldots,\eta_K$, and to perform statistical inference on
each change-point location.
We impose the following assumption to allow temporal dependence in the observed time series.  
\begin{assumption}
\label{ass:data-conditions}
Let $X_1,\ldots,X_n$ be a time series in $\R^D$   satisfying the following conditions.
\begin{enumerate}[label=\textbf{(\Alph*)},ref=(\Alph*),leftmargin=*,itemsep=4pt]
\item \textbf{Sub-Gaussian tails.}
There exists a fixed constant $1\leq\sigma<\infty$, independent of $n$, such that
\begin{align}
\sup_{1\leq t\leq n}\sup_{v\in\Sph^{D-1}}
\norm{v^\top X_t}_{\psi_2}
&\leq\sigma.
\nonumber
\end{align}
\item \textbf{$\alpha$-mixing.}
The $\alpha$-mixing coefficients (see  the definition in \eqref{eq:alpha-mixing}) of $\{X_t\}_{t = 1}^n$ decay at the
generalized-exponential rate
\begin{align}
\alpha_X(\ell)&\leq b_0\exp(-b_1\ell^\gamma),\qquad 1\leq\ell<n,
\label{eq:model-geometric-mixing}
\end{align}
where $b_0,b_1>0$ and $\gamma\in[1,\infty]$ are fixed constants.
\end{enumerate}
\end{assumption}
The case $\gamma=\infty$ corresponds to temporal independence,
characterized by $\alpha_X(\ell)=0$ for every positive lag.
Under \Cref{ass:data-conditions}, $Y_t$ has the same mixing coefficients
as $X_t$ and a sub-Gaussian scale bounded by a universal constant times
$\sigma$.

\subsection{Related literature}
Moment-based methods often target means or covariance matrices.
\citet{WangSamworth2018} use CUSUM projections for multiple sparse mean
changes under independence. For a single change under temporal dependence, their localization bound achieves the
minimax optimality.
\citet{WangYuRinaldo2021Covariance} use independent projections for
multiple covariance changes under independence and achieving minimax optimal localization
error.
\citet{DettePanYang2022} screen covariance entries and prove consistency for a single change under independence, assuming
the mean is constant.  For broader distributional changes,
\citet{KanrarJiangCai2025} use classifier AUC to test for and locate both single change and multiple changes under independence. Their localization error bound additionally requires negligible classifier error. Kernel-based methods for detecting distributional changes include
\citet{ArlotCelisseHarchaoui2019} and \citet{madrid2023change}. When changes
take a parametric form, such as a shift in the mean or covariance, general-purpose
kernel methods can be  less effective than parametric methods.
Indeed, \citet{MadridPadillaYuWangRinaldo2022} establish the minimax
localization rate for kernel density-based change-point detection whose
dependence on the observation dimension is exponential, illustrating the
curse of dimensionality.

A closely related approach is developed by \citet{CuiPanWangZou2025}
for the simultaneous detection of mean and covariance changes.
They construct separate U-statistic scans for mean and covariance, and combine their
$p$-values using Fisher's method. For location estimation, they compute
the Fisher statistic at each candidate split and choose its maximizer.
Their main focus is
testing under independence and a single change point. They establish consistency of the location estimator but no localization rate or location limiting distribution. In contrast, our method addresses
multiple changes in joint moments through order $p$ under temporal
dependence, with minimax-optimal localization and distributional
inference for the change locations.

Change-point location limiting distributions have been established for
multiple univariate mean changes with i.i.d.\ errors, together with
bootstrap confidence intervals \citep{ChoKirch2022Bootstrap}; single
functional mean changes \citep{AueEtAl2009}; and multiple covariance
changes in fragmented functional data \citep{XueXuYu2026}.
Under temporal dependence, change-point location limits are available for sparse high-dimensional regression
\citep{XuWangZhaoYu2024}, functional regression
\citep{KumarEtAl2024Functional}. We derive change-point location limit laws for nonvanishing and vanishing regime under temporal dependence and construct data-driven confidence
intervals with asymptotically valid coverage in the vanishing regime.

\subsection{Contributions}

\begin{enumerate}[leftmargin=*,itemsep=4pt]
\item \textbf{Simultaneous detection across moment orders.}
Distributional changes can affect different moments. We develop a unified procedure that detects
changes in all joint moments through any prescribed fixed order $p$.
The framework can capture  mean and covariance changes jointly and extend 
simultaneous to higher-order changes.

\item \textbf{Joint dimension reduction for moment-based detection.}
Higher-order moments reveal more complex distributional features but become
increasingly difficult to estimate in high dimensions. Our tensor
representation brings different moment orders into a unified linear algebra
framework, allowing us to construct powerful test statistics that combine
information across moment orders to detect changes. The resulting method
enables precise change-point localization at a rate that explicitly reflects
the moment order $p$.

\item \textbf{Minimax-optimal localization.}
We establish the localization rate and prove a newly developed minimax lower bound that matches this rate.   These results reveal the statistical difficulty of locating changes using higher-order moment information.

\item \textbf{Data-driven confidence interval.}
We quantify uncertainty in the estimated change-point
locations by deriving limiting distributions for both vanishing and nonvanishing jumps under
temporal dependence, see \Cref{sec:procedure-inference}. For vanishing regime, we further construct data-driven confidence inter-
vals with asymptotically valid coverage, as established in \Cref{thm:ci-coverage}.
\end{enumerate}

\subsection{Notation}
\label{sec:notation}
\textbf{Intervals.}
For integer endpoints $0\leq s<e\leq n$, $(s,e]$ denotes the integer
indices strictly greater than $s$ and at most $e$, with cardinality $e-s$.
\\
\textbf{Vectors.}
For $v\in\R^d$, let $v_i$ denote its $i$th coordinate and $\norm v_2$
its Euclidean norm. Write $\Sph^{d-1}=\{v\in\R^d:\norm v_2=1\}$.
Throughout the paper, for $u,v\in\Sph^{d-1}$, the distance up to sign is
$d_\pm(u,v)=\min\{\norm{u-v}_2,\norm{u+v}_2\}$.
\\
\textbf{Matrices and tensors.}
Matrix and tensor entries are denoted by $A_{ij}$ and
$\mathcal A_{i_1,\ldots,i_p}$. Write $u_1\otimes\cdots\otimes u_p$
for the outer product of vectors $u_1,\ldots,u_p$, and $u^{\otimes p}$
when all factors equal $u$. The Frobenius inner product
$\ip{\mathcal A}{\mathcal B}_{\F}$ is the sum of products of corresponding
entries, with $\norm{\mathcal A}_{\F}^2=\ip{\mathcal A}{\mathcal A}_{\F}$.
For $\mathcal A\in\R^{d\times\cdots\times d}$, define
\begin{align}
\norm{\mathcal A}_{\op}
&=\sup_{u_1,\ldots,u_p\in\Sph^{d-1}}
\abs{\ip{\mathcal A}{u_1\otimes\cdots\otimes u_p}_{\F}}.
\nonumber
\end{align}
\\
\textbf{Moments and tails.}
For a scalar random variable $W$, let $\norm W_{L^q}=(\E\abs W^q)^{1/q}$,
$q\geq1$. For $\beta>0$, define
\begin{align}
\norm W_{\psi_\beta}
&=\sup_{q\geq1}q^{-1/\beta}\norm W_{L^q}.
\nonumber
\end{align}
\\
\textbf{Temporal dependence.}
For sigma-fields $\mathcal A,\mathcal B$, define
\begin{align}
\alpha(\mathcal A,\mathcal B)
&=\sup_{A\in\mathcal A,\,B\in\mathcal B}
\abs{\Pp(A\cap B)-\Pp(A)\Pp(B)}.
\nonumber
\end{align}
For $W_1,\ldots,W_n$, the $\alpha$-mixing coefficients are
\begin{align}
\alpha_W(\ell)
&=\sup_{1\leq s\leq n-\ell}
\alpha\{\sigma(W_1,\ldots,W_s),
\sigma(W_{s+\ell},\ldots,W_n)\},
\qquad 1\leq\ell<n.
\label{eq:alpha-mixing}
\end{align}
\\
\begin{samepage}
\textbf{Asymptotic notation and constants.}
Throughout the paper, the tail scale $1\leq\sigma<\infty$ in
\Cref{ass:data-conditions} is fixed and does not depend on $n$.
We retain its powers in the bounds to display their scale dependence.
We use the standard asymptotic notation $O$, $o$, $O_{\Pp}$, and
$o_{\Pp}$, and write $\xrightarrow{\Pp}$ and $\xrightarrow{d}$ for
convergence in probability and distribution. Unless otherwise stated,
$C,C_1,C_2,\ldots$ are positive constants independent of $n$ and $D$,
but may depend on $p$ and the mixing parameters $\gamma,b_0,b_1$ in
\Cref{ass:data-conditions}.
\par
\end{samepage}

\section{Change-point localization}
\label{sec:change-point-localization}

We propose a two-stage procedure for change-point localization. The first
stage, discussed in \Cref{sec:preliminary-seeded-detection},
evaluates tensor CUSUM statistics over seeded intervals to estimate
the number of changes and obtain preliminary locations. The second stage, discussed in \Cref{sec:change-point-local-refinement},
uses these preliminary locations as proxies  
to estimate the adjacent moment tensors and their leading jump directions.
These estimates guide a local refinement step that yields more precise
change-point locations.

\subsection{Preliminary seeded detection}
\label{sec:preliminary-seeded-detection}

We begin by defining the empirical moment tensor on an interval $(s,e]$,
\begin{align}
\overline{\mathcal M}_{(s,e]}^{(p)}
&=\frac{1}{e-s}\sum_{t=s+1}^{e}Y_t^{\otimes p}.
\nonumber
\end{align}
For any candidate  $t\in(s,e)$, define the tensor CUSUM statistic
\begin{align}
\mathcal T_{(s,e]}(t)
&=\sqrt{\frac{(t-s)(e-t)}{e-s}}
\left\{
\overline{\mathcal M}_{(s,t]}^{(p)}
-\overline{\mathcal M}_{(t,e]}^{(p)}
\right\}.
\nonumber
\end{align}

When the population moment tensor is constant on $(s,e]$, we have
$\E\mathcal T_{(s,e]}(t)=0$. If the interval contains a single change
point $\eta\in(s,e-1]$, then
$\norm{\E\mathcal T_{(s,e]}(t)}_{\op}$ is uniquely maximized at $t=\eta$.
These properties allow us to use $\norm{\mathcal T_{(s,e]}(t)}_{\op}$ to
measure the signal of a change.

To efficiently locate the change points using the tensor cusum statistics,    we evaluate these
statistics over a deterministic collection of overlapping intervals at
multiple scales, called seeded intervals \citep{KovacsBuhlmannLiMunk2023}.
Define dyadic interval scale set as 
\begin{align}
\mathcal L_n
&=\left\{2^r:r=4,\ldots,\left\lfloor\log_2n\right\rfloor\right\}.
\nonumber
\end{align}
The intervals at scale $\ell\in\mathcal L_n$ are
\begin{align}
\mathcal I_\ell
&=\left\{
\left(\frac{j\ell}{2},\frac{j\ell}{2}+\ell\right]:
j=0,\ldots,\left\lfloor\frac{2(n-\ell)}{\ell}\right\rfloor
\right\}\cup\{(n-\ell,n]\}.
\nonumber
\end{align}
At scale $\ell$, $\mathcal I_\ell$ contains intervals of length $\ell$
starting at equally spaced grid points $0,\ell/2,\ell,\ldots$, together
with the final interval $(n-\ell,n]$.
Taking the union over scales gives the seeded collection
\begin{align}
\mathcal I&=\bigcup_{\ell\in\mathcal L_n}\mathcal I_\ell.
\label{eq:preliminary-seeded-collection}
\end{align}
Seeded intervals have proved effective for univariate mean change detection
\citep{KovacsBuhlmannLiMunk2023}. We adapt this construction to detect
changes in multivariate moment tensors.

\begin{algorithm}[!htb]
\caption{Preliminary seeded detection}
\label{alg:preliminary-seeded-detection}
\begin{algorithmic}[1]
\Require Observations $X_1,\ldots,X_n$, moment order $p$, and parameters $\gamma,\tau>0$.
\State $Y_t\gets (1, X_t^{\top})^{\top}$, $t=1,\ldots,n$.
\State Construct $\mathcal I$ according to \eqref{eq:preliminary-seeded-collection}.
\For{$(s,e]\in\mathcal I$}
  \State Compute
  \begin{align}
  G_{(s,e]}&\gets
  \max_{s+(e-s)/4\leq t\leq e-(e-s)/4}
  \norm{\mathcal T_{(s,e]}(t)}_{\op}.
  \label{eq:preliminary-interval-statistic}
  \end{align}
\EndFor
\State Set
\begin{align}
\mathcal I_{\mathrm{act}}&\gets
\left\{(s,e]\in\mathcal I:
G_{(s,e]}>\tau\left[\sqrt{D+\log n}
+\frac{(D+\log n)^{p/2+1/\gamma}}{\sqrt{e-s}}\right]\right\}.
\label{eq:preliminary-threshold}
\end{align}
\State $\widetilde{\mathcal E}\gets\varnothing$.
\While{$\mathcal I_{\mathrm{act}}\ne\varnothing$}
  \State Choose the shortest $(s,e]\in\mathcal I_{\mathrm{act}}$.
  \State $\widetilde\eta\gets
  \min\operatorname*{arg\,max}_{s+(e-s)/4\leq t\leq e-(e-s)/4}
  \norm{\mathcal T_{(s,e]}(t)}_{\op}$.
  \State $\widetilde{\mathcal E}\gets
  \widetilde{\mathcal E}\cup\{\widetilde\eta\}$.
  \State $\mathcal I_{\mathrm{act}}\gets
  \{(a,b]\in\mathcal I_{\mathrm{act}}:(a,b]\cap(s,e]=\varnothing\}$.
\EndWhile
\State \Return The estimated change points
$\widetilde\eta_1<\cdots<\widetilde\eta_{\widetilde K}$ of
$\widetilde{\mathcal E}$ in increasing order.
\end{algorithmic}
\end{algorithm}

To obtain preliminary change-point estimates,
\Cref{alg:preliminary-seeded-detection} applies narrowest-over-threshold
selection \citep{BaranowskiChenFryzlewicz2019} to the tensor CUSUM
statistics on seeded intervals.
The statistic $G_{(s,e]}$, defined in \eqref{eq:preliminary-interval-statistic},
measures the strongest evidence of a moment
change within the central half of $(s,e]$. This restriction avoids splits
near interval boundaries, where empirical higher-order moments can be
unstable because  of  few observations. The threshold in
\eqref{eq:preliminary-threshold}
adjusts for interval length, moment order, and temporal dependence, with
the tuning parameter $\tau$ controlling detection sensitivity.

Among intervals exceeding the threshold,
\Cref{alg:preliminary-seeded-detection} selects a shortest interval to
favor isolating an individual change. Within the selected interval, we
estimate the change point by the split with the largest test statistic.
We then remove overlapping intervals to avoid detecting the same change. The resulting locations
initialize the second-stage refinement.
We next establish conditions under which the preliminary procedure
consistently estimates the number and locations of the change points.

\begin{assumption}
\label{ass:preliminary-signal}
For a sufficiently large constant $C_{\mathrm I}>0$,
\begin{align}
\frac{\Delta_{\min}\kappa_{\min}^2}{\sigma^{2p}}
&\geq C_{\mathrm I}(D+\log n),
\qquad
\frac{\Delta_{\min}\kappa_{\min}}{\sigma^p}
\geq\sqrt{C_{\mathrm I}}\,(D+\log n)^{p/2+1/\gamma}.
\nonumber
\end{align}
\end{assumption}

Under temporal independence setting that $\gamma=\infty$, if
$\Delta_{\min}\geq(D+\log n)^{p-1}$, then
\Cref{ass:preliminary-signal} reduces to
\begin{align}
\frac{\Delta_{\min}\kappa_{\min}^2}{\sigma^{2p}}
&\geq C_{\mathrm I}(D+\log n).
\label{eq:preliminary-independent-signal}
\end{align}
  For $p=2$,
\eqref{eq:preliminary-independent-signal} up to logarithmic factors, matches the covariance
change-point detection condition of \citet{WangYuRinaldo2021Covariance},
and  is minimax optimal. In addition,
for $p=2$, the requirement $\Delta_{\min}\gtrsim D$ also appears in
\citet{WangYuRinaldo2021Covariance}.

\begin{theorem}
\label{thm:preliminary-localization}
Let $\{ \widetilde \eta_k\}_{k=1}^{\widetilde K}$ be the estimated change points return by  \Cref{alg:preliminary-seeded-detection} with input  $\gamma$ and $\tau=c_\tau\sigma^p$, where $c_\tau>0$ is sufficiently large.
Suppose   
\Cref{ass:data-conditions,ass:preliminary-signal} hold.
 Then, with probability at least $1-n^{-5}$,
\Cref{alg:preliminary-seeded-detection} returns estimates satisfying
\begin{align}
\widetilde K&=K,
\qquad
\max_{1\leq k\leq K}\abs{\widetilde\eta_k-\eta_k}
\leq\frac{3\Delta_{\min}}{64}.
\label{eq:preliminary-localization-bound}
\end{align}
\end{theorem}

\begin{remark}
The constant $3/64$ in \Cref{thm:preliminary-localization} is chosen for
convenience. The preliminary estimates need only be sufficiently accurate
relative to the minimum spacing to ensure that there is a one to one correpsondece between the estimate change points and the true ones. 
 \Cref{thm:preliminary-localization} provides sufficient  accuracy
for the subsequent refinement procedure.
\end{remark}

\subsection{Local refinement}
\label{sec:change-point-local-refinement}

To obtain optimal change-point estimates, we use the preliminary estimates
$\{\widetilde \eta _k\}_{k=1}^K$ from
\Cref{alg:preliminary-seeded-detection} to estimate the moment jumps
$\Theta_k$ defined in \eqref{eq:moment-jump}. We then estimate the most
significant projection direction of each jump,
\begin{align}
\nu_k&\in\operatorname*{arg\,max}_{u\in\Sph^D}
\abs{\ip{\Theta_k}{u^{\otimes p}}_{\F}},
\label{eq:population-maximizing-direction}
\end{align}
and use the estimated projection to refine the change-point location.

Set $\widetilde\eta_0=0$ and $\widetilde\eta_{\widetilde K+1}=n$.
To avoid uncertain segment boundaries, we fit moments on the central half
$(a_k,b_k]$ of each estimated segment, where, for
$k=0,\ldots,\widetilde K$,
\begin{align}
a_k&=\left\lceil\frac{3\widetilde\eta_k+\widetilde\eta_{k+1}}{4}\right\rceil-1,
\qquad
b_k=\left\lfloor\frac{\widetilde\eta_k+3\widetilde\eta_{k+1}}{4}\right\rfloor.
\label{eq:central-fitting-interval}
\end{align}
Under \eqref{eq:preliminary-localization-bound}, $\widetilde K =K$ and 
$[a_k,b_k]\subseteq(\eta_k,\eta_{k+1}]$ for $k=0,\ldots,K$.
For $k=1,\ldots, K$, estimate the moment jump by
\begin{align}
\widehat\Theta_k
&=\frac{1}{b_{k-1}-a_{k-1}}
\sum_{t=a_{k-1}+1}^{b_{k-1}}Y_t^{\otimes p}
-\frac{1}{b_k-a_k}\sum_{t=a_k+1}^{b_k}Y_t^{\otimes p}.
\label{eq:estimated-moment-jump}
\end{align}
We then estimate $\nu_k$ in \eqref{eq:population-maximizing-direction} by
\begin{align}
\widehat \nu_k
&\in\operatorname*{arg\,max}_{u\in\Sph^D}
\abs{\ip{\widehat\Theta_k}{u^{\otimes p}}_{\F}}.
\label{eq:estimated-leading-direction}
\end{align}
Under standard identifiability and eigengap-type conditions on $\Theta_k$
(see \Cref{ass:local-refinement-separation} in the appendix),
we establish an error bound for the estimated projection directions in
\Cref{app:direction-setup} of the appendix.
For $k=1,\ldots,  K$, we refine the location within $(s_k,e_k]$,
with endpoints
\begin{align}
s_k&=\left\lfloor\frac{9\widetilde\eta_{k-1}+\widetilde\eta_k}{10}\right\rfloor,
\qquad
e_k=\left\lfloor\frac{\widetilde\eta_k+9\widetilde\eta_{k+1}}{10}\right\rfloor.
\nonumber
\end{align}
Under \eqref{eq:preliminary-localization-bound}, each window
$(s_k,e_k]$ contains exactly one true change point, $\eta_k$.
For a candidate  $t\in(s_k,e_k-1]$, define 
\begin{align}
\widehat Q_k(t)
&=\sum_{i=s_k+1}^{t}
\left\{(\widehat \nu_k^\top Y_i)^p
-\frac{\sum_{j=a_{k-1}+1}^{b_{k-1}}(\widehat  \nu_k^\top Y_j)^p}
{b_{k-1}-a_{k-1}}\right\}^2
\nonumber\\
&\quad+\sum_{i=t+1}^{e_k}
\left\{(\widehat  \nu_k^\top Y_i)^p
-\frac{\sum_{j=a_k+1}^{b_k}(\widehat  \nu_k^\top Y_j)^p}
{b_k-a_k}\right\}^2.
\nonumber
\end{align}
The refined estimate is
\begin{align}
\widehat\eta_k
&=\min\operatorname*{arg\,min}_{t\in(s_k,e_k-1]}\widehat Q_k(t).
\label{eq:local-refined-estimate}
\end{align}

To obtain theoretical guarantees of the refined change point estimators $\{ \widehat \eta_k\}_{k=1}^K $, we
impose the following signal-strength conditions.

\begin{assumption}
\label{ass:local-refinement-signal}
As $n\to\infty$,
\begin{align}
\frac{\Delta_{\min}\kappa_{\min}^2}
{\sigma^{2p}D(D+\log n)}
&\to\infty
\quad\text{and}\quad
\frac{\Delta_{\min}\kappa_{\min}}
{\sigma^p\sqrt D\,(D+\log n)^{p/2+1/\gamma}}
\to\infty.
\nonumber
\end{align}
\end{assumption}

\Cref{ass:local-refinement-signal} strengthens
\Cref{ass:preliminary-signal} to ensure accurate estimation of the
order-$p$ moment tensors relative to the jump magnitudes.  
The conditions allow both a growing dimension and vanishing
jump magnitudes, provided the segments contain sufficiently many
observations to estimate the higher-order moments accurately.

\begin{theorem}
\label{thm:local-refinement}
Suppose the moment-change model \eqref{eq:piecewise-moment-model} holds,
$p\geq2$ is fixed, and
\Cref{ass:data-conditions,ass:local-refinement-separation,ass:local-refinement-signal}
hold. Use \Cref{alg:preliminary-seeded-detection} with
$\tau=c_\tau\sigma^p$ for a sufficiently large fixed constant $c_\tau$.
Then the refined estimators in \eqref{eq:local-refined-estimate} satisfy
\begin{align}
\max_{1\leq k\leq K}
\frac{\kappa_k^2}{\sigma^{2p}}\abs{\widehat\eta_k-\eta_k}
&=O_{\Pp}(1).
\nonumber
\end{align}
\end{theorem}
\Cref{thm:local-refinement} shows that the refined localization error for
each change point is $O_{\Pp}(\sigma^{2p}/\kappa_k^2)$.
The constants and signal requirements depend on the fixed order $p$.
Under additional assumptions, \Cref{sec:procedure-inference}
establishes limiting distributions for the refined estimators and constructs
asymptotically valid confidence intervals for vanishing jumps.

\subsection{Minimax lower bound}
\label{sec:minimax-lower-bound}

We establish a minimax lower bound matching the localization rate in
\Cref{thm:local-refinement}. For positive integers $n,D,\Delta$ with
$2\Delta\leq n$ and $\kappa>0$, let
$\mathcal P(n,D,\Delta,\sigma,\kappa)$ denote the class of joint
distributions of independent observations with the fixed tail scale
$\sigma$ from \Cref{ass:data-conditions}, satisfying
\begin{align}
X_1,\ldots,X_\eta&\overset{\mathrm{i.i.d.}}{\sim}F_L,
\qquad
X_{\eta+1},\ldots,X_n\overset{\mathrm{i.i.d.}}{\sim}F_R,
\nonumber
\end{align}
where $F_L,F_R$ are distributions on $\R^D$ and the unknown change point
satisfies $\Delta\leq\eta\leq n-\Delta$.
The observations satisfy the uniform sub-Gaussian bound
\begin{align}
\sup_{1\leq t\leq n}\sup_{u\in\Sph^{D-1}}
\norm{u^\top X_t}_{\psi_2}
&\leq\sigma.
\nonumber
\end{align}
Define the two population moment tensors by
\begin{align}
\mathcal M_L^{(p)}
&=\E_{X\sim F_L}\left[Y^{\otimes p}\right],
\qquad
\mathcal M_R^{(p)}
=\E_{X\sim F_R}\left[Y^{\otimes p}\right],
\nonumber
\end{align}
where $Y= (1, X^{\top})^{\top}$. Suppose that 
$\norm{\mathcal M_L^{(p)}-\mathcal M_R^{(p)}}_{\op}=\kappa$.

\begin{theorem}
\label{thm:minimax-localization-lower}
Fix an integer $p\geq2$ and a constant $1\leq\sigma<\infty$.
Let $D=D_n$ and $\Delta=\Delta_n$ be positive integers with
$\Delta\leq n/4$, and let $\kappa=\kappa_n>0$. Suppose \Cref{ass:local-refinement-signal} holds with
$\Delta_{\min}=\Delta$, $\kappa_{\min}=\kappa$, and $\gamma=\infty$.
There exist constants $c_0,c_1>0$, depending only on $p$, such that,
whenever $0<\kappa\leq c_0\sigma^p$, the following bound holds for all
sufficiently large $n$:
\begin{align}
\inf_{\check\eta}
\sup_{P\in\mathcal P(n,D,\Delta,\sigma,\kappa)}
\E_P\abs{\check\eta-\eta}
&\geq c_1\frac{\sigma^{2p}}{\kappa^2},
\label{eq:lower-risk-bound}
\end{align}
where the infimum ranges over all estimators based on $X_1,\ldots,X_n$.
\end{theorem}

 \section{Limiting distributions and confidence intervals}
\label{sec:procedure-inference}

We next quantify uncertainty in the refined change-point estimates.
We establish their limiting distributions under nonvanishing and vanishing
moment jumps, and use the vanishing-jump limit to construct data-driven
confidence intervals.

\subsection{The nonvanishing regime}

  We impose the following additional conditions on the observations
around $\eta_k$.

\begin{assumption}
\label{ass:nonvanishing-observation-limit}
\begin{enumerate}[label=\textbf{(\Alph*)},ref=(\Alph*),leftmargin=*,itemsep=4pt]
\item The moment order $p$ is fixed, and, as $n\to\infty$,
\begin{align}
\kappa_k&\to\kappa_k^*>0.
\nonumber
\end{align}
The dimension $D$ may grow with $n$.
\item There exists a real-valued process
$\{\xi_{k,t}\}_{t\in\mathbb Z}$, whose law does not depend on $n$,
such that, for every fixed integer $H\geq1$,
\begin{align}
\bigl(\nu_k^\top Y_{\eta_k+t}\bigr)_{t=-H+1}^{H}
&\xrightarrow{d}
\bigl(\xi_{k,t}\bigr)_{t=-H+1}^{H}.
\nonumber
\end{align}
Here $\nu_k$ is defined  in
\eqref{eq:population-maximizing-direction}, with its sign chosen
consistently. 
\item For every fixed integer $H\geq1$, the vector
$(\xi_{k,-H+1},\ldots,\xi_{k,H})$ has a joint density
with respect to Lebesgue measure on $\R^{2H}$.
\end{enumerate}
\end{assumption}

Condition (A) of \Cref{ass:nonvanishing-observation-limit} follows the
usual formulation of the nonvanishing regime in change-point analysis.
Condition (B) holds, for example, when $D$ is fixed, $\nu_k$ converges,
and the segment
$(X_{\eta_{k-1}+1},\ldots,X_{\eta_k},X_{\eta_k+1},\ldots,X_{\eta_{k+1}})$
admits an infinite extension whose joint law does not depend on $n$. This latter formulation is
commonly used in the change-point literature. Condition (C) is a standard
continuity condition that  ensures uniqueness of the
limiting random-walk minimizer.

Define the population limiting moment jump by
\begin{align}
\rho_k
&=\E\bigl\{\xi_{k,0}^p\bigr\}
-\E\bigl\{\xi_{k,1}^p\bigr\}.
\label{eq:limit moment jump size}
\end{align}
Under \Cref{ass:data-conditions,ass:nonvanishing-observation-limit},
it follows from \eqref{eq:limit moment jump size} that   $\abs{\rho_k}=\kappa_k^*$.
Define the two-sided random walk
\begin{align}
\mathcal L_k(r)&=
\begin{cases}
|r|(\kappa_k^*)^2
+2\rho_k\displaystyle\sum_{t=r+1}^{0}
\left[\xi_{k,t}^p-\E\bigl\{\xi_{k,t}^p\bigr\}\right],&r<0,\\[6pt]
0,&r=0,\\[6pt]
r(\kappa_k^*)^2
-2\rho_k\displaystyle\sum_{t=1}^{r}
\left[\xi_{k,t}^p-\E\bigl\{\xi_{k,t}^p\bigr\}\right],&r>0,
\end{cases}
\qquad r\in\mathbb Z.
\label{eq:nonvanishing-random-walk}
\end{align}

\begin{theorem}
\label{thm:nonvanishing-localization-limit}
Fix $k\in\{1,\ldots,K\}$. Suppose  
\Cref{ass:data-conditions,ass:local-refinement-separation,ass:local-refinement-signal,ass:nonvanishing-observation-limit}
hold. 
Then the refined estimator in
\eqref{eq:local-refined-estimate} satisfies
\begin{align}
\widehat\eta_k-\eta_k
&\xrightarrow{d}\operatorname*{arg\,min}_{r\in\mathbb Z}\mathcal L_k(r).
\nonumber
\end{align}
\end{theorem}

\subsection{The vanishing regime}

We study the regime in which the moment jump size $\kappa_k$, defined in
\eqref{eq:jump-size}, tends to zero. As we show in this subsection, the
rescaled localization error $\kappa_k^2(\widehat\eta_k-\eta_k)$ converges
in distribution to the minimizer of a two-sided Brownian motion with
drift. To establish this result, we impose the following assumptions.

\begin{assumption}
\label{ass:vanishing-projected-process}
\begin{enumerate}[label=\textbf{(\Alph*)},ref=(\Alph*),leftmargin=*,itemsep=4pt]
\item The moment order $p$ is fixed, and, as $n\to\infty$,
\begin{align}
\kappa_k&\to0.
\nonumber
\end{align}
\item For each $n$, the segment
$(X_{\eta_{k-1}+1},\ldots,X_{\eta_{k+1}})$ admits an extension
indexed by $\mathbb Z$ that is strictly stationary separately on
$t\leq\eta_k$ and $t>\eta_k$. The extended time series satisfies the
tail and mixing bounds in \Cref{ass:data-conditions}, with the mixing
bound imposed at every positive lag.
\end{enumerate}
\end{assumption}

Using the population direction $\nu_k$ in
\eqref{eq:population-maximizing-direction}, define  \begin{align}
Z_{k,t}
&=(\nu_k^\top Y_t)^p-\E\bigl\{(\nu_k^\top Y_t)^p\bigr\},
\qquad t=1,\ldots,n.
\label{eq:vanishing-projected-noise}
\end{align}
For the variance limits below, evaluate the same formula on the two-regime
extension in \Cref{ass:vanishing-projected-process}. This defines
$Z_{k,t}$ for every $t\in\mathbb Z$ and agrees with the observed
projection on $(\eta_{k-1},\eta_{k+1}]$.

\begin{assumption}
\label{ass:vanishing-partial-sum-variance}
There exist constants $\omega_{k,L},\omega_{k,R}\in(0,\infty)$ such that,
for every integer sequence $m=m_n\to\infty$, as $n\to\infty$,
\begin{align}
\frac{1}{m}
\operatorname{Var}\left(\sum_{j=1}^{m}Z_{k,\eta_k+j}\right)
\longrightarrow\omega_{k,R}^2
 \quad \text{and} \quad 
\frac{1}{m}
\operatorname{Var}\left(\sum_{j=1}^{m}Z_{k,\eta_k-j+1}\right)
\longrightarrow\omega_{k,L}^2. 
\nonumber
\end{align}
 
\end{assumption}
Let $B_{k,L}$ and $B_{k,R}$ be independent standard Brownian motions.
Define the two-sided Brownian motion with drift by
\begin{align}
\mathcal G_k(r)&=
\begin{cases}
|r|+2\omega_{k,L}B_{k,L}(-r),&r<0,\\
0,&r=0,\\
r+2\omega_{k,R}B_{k,R}(r),&r>0,
\end{cases}
\qquad r\in\R.
\label{eq:vanishing-brownian-process}
\end{align}

\begin{theorem}
\label{thm:vanishing-localization-limit}
Fix $k\in\{1,\ldots,K\}$. Suppose
\Cref{ass:data-conditions,ass:local-refinement-separation,ass:local-refinement-signal,ass:vanishing-projected-process,ass:vanishing-partial-sum-variance}
hold. Then   the refined estimator in
\eqref{eq:local-refined-estimate} satisfies
\begin{align}
\kappa_k^2(\widehat\eta_k-\eta_k)
&\xrightarrow{d}\operatorname*{arg\,min}_{r\in\R}\mathcal G_k(r).
\nonumber
\end{align}
\end{theorem}

\subsection{Confidence intervals}

We propose a data-driven procedure with theoretical guarantees for
constructing confidence intervals for $\eta_k$ based on the refined
estimator $\widehat\eta_k$. To do so, we first estimate the jump
magnitude $\kappa_k$.
Recall that $(a_k,b_k]$, with endpoints defined in
\eqref{eq:central-fitting-interval}, lies within the homogeneous segment
$(\eta_k,\eta_{k+1}]$ with high probability under the conditions of
\Cref{thm:preliminary-localization}.
Using the estimated leading direction $\widehat\nu_k$ defined in
\eqref{eq:estimated-leading-direction}, we estimate $\kappa_k$ by
comparing the projected moment averages on the two adjacent fitting
intervals:
\begin{align}
\widehat\kappa_k
&=\left|
\frac{1}{b_{k-1}-a_{k-1}}
\sum_{t=a_{k-1}+1}^{b_{k-1}}(\widehat\nu_k^\top Y_t)^p
-\frac{1}{b_k-a_k}
\sum_{t=a_k+1}^{b_k}(\widehat\nu_k^\top Y_t)^p
\right|.
\nonumber
\end{align}

Note that by \eqref{eq:estimated-moment-jump}, \begin{align}
\widehat\kappa_k
&=\abs{\ip{\widehat\Theta_k}{\widehat\nu_k^{\otimes p}}_{\F}}
\approx\abs{\ip{\Theta_k}{\nu_k^{\otimes p}}_{\F}}
=\kappa_k,
\nonumber
\end{align}
which motivates $\widehat\kappa_k$ as an estimator of $\kappa_k$.

We next estimate the long-run variances of $Z_{k,t}$, defined in
\eqref{eq:vanishing-projected-noise}, before and after the change at
$\eta_k$.
Denote \begin{align}
\ell_n&=\left\lfloor(\log n)^{(1+1/\gamma)/2}\right\rfloor.
\nonumber
\end{align}
The Bartlett estimators \citep{NeweyWest1987} on the left and right
fitting intervals are, respectively,
\begin{align}
\widehat\omega_{k,L}^2
&=\frac{1}{b_{k-1}-a_{k-1}}
\sum_{\substack{a_{k-1}<t,r\leq b_{k-1}\\|t-r|\leq\ell_n}}
\left(1-\frac{|t-r|}{\ell_n+1}\right)\nonumber\\
&\quad
\times \left\{(\widehat\nu_k^\top Y_t)^p
-\frac{1}{b_{k-1}-a_{k-1}}
\sum_{j=a_{k-1}+1}^{b_{k-1}}(\widehat\nu_k^\top Y_j)^p\right\}
 \left\{(\widehat\nu_k^\top Y_r)^p
-\frac{1}{b_{k-1}-a_{k-1}}
\sum_{j=a_{k-1}+1}^{b_{k-1}}(\widehat\nu_k^\top Y_j)^p\right\}
\label{eq:ci-left-long-run-variance}
\end{align}
and
\begin{align}
\widehat\omega_{k,R}^2
&=\frac{1}{b_k-a_k}
\sum_{\substack{a_k<t,r\leq b_k\\|t-r|\leq\ell_n}}
\left(1-\frac{|t-r|}{\ell_n+1}\right)
\nonumber\\
&\quad\times\left\{(\widehat\nu_k^\top Y_t)^p
-\frac{1}{b_k-a_k}
\sum_{j=a_k+1}^{b_k}(\widehat\nu_k^\top Y_j)^p\right\}
 \left\{(\widehat\nu_k^\top Y_r)^p
-\frac{1}{b_k-a_k}
\sum_{j=a_k+1}^{b_k}(\widehat\nu_k^\top Y_j)^p\right\}.
\label{eq:ci-right-long-run-variance}
\end{align}
In \eqref{eq:ci-left-long-run-variance} and
\eqref{eq:ci-right-long-run-variance}, the summands with $t=r$ give the
sample variance of $(\widehat\nu_k^\top Y_t)^p$ on the corresponding
fitting interval, while the summands with $t\ne r$ account for temporal
dependence through weighted sample autocovariances.
The estimators $\widehat\omega_{k,L}^2$ and $\widehat\omega_{k,R}^2$
estimate $\omega_{k,L}^2$ and $\omega_{k,R}^2$ in
\Cref{ass:vanishing-partial-sum-variance}. Let $\widehat\omega_{k,L}$
and $\widehat\omega_{k,R}$ denote their nonnegative square roots.
Let $B_{k,L}$ and $B_{k,R}$ be independent standard Brownian motions,
independent of the observations. Define
\begin{align}
\widehat{\mathcal H}_k(u)
&=\begin{cases}
|u|+2\widehat\omega_{k,L}B_{k,L}(-u),&u<0,\\
0,&u=0,\\
u+2\widehat\omega_{k,R}B_{k,R}(u),&u>0,
\end{cases}
\qquad u\in\R.
\label{eq:ci-simulated-process}
\end{align}

For $a\in(0,1)$, let $\widehat q_{k,a}$ be the conditional $a$-quantile,
given the observations, of the minimizer of
\eqref{eq:ci-simulated-process}. These quantiles can be approximated by
simulating independent paths and taking empirical quantiles of their
minimizers. For $\alpha\in(0,1)$, define the confidence interval by
\begin{align}
\mathcal J_{k,\alpha}
&=\left[
\widehat\eta_k-\frac{\widehat q_{k,1-\alpha/2}}{\widehat\kappa_k^2},\,
\widehat\eta_k-\frac{\widehat q_{k,\alpha/2}}{\widehat\kappa_k^2}
\right] .
\label{eq:ci-confidence-interval}
\end{align}

\begin{theorem}
\label{thm:ci-coverage}
Fix $k\in\{1,\ldots,K\}$ and suppose the conditions of
\Cref{thm:vanishing-localization-limit} hold. Then for every fixed $\alpha\in(0,1)$, the confidence interval
in \eqref{eq:ci-confidence-interval} satisfies
\begin{align}
\Pp\bigl(\eta_k\in\mathcal J_{k,\alpha}\bigr)
&\longrightarrow1-\alpha.
\nonumber
\end{align}
\end{theorem}

\section{Simulation and real data example}

In this section, we study the numerical properties of our method and compare
it with state-of-the-art methods targeting mean, covariance, and general
distributional changes. In particular, we compare our method with CPWZ
\citep{CuiPanWangZou2025}, changeAUC
\citep{KanrarJiangCai2025}, a mean-change binary segmentation method, and MNSBS
\citep{madrid2023change} across various
simulated settings and a real-data example.
CPWZ jointly detects changes in the mean
and covariance by combining statistics for the two types of changes.
ChangeAUC detects general distributional changes by training a classifier and
using its area under the receiver operating characteristic curve as the change
statistic. The mean-change method applies binary segmentation to multivariate
mean CUSUM statistics. MNSBS detects general multivariate distributional changes
under
temporal dependence through nonparametric density estimation and seeded binary
segmentation.

\subsection{Tuning parameter selection}\label{sec:tuning-parameters}

For the competing methods, we use their default or paper-recommended tuning
parameters and their associated data-driven selection procedures.

For our method, we set $\gamma=1$ in
\Cref{alg:preliminary-seeded-detection} and jointly select the moment order
$p\in\{2,3\}$ and the threshold parameter
$
\tau\in 
 \{1,1.25,1.5,1.75,2,2.25,2.5,3,4,6,8,12,16\}.
$ For each candidate pair
$(p,\tau)$, we split the time series into its odd- and even-indexed
subsequences, denoted by $X_t^{\mathrm{odd}}=X_{2t-1}$ and
$X_t^{\mathrm{even}}=X_{2t}$. We apply
\Cref{alg:preliminary-seeded-detection}, with the threshold in
\eqref{eq:preliminary-threshold}, to the odd-indexed time series and
subsequently apply the local refinement procedure in
\Cref{sec:change-point-local-refinement}. The resulting change-point estimates
produce the estimated segments
$\widehat{\mathcal I}_1,\ldots,\widehat{\mathcal I}_{\widehat K}$. To select
the candidate pair in a data-driven manner, we evaluate these segments on the
even-indexed observations using an empirical kernel score
\citep{SteinwartZiegel2021}. Let $\mathcal H$ be the reproducing kernel Hilbert
space associated with a bounded Gaussian kernel, and let $\phi$ be its induced
feature map. The score is
\begin{align}
\mathcal S(p,\tau)
&=
\sum_{k=1}^{\widehat K}
\sum_{t\in\widehat{\mathcal I}_k}
\left\|
\phi\left(X_t^{\mathrm{even}}\right)
-\frac{1}{|\widehat{\mathcal I}_k|}
\sum_{s\in\widehat{\mathcal I}_k}
\phi\left(X_s^{\mathrm{odd}}\right)
\right\|_{\mathcal H}^2.
\nonumber
\end{align}
We select the tuning parameters as
$(\widehat p,\widehat\tau)=\arg\min_{p,\tau}\mathcal S(p,\tau)$ and then apply
\Cref{alg:preliminary-seeded-detection} to the complete time series using the
selected parameters.

\subsection{Hausdorff distance}
\label{app:hausdorff-distance}

For nonempty change-point sets $A,B\subseteq\{1,\ldots,n-1\}$, define
\begin{align*}
d_H(A,B)
&=\max\left\{
\max_{a\in A}\min_{b\in B}|a-b|,
\max_{b\in B}\min_{a\in A}|b-a|
\right\}.
\end{align*}
We apply this distance to the true set $A=\{\eta_1,\ldots,\eta_K\}$ and
estimated set $B=\{\widehat\eta_1,\ldots,\widehat\eta_{\widehat K}\}$.
It compares both sets without requiring $\widehat K=K$ and is reported in
observation-index units, without division by $n$.
For empty outputs, the simulation implementation uses the finite convention
\begin{align*}
d_H(A,\varnothing)=d_H(\varnothing,A)
&=\max_{a\in A}\min\{a,n-a\},\qquad A\ne\varnothing,\\
d_H(\varnothing,\varnothing)&=0.
\end{align*}

\subsection{Simulation studies}

We use $D=100$ and 100 independent repetitions per setting. With independent
standard normal initial values $Z_{1j}$ and innovations $\varepsilon_{tj}$,
generate stationary AR(1) sequences
\begin{align*}
Z_{tj}=\rho Z_{t-1,j}+\sqrt{1-\rho^2}\,\varepsilon_{tj},
\qquad t=2,\ldots,n,\quad j=1,\ldots,D.
\end{align*}
\noindent\textbf{Case 1: a single change point.}
Settings R1--R3 have $n=1600$ and a single change at $\eta=880$.

\noindent\textbf{R1 (variance and skewness changes).} With $\rho=0.6$, set
$X_{tj}=Z_{tj}+0.25(Z_{tj}^2-1)$ for $t>\eta$ and $j\leq10$, and
$X_{tj}=Z_{tj}$ otherwise. The mean is unchanged.

\noindent\textbf{R2 (skewness changes).} With $\rho=0.4$, let $W$ be an independent
copy of $Z$. For $j\leq8$, define
\begin{align*}
B_{tj}=\mathbf 1\{W_{tj}<\Phi^{-1}(0.2)\},\qquad
Y_{tj}=\frac{B_{tj}-0.2}{\sqrt{0.2(1-0.2)}}+0.3Z_{tj},
\end{align*}
and set $X_{tj}=Y_{tj}$ for $t\leq\eta$ and $X_{tj}=-Y_{tj}$ for $t>\eta$.
For $j>8$, set $X_{tj}=Z_{tj}$. Reversing sign changes skewness but preserves
the mean vector and covariance matrix.

\noindent\textbf{R3 (sparse mean and variance changes).} With $\rho=0.4$, set
$X_{t1}=1.75Z_{t1}+0.3$ for $t>\eta$; all other observations satisfy
$X_{tj}=Z_{tj}$. Only one coordinate changes.

\medskip
\noindent\textbf{Case 2: multiple change points.}
Settings R4--R5 have $n=1200$ and two changes. We omit CPWZ, which targets
a single change point.

\noindent\textbf{R4 (unequal jumps).} Set $\rho=0.2$ and
$(\eta_1,\eta_2)=(400,800)$. For $j\leq10$, let
\begin{align*}
X_{tj}=
\begin{cases}
Z_{tj}, & t\leq\eta_1,\\
0.65+1.3Z_{tj}, & \eta_1<t\leq\eta_2,\\
1.55+1.8Z_{tj}, & t>\eta_2.
\end{cases}
\end{align*}
For $j>10$, set $X_{tj}=Z_{tj}$. Both mean and variance increase, with unequal
jump sizes at equally spaced locations.

\noindent\textbf{R5 (unequal spacings).} With $\rho=0.1$ and
$(\eta_1,\eta_2)=(450,950)$, let $H$ be an independent copy of $Z$.
For $j\leq30$, define
\begin{align*}
U_{tj}=\frac{2\cdot\mathbf 1\{H_{tj}<0\}-1+0.8Z_{tj}}{\sqrt{1.64}},
\qquad
X_{tj}=
\begin{cases}
0.35+2U_{tj}, & \eta_1<t\leq\eta_2,\\
U_{tj}, & \text{otherwise}.
\end{cases}
\end{align*}
For $j>30$, set $X_{tj}=Z_{tj}$. The Gaussian mixture undergoes a temporary
mean and scale shift, yielding equal jump sizes and unequal segment lengths.

\begin{table}[tbp]
\centering
\setlength{\abovecaptionskip}{\baselineskip}
\setlength{\belowcaptionskip}{\baselineskip}
\caption{Mean \hyperref[app:hausdorff-distance]{Hausdorff distance} (SD) and percentages
$[\widehat K<K,\widehat K=K,\widehat K>K]$ over 100 repetitions.
Tensor reports refined estimates; lowest means are bold. Dashes: not applicable.}
\label{tab:simulation-results}
\setlength{\tabcolsep}{2pt}
\renewcommand{\arraystretch}{1}
\newcommand{\simresult}[3]{\begin{tabular}[c]{@{}c@{}}#1 (#2)\\$[#3]$\end{tabular}}
\begin{tabular}{@{}lccccc@{}}
\toprule
\textbf{Setting} & \textbf{Moment Tensor} & \textbf{CPWZ}
& \textbf{changeAUC} & \textbf{Mean-change} & \textbf{MNSBS} \\
\midrule
\textbf{R1}
& \simresult{\textbf{227.81}}{180.77}{0,100,0}
& \simresult{304.43}{151.49}{0,100,0}
& \simresult{611.18}{152.48}{61,39,0}
& \simresult{510.53}{236.89}{0,100,0}
& \simresult{391.14}{219.32}{0,100,0} \\
\addlinespace[1pt]
\textbf{R2}
& \simresult{\textbf{7.15}}{8.27}{0,100,0}
& \simresult{343.62}{152.84}{0,100,0}
& \simresult{665.37}{114.75}{79,21,0}
& \simresult{549.07}{228.91}{0,100,0}
& \simresult{416.62}{231.35}{2,98,0} \\
\addlinespace[1pt]
\textbf{R3}
& \simresult{\textbf{10.95}}{19.33}{0,100,0}
& \simresult{49.95}{79.17}{0,100,0}
& \simresult{585.57}{181.78}{58,42,0}
& \simresult{356.99}{268.00}{0,100,0}
& \simresult{410.56}{228.78}{1,99,0} \\
\addlinespace[1pt]
\textbf{R4}
& \simresult{\textbf{1.82}}{2.25}{0,100,0}
& ---
& \simresult{17.09}{24.09}{0,99,1}
& \simresult{12.19}{56.40}{0,97,3}
& \simresult{295.60}{103.20}{5,11,84} \\
\addlinespace[1pt]
\textbf{R5}
& \simresult{\textbf{1.00}}{1.53}{0,100,0}
& ---
& \simresult{23.69}{33.24}{0,94,6}
& \simresult{57.53}{97.82}{3,65,32}
& \simresult{134.26}{154.40}{0,50,50} \\
\bottomrule
\end{tabular}
\end{table}

\Cref{tab:simulation-results} shows that Tensor recovers the correct number
of changes in all 100 repetitions of each setting, and its refined estimates
attain the lowest mean \hyperref[app:hausdorff-distance]{Hausdorff distance}. In R2 and R3, the mean distances
are 7.15 and 10.95, compared with 343.62 and 49.95 for CPWZ, the closest
competitor. R1 remains harder to localize, as
the moment signal is weaker relative to its variability: on each affected
coordinate, the variance and third-moment jumps are 0.125 and 1.625, whereas
the third-moment jump in R2 is 3. These differences make R1 harder to
localize even though Tensor detects the change in every repetition.
For the two-change settings R4 and R5, Tensor attains mean distances of 1.82
and 1.00; competing methods have larger localization errors and less reliable
estimates of the number of changes.
The simulation-calibrated 95\% intervals have empirical coverage of 95\% in
R4 and 96\% in R5 at each change point, conditional on $\widehat K=K=2$
(\Cref{tab:multiple-change-coverage}).

\begin{table}[tbp]
\centering
\caption{Empirical coverage of simulation-calibrated 95\% Tensor intervals,
conditional on $\widehat K=K=2$. Entries give covered/eligible repetitions.}
\label{tab:multiple-change-coverage}
\small
\renewcommand{\arraystretch}{1.25}
\begin{tabular}{@{}lcc@{}}
\toprule
\textbf{Setting} & \textbf{First change point}
& \textbf{Second change point} \\
\midrule
\textbf{R4} & 95/100 & 95/100 \\
\textbf{R5} & 96/100 & 96/100 \\
\bottomrule
\end{tabular}
\end{table}
\Cref{fig:r4-tensor-ci,fig:r5-tensor-ci} show the intervals from the first ten
repetitions.

\begin{figure}[H]
\centering
\includegraphics[width=0.75\textwidth]{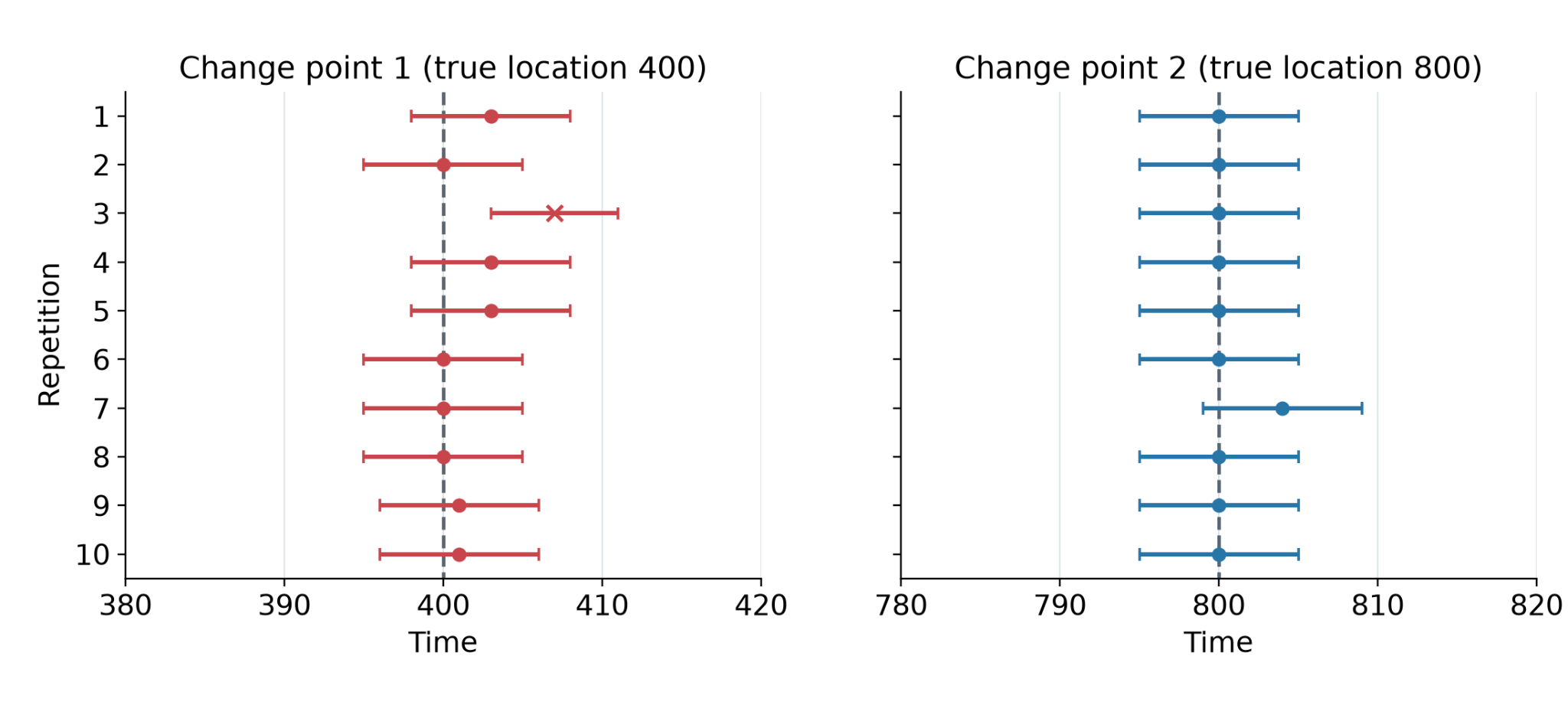}
\caption{Simulation-calibrated 95\% confidence intervals from the Tensor
method for the first ten repetitions in R4.  }
\label{fig:r4-tensor-ci}
\end{figure}

\begin{figure}[H]
\centering
\includegraphics[width=0.75\textwidth]{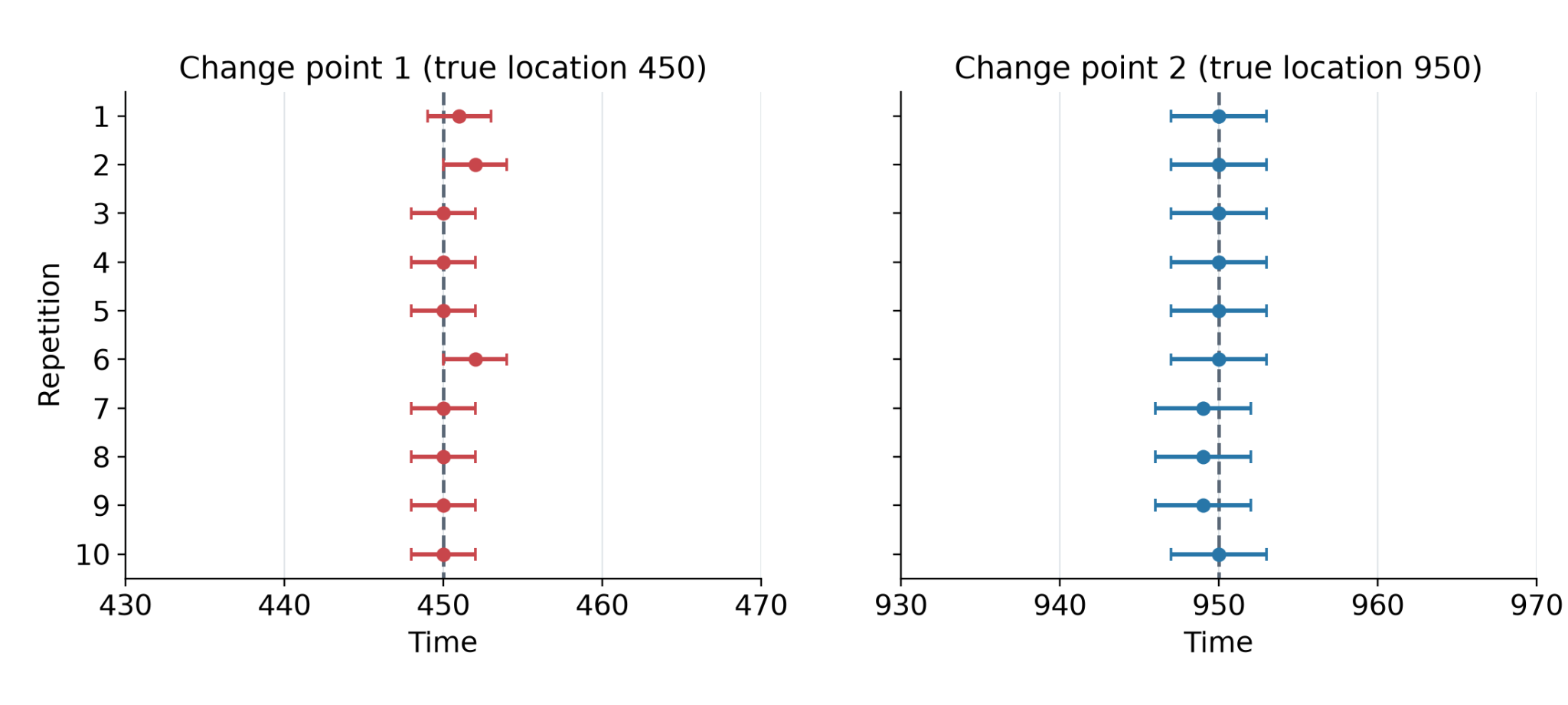}
\caption{Simulation-calibrated 95\% confidence intervals from the Tensor
method for the first ten repetitions in R5.  }
\label{fig:r5-tensor-ci}
\end{figure}

\subsection{Real data analysis}
\label{sec:nasdaq-real-data}

We analyze two panels of stocks that remained in the Nasdaq-100 during their
respective study periods: February 1, 2007 through December 31, 2010, and
January 2, 2018 through December 30,
2022.
The respective panels contain $n=988$ observations for $D=58$ stocks and
$n=1258$ observations for $D=59$ stocks. For stock $j$, we calculate
$r_{tj}=\log(P_{tj}/P_{t-1,j})$, where $P_{tj}$ is its closing
price on retained trading day $t$ \citep{QuantiacsFinancialData}. We analyze
$X_{tj}=(r_{tj}-\bar r_j)/s_j$, where $\bar r_j$ and $s_j$ are the mean and
standard deviation of stock $j$'s returns over that period. All methods use
the same standardized matrix $X\in\mathbb R^{n\times D}$ within each period.

We follow the joint $(p,\tau)$ tuning framework of \Cref{sec:tuning-parameters},
using the candidate grid for $p$ and $\tau$ and the
implementation's moment-kernel validation score. It selects $(p,\tau)=(2,16)$ in both periods. We also report the
preliminary and refined locations obtained with $p=3$ for comparison.
\Cref{tab:nasdaq-real-data} lists the Tensor locations, and
\Cref{fig:nasdaq-real-data} displays the daily cross-sectional RMS
$R_t=(D^{-1}\sum_{j=1}^D X_{tj}^2)^{1/2}$ and its average over the current and
up to 20 preceding trading days. Peaks of the smoothed curve in November
2008 and April 2020 indicate sustained elevated return magnitudes; a peak
may follow the initial shock and need not coincide with an estimated
change-point date. For $p=2$, refinement moves the earlier estimate from
September 25 to September
15, 2008, the date of the Lehman Brothers bankruptcy filing, a pivotal event
in the escalation of the financial crisis \citep{USTreasuryAboutTARP}. In the later period, the refined estimate is
February 21, 2020, in line with the onset of the market reversal associated
with the international spread of COVID-19 \citep{NasdaqFebruary2020}.

\begin{table}[tbp]
\centering
\caption{Preliminary and refined Tensor change-point locations for moment
orders $p=2$ and $p=3$ in the two Nasdaq-100 stock panels. Each configuration
returns one estimated change in each period.}
\label{tab:nasdaq-real-data}
\setlength{\tabcolsep}{5pt}
\begin{tabular}{@{}lcccc@{}}
\toprule
& \multicolumn{2}{c}{$p=2$} & \multicolumn{2}{c}{$p=3$} \\
\cmidrule(lr){2-3}\cmidrule(l){4-5}
Period & Preliminary & Refined & Preliminary & Refined \\
\midrule
2007--2010 & Sep. 25, 2008 & Sep. 15, 2008
& Oct. 13, 2008 & Oct. 10, 2008 \\
2018--2022 & Jan. 28, 2020 & Feb. 21, 2020
& Jan. 28, 2020 & Feb. 21, 2020 \\
\bottomrule
\end{tabular}
\end{table}

\begin{figure}[H]
\centering
\includegraphics[width=0.78\textwidth]{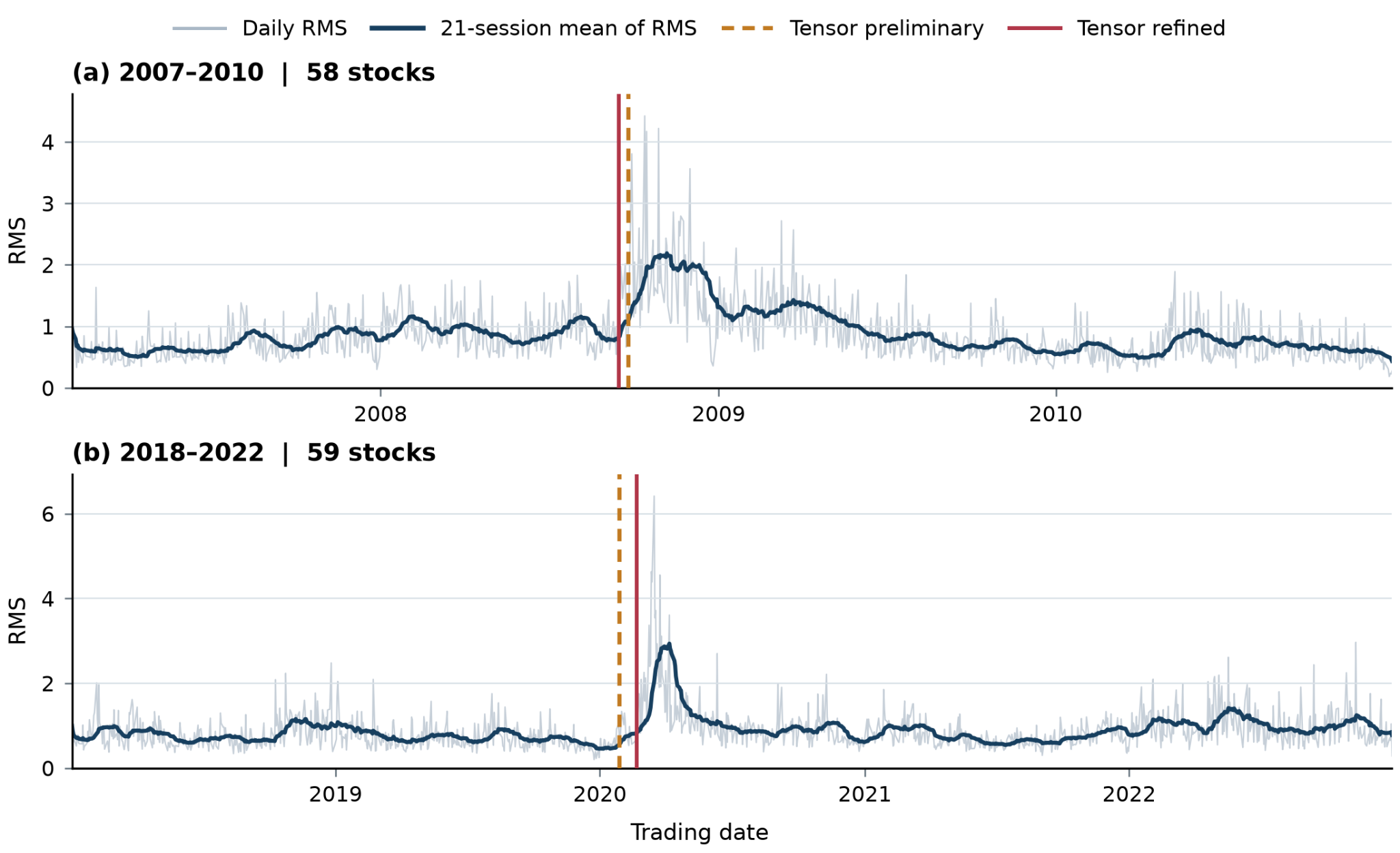}
\caption{Daily cross-sectional RMS $R_t$ (light curves) and its trailing
21-session average (dark curves) for the two Nasdaq-100 stock panels. Dashed
and solid vertical lines mark the $p=2$ preliminary and refined dates,
respectively.}
\label{fig:nasdaq-real-data}
\end{figure}

MNSBS estimates changes on July 17, 2007; August 29, 2008; August 18, 2009;
February 26, April 27, and December 22, 2010, in the first panel, and on
January 23, 2020 and September 27, 2021 in the second. In a recorded
two-change changeAUC run for each panel, its respective dates are November 13,
2008 and April 21, 2009; and February 4, 2020 and December 30, 2021.
Other recorded changeAUC runs return up to three dates in the later panel.
These alternative estimates support the broader finding of changing market
behavior during the 2008 crisis and the early COVID-19 disruption, although
their exact locations differ from the Tensor estimates.

\bibliographystyle{plainnat}
\bibliography{references}

\clearpage
\appendix

\section[Proofs for preliminary seeded detection]{Proofs for \Cref{sec:preliminary-seeded-detection}}

Recall the uniform CUSUM event from \eqref{eq:uniform-cusum-event}:
\begin{align*}
\mathcal E_{\mathrm{cusum}}
&=\left\{
\sup_{\substack{I=(s,e]\\1\leq h\leq\abs I/2}}
\max_{s+h\leq q\leq e-h}
\frac{\norm{\mathcal T_I(q)-\E\mathcal T_I(q)}_{\op}}
{\sqrt{D+\log n}
+(D+\log n)^{p/2+1/\gamma}/\sqrt h}
\leq C_2\sigma^p
\right\}.
\end{align*}
Here $I$ ranges over all nonempty index intervals in $(0,n]$.
The constant $C_2$ is from \Cref{thm:cusum}, which gives
$\Pp(\mathcal E_{\mathrm{cusum}})\geq1-n^{-5}$.

Use $\tau=c_\tau\sigma^p$ as in \Cref{thm:preliminary-localization}.
For $I\in\mathcal I$, denote the threshold in
\eqref{eq:preliminary-threshold} by
\begin{align*}
\iota_I&=\tau\left\{\sqrt{D+\log n}
+\frac{(D+\log n)^{p/2+1/\gamma}}{\sqrt{|I|}}\right\}.
\end{align*}
Set $C_0=2C_2$.
Throughout this appendix, take $c_\tau>C_0$ and $C_{\mathrm I}$
sufficiently large depending on $c_\tau$, as required by
\Cref{thm:preliminary-localization}.

\begin{lemma}[Detection on an isolating seeded interval]
\label{lem:detection-isolating-seeded-interval}
Suppose \Cref{ass:data-conditions,ass:preliminary-signal} hold for fixed
$p\geq2$. Take $\tau=c_\tau\sigma^p$ for sufficiently large $c_\tau$.
Let $I_k^*=(a_k^*,b_k^*]$ be an
isolating interval defined in \Cref{lem:isolating-seeded-interval}.
Then, on the event
\(\mathcal E_{\mathrm{cusum}}\) defined in
\eqref{eq:uniform-cusum-event}, \(G_{I_k^*}>\iota_{I_k^*}\),
where $G_I$ is the maximum of $\norm{\mathcal T_I(t)}_{\op}$ over the
central half of $I$, as defined in \eqref{eq:preliminary-interval-statistic}.
\end{lemma}
\begin{proof}[Proof of \Cref{lem:detection-isolating-seeded-interval}]
Choose $c_\tau>C_0=2C_2$, with $C_2$ from \Cref{thm:cusum}, and then
$C_{\mathrm I}$ sufficiently large depending on $c_\tau$.
Because \(I_k^*\) contains exactly one change, the CUSUM definition gives
\begin{align*}
\E\mathcal T_{I_k^*}(\eta_k)
&=
\sqrt{\frac{(\eta_k-a_k^*)(b_k^*-\eta_k)}{\abs{I_k^*}}}\,\Theta_k.
\end{align*}
By \eqref{eq:isolating-seeded-margins}, both
$\eta_k-a_k^*$ and $b_k^*-\eta_k$ are at least $\abs{I_k^*}/4$.
Taking operator norms and using $\norm{\Theta_k}_{\op}=\kappa_k$ gives
\begin{align}
\norm{\E\mathcal T_{I_k^*}(\eta_k)}_{\op}
&\geq
\frac14\sqrt{\abs{I_k^*}}\,\kappa_k
>
\frac{\sqrt{\Delta_{\min}}}{16\sqrt2}\,
\kappa_k.
\label{eq:isolating-population-signal}
\end{align}
The last inequality follows from
$\abs{I_k^*}>\Delta_{\min}/32$ in \eqref{eq:isolating-seeded-length}.
On $\mathcal E_{\mathrm{cusum}}$, \eqref{eq:uniform-cusum-event} gives
\begin{align*}
\max_{a_k^*+\abs{I_k^*}/4\leq q\leq b_k^*-\abs{I_k^*}/4}
\norm{\mathcal T_{I_k^*}(q)-\E\mathcal T_{I_k^*}(q)}_{\op}
&\leq
C_0\sigma^p
\left\{
\sqrt{D+\log n}
+\frac{(D+\log n)^{p/2+1/\gamma}}{\sqrt{\abs{I_k^*}}}
\right\}.
\end{align*}
\Cref{ass:preliminary-signal} therefore yields
\begin{align*}
&(c_\tau+C_0)\sigma^p
\left\{
\sqrt{D+\log n}
+\frac{(D+\log n)^{p/2+1/\gamma}}{\sqrt{\abs{I_k^*}}}
\right\}\\
&\qquad\leq
\left(c_\tau+C_0\right)
\left(1+4\sqrt2\right)C_{\mathrm I}^{-1/2}
\sqrt{\Delta_{\min}}\,\kappa_{\min}\\
&\qquad\leq
\left(c_\tau+C_0\right)
\left(1+4\sqrt2\right)C_{\mathrm I}^{-1/2}
\sqrt{\Delta_{\min}}\,\kappa_k
<
\frac{\sqrt{\Delta_{\min}}}{32\sqrt2}\,
\kappa_k.
\end{align*}
By \eqref{eq:isolating-seeded-margins},
\begin{align*}
a_k^*+\frac{\abs{I_k^*}}4
\leq\eta_k\leq b_k^*-\frac{\abs{I_k^*}}4,
\end{align*}
so $\eta_k$ is among the candidate splits defining $G_{I_k^*}$ in
\eqref{eq:preliminary-interval-statistic}. The definition of
$\iota_{I_k^*}$ and the preceding bound give
\begin{align*}
\left(1+\frac{C_0}{c_\tau}\right)\iota_{I_k^*}
&<\frac{\sqrt{\Delta_{\min}}}{32\sqrt2}\,\kappa_k
<\norm{\E\mathcal T_{I_k^*}(\eta_k)}_{\op},
\end{align*}
where the second inequality follows from
\eqref{eq:isolating-population-signal}.
On $\mathcal E_{\mathrm{cusum}}$, the reverse triangle inequality and
the deviation bound therefore yield
\begin{align*}
G_{I_k^*}
&\geq\norm{\mathcal T_{I_k^*}(\eta_k)}_{\op}\\
&\geq\norm{\E\mathcal T_{I_k^*}(\eta_k)}_{\op}
-\norm{\mathcal T_{I_k^*}(\eta_k)
-\E\mathcal T_{I_k^*}(\eta_k)}_{\op}\\
&\geq\norm{\E\mathcal T_{I_k^*}(\eta_k)}_{\op}
-\frac{C_0}{c_\tau}\iota_{I_k^*}\\
&>\left(1+\frac{C_0}{c_\tau}\right)\iota_{I_k^*}
-\frac{C_0}{c_\tau}\iota_{I_k^*}
=\iota_{I_k^*}.
\end{align*}
\end{proof}

\begin{proof}[Proof of \Cref{thm:preliminary-localization}]
\textbf{Step 1: uniform control and significant isolating intervals.}
By \Cref{thm:cusum},
$\Pp(\mathcal E_{\mathrm{unif}})\geq1-n^{-5}$ and
$\mathcal E_{\mathrm{unif}}\subseteq\mathcal E_{\mathrm{cusum}}$.
Work on $\mathcal E_{\mathrm{unif}}$. Applying
\eqref{eq:uniform-cusum-event} with $h=|I|/4$ gives, simultaneously over
$I=(s,e]\in\mathcal I$,
\begin{align*}
\max_{s+|I|/4\leq q\leq e-|I|/4}
\norm{\mathcal T_I(q)-\E\mathcal T_I(q)}_{\op}
&\leq\frac{C_0}{c_\tau}\iota_I<\iota_I.
\end{align*}
If $I$ contains no true change in its interior, its population CUSUM
vanishes, so $G_I<\iota_I$. Thus every significant interval contains
an interior change.

The augmented projection bound used in \Cref{lem:tensor-product-tail}
and H\"older's inequality give
\begin{align*}
\sup_{1\leq t\leq n}\norm{\mathcal M_t^{(p)}}_{\op}
&\leq(2p)^{p/2}\sigma^p,
&\kappa_{\min}&\leq2(2p)^{p/2}\sigma^p.
\end{align*}
Consequently, \Cref{ass:preliminary-signal} implies
\begin{align}
\Delta_{\min}
&\geq\frac{C_{\mathrm I}}{4(2p)^p}(D+\log n)
\geq256
\label{eq:stage-one-min-spacing}
\end{align}
for sufficiently large $C_{\mathrm I}$.
\Cref{lem:isolating-seeded-interval,lem:detection-isolating-seeded-interval}
therefore give fixed intervals $I_k^*\in\mathcal I$ satisfying
\begin{align}
G_{I_k^*}&>\iota_{I_k^*},\qquad k=1,\ldots,K.
\label{eq:all-isolating-intervals-significant}
\end{align}

\noindent\textbf{Step 2: selection and pruning.}
After $r$ selections, let $\mathcal I_{\mathrm{act}}^{(r)}$ and
$\widetilde{\mathcal E}^{(r)}$ be the active intervals and recorded
points. Initially,
\begin{align*}
\mathcal I_{\mathrm{act}}^{(0)}
&=\{I\in\mathcal I:G_I>\iota_I\},
&\widetilde{\mathcal E}^{(0)}&=\varnothing.
\end{align*}
The active family only shrinks, so every active interval remains
significant. We prove by induction that there is a set
$\mathcal K_r\subseteq\{1,\ldots,K\}$ of matched changes such that:
\begin{enumerate}[label=(\roman*)]
\item the $r$ distinct recorded points can be labeled as
$\widetilde{\mathcal E}^{(r)}=\{\check\eta_k:k\in\mathcal K_r\}$,
where $|\mathcal K_r|=r$ and
$|\check\eta_k-\eta_k|\leq3\Delta_{\min}/64$;
\item no active interval contains a matched change in its interior;
\item $I_k^*\in\mathcal I_{\mathrm{act}}^{(r)}$ for every
$k\notin\mathcal K_r$.
\end{enumerate}
These statements hold at $r=0$ by
\eqref{eq:all-isolating-intervals-significant}.

Suppose they hold for $r<K$. By (iii), the active family is nonempty.
Let $J=(s,e]$ be the selected shortest interval. Step~1 gives an
interior change $\eta_k$ in $J$, and (ii) implies $k\notin\mathcal K_r$.
Since $I_k^*$ is active by (iii),
\begin{align*}
|J|&\leq|I_k^*|\leq\Delta_{\min}/16.
\end{align*}
Thus $J$ contains only $\eta_k$. Every point in its central half,
including the recorded maximizer $\check\eta_k$, satisfies
\begin{align*}
|\check\eta_k-\eta_k|
&\leq3|J|/4\leq3\Delta_{\min}/64.
\end{align*}
The selected intervals are pairwise disjoint, because each selection
removes all overlapping intervals. Their recorded points are therefore
distinct, proving (i) for the new selection.

Pruning removes every active interval containing $\eta_k$ in its
interior, so (ii) is preserved. For any remaining unmatched change
$\eta_\ell$, an intersection $I_\ell^*\cap J\ne\varnothing$ would imply
\begin{align*}
\Delta_{\min}
&\leq|\eta_\ell-\eta_k|
\leq|I_\ell^*|+|J|
\leq\Delta_{\min}/8,
\end{align*}
a contradiction. Hence every such $I_\ell^*$ remains active, proving
(iii). Taking $\mathcal K_{r+1}=\mathcal K_r\cup\{k\}$ completes the
induction.

\noindent\textbf{Step 3: termination and ordering.}
Property (iii) prevents termination before $K$ selections. After $K$
selections, (ii) implies that any remaining active interval has no
interior change, contradicting Step~1. Thus the algorithm returns
exactly $K$ distinct points. Their true order follows from
\begin{align*}
\check\eta_{k+1}-\check\eta_k
&\geq\Delta_k-3\Delta_{\min}/32
\geq29\Delta_k/32>0,
\qquad k=1,\ldots,K-1.
\end{align*}
Sorting therefore gives $\widetilde\eta_k=\check\eta_k$, proving
\eqref{eq:preliminary-localization-bound} on $\mathcal E_{\mathrm{unif}}$.
Its probability bound from Step~1 completes the proof.
\end{proof}

\section{Proofs for Section~\ref{sec:change-point-local-refinement}}
\label{app:refinement-tools}

\subsection{Leading jump direction estimation}
\label{app:refinement-estimation}

Let $\mathcal E_{\mathrm{pre}}$ denote the preliminary-localization event
in \eqref{eq:preliminary-localization-bound}, established by
Theorem~\ref{thm:preliminary-localization} with probability at least
$1-n^{-5}$. Set
$\widetilde\eta_0=\eta_0=0$ and
$\widetilde\eta_{K+1}=\eta_{K+1}=n$ on this event.
The fitting endpoints $a_j,b_j$ and refinement endpoints $s_k,e_k$
are those of Section~\ref{sec:change-point-local-refinement}.

\begin{lemma}[Fitting intervals and refinement windows]
\label{lem:coarse-local-geometry}
Suppose $\mathcal E_{\mathrm{pre}}$ holds and $\Delta_{\min}\geq256$.
For $j=0,\ldots,K$, write
$\widetilde\Delta_j=\widetilde\eta_{j+1}-\widetilde\eta_j$. Then
\begin{align}
\frac{29}{32}\Delta_j
\leq\widetilde\Delta_j\leq\frac{35}{32}\Delta_j,
\label{eq:estimated-gap-bounds}
\end{align}
\begin{align}
(a_j,b_j] \subseteq(\eta_j,\eta_{j+1}] \quad \text{and} \quad
b_j-a_j \geq\frac7{16}\Delta_j.
\label{eq:refinement-fitting-length}
\end{align}
For $k=1,\ldots,K$, we have 
$\eta_{k-1}\leq s_k<\eta_k<e_k\leq\eta_{k+1}$ and
\begin{align}
\eta_k-s_k \geq\frac45\Delta_{k-1} \quad \text{and} \quad
e_k-\eta_k \geq\frac45\Delta_k.
\label{eq:refinement-window-margin}
\end{align}
Thus each refinement window contains exactly one true change point.
\end{lemma}

\begin{proof}[Proof of \Cref{lem:coarse-local-geometry}]
The preliminary errors at either endpoint of the $j$th segment satisfy
$\abs{\widetilde\eta_j-\eta_j},\allowbreak
\abs{\widetilde\eta_{j+1}-\eta_{j+1}}\leq3\Delta_j/64$. Hence
$\abs{\widetilde\Delta_j-\Delta_j}\leq3\Delta_j/32$, proving
\eqref{eq:estimated-gap-bounds}.

Recall the fitting endpoints in \eqref{eq:central-fitting-interval} are
\begin{align*}
a_j =\left\lceil\frac{3\widetilde\eta_j+\widetilde\eta_{j+1}}4\right\rceil-1 \quad \text{and} \quad b_j =\left\lfloor\frac{\widetilde\eta_j+3\widetilde\eta_{j+1}}4\right\rfloor.
\end{align*}
On $\mathcal E_{\mathrm{pre}}$, both endpoint errors are at most
$3\Delta_{\min}/64\leq3\Delta_j/64$. The triangle inequality gives
\begin{align*}
\left|\frac{3\widetilde\eta_j+\widetilde\eta_{j+1}}4
-\frac{3\eta_j+\eta_{j+1}}4\right|
&\leq\frac34|\widetilde\eta_j-\eta_j|
+\frac14|\widetilde\eta_{j+1}-\eta_{j+1}|
\leq\frac{3\Delta_j}{64},\\
\left|\frac{\widetilde\eta_j+3\widetilde\eta_{j+1}}4
-\frac{\eta_j+3\eta_{j+1}}4\right|
&\leq\frac14|\widetilde\eta_j-\eta_j|
+\frac34|\widetilde\eta_{j+1}-\eta_{j+1}|
\leq\frac{3\Delta_j}{64}.
\end{align*}
Since $\lceil x\rceil\geq x$ and
$\lfloor x\rfloor\leq x$, the first and last fitted indices satisfy
\begin{align*}
a_j+1-\eta_j
&\geq\frac{3\widetilde\eta_j+\widetilde\eta_{j+1}}4-\eta_j
\geq\frac{\Delta_j}4-\frac{3\Delta_j}{64}
=\frac{13\Delta_j}{64}>0,\\
\eta_{j+1}-b_j
&\geq\eta_{j+1}-\frac{\widetilde\eta_j+3\widetilde\eta_{j+1}}4
\geq\frac{13\Delta_j}{64}>0.
\end{align*}
Thus $(a_j,b_j]\subseteq(\eta_j,\eta_{j+1}]$.
The inequalities $\lceil x\rceil-1\leq x$ and
$\lfloor x\rfloor\geq x-1$ also give
\begin{align*}
b_j-a_j
&\geq\frac{\widetilde\Delta_j}2-1
\geq\frac{29\Delta_j}{64}-1
=\frac7{16}\Delta_j+\frac{\Delta_j}{64}-1
\geq\frac7{16}\Delta_j,
\end{align*}
where the last inequality uses $\Delta_j\geq256$. This proves
\eqref{eq:refinement-fitting-length}.

Recall the refinement endpoints, including integer rounding, are
\begin{align*}
s_k&=\left\lfloor\frac{9\widetilde\eta_{k-1}+\widetilde\eta_k}{10}\right\rfloor,
&e_k&=\left\lfloor\frac{\widetilde\eta_k+9\widetilde\eta_{k+1}}{10}\right\rfloor.
\end{align*}
The corresponding true-endpoint averages are
$\eta_{k-1}+\Delta_{k-1}/10$ and $\eta_k+9\Delta_k/10$.
Their preliminary versions differ by at most $3\Delta_{\min}/64$.
Using $x-1\leq\lfloor x\rfloor\leq x$ and
$\Delta_{\min}\leq\Delta_{k-1},\Delta_k$ therefore gives
\begin{align*}
\eta_k-s_k&\geq\frac{273}{320}\Delta_{k-1},
&s_k-\eta_{k-1}&\geq\frac{17}{320}\Delta_{k-1}-1,\\
e_k-\eta_k&\geq\frac{273}{320}\Delta_k-1,
&\eta_{k+1}-e_k&\geq\frac{17}{320}\Delta_k.
\end{align*}
For $\Delta\geq256$, we have
$17\Delta/320-1>0$ and $273\Delta/320-1\geq4\Delta/5$.
Applying these inequalities with $\Delta=\Delta_{k-1}$ and
$\Delta=\Delta_k$ proves both margins in
\eqref{eq:refinement-window-margin}. The outer distances also satisfy
$s_k-\eta_{k-1}>0$ and $\eta_{k+1}-e_k>0$. Consequently,
$\eta_{k-1}<s_k<\eta_k<e_k<\eta_{k+1}$, so the window contains
$\eta_k$ and no other true change point.
\end{proof}

For a sufficiently large fixed constant $C$, define
\begin{align}
\epsilon_n
&=C\sigma^p\left\{
\sqrt{\frac{D+\log n}{\Delta_{\min}}}
+\frac{(D+\log n)^{p/2+1/\gamma}}{\Delta_{\min}}
\right\}.
\label{eq:stage-ii-common-radius}
\end{align}
For each change $k$, on $\mathcal E_{\mathrm{pre}}$, use the adjacent
population and fitted tensors
\begin{align*}
\mathcal M_L^{(p)}&=\mathcal M_{\eta_k}^{(p)},
&\mathcal M_R^{(p)}&=\mathcal M_{\eta_k+1}^{(p)},\\
\widehat{\mathcal M}_L^{(p)}
&=\overline{\mathcal M}_{(a_{k-1},b_{k-1}]}^{(p)},
&\widehat{\mathcal M}_R^{(p)}
&=\overline{\mathcal M}_{(a_k,b_k]}^{(p)},
\end{align*}
where $k$ is suppressed in these four symbols.
Let $\mathcal E_{\mathrm{mom},k}$ be the event
\eqref{eq:stage-ii-event} for these tensors and radius $\epsilon_n$, and define
\begin{align}
\mathcal E_{\mathrm{ref},k}
&=\mathcal E_{\mathrm{mom},k}
\cap\{\eta_{k-1}\leq s_k<\eta_k<e_k\leq\eta_{k+1}\}.
\label{eq:refinement-event}
\end{align}
When a single change is fixed, suppress $k$ in this event.

\begin{lemma}[Simultaneous preliminary accuracy]
\label{lem:stage-ii-preliminary-event}
Suppose the conditions of \Cref{thm:preliminary-localization} hold, and
choose $C$ sufficiently large in \eqref{eq:stage-ii-common-radius}.
For all sufficiently large $n$, on $\mathcal E_{\mathrm{unif}}$ from
\eqref{eq:uniform-moment-event},
\begin{align}
\max_{0\leq j\leq K}
\norm{\overline{\mathcal M}_{(a_j,b_j]}^{(p)}
      -\mathcal M_{\eta_j+1}^{(p)}}_{\op}
&\leq\epsilon_n.
\end{align}
Moreover,
\begin{equation}
\begin{aligned}
\mathcal E_{\mathrm{unif}}
\subseteq\mathcal E_{\mathrm{pre}}
\cap\bigcap_{k=1}^K\mathcal E_{\mathrm{ref},k} \quad \text{and} \quad
\Pp(\mathcal E_{\mathrm{unif}}) \geq1-n^{-5}.
\end{aligned}
\label{eq:stage-ii-simultaneous-preliminary-event}
\end{equation}
\end{lemma}

\begin{proof}[Proof of \Cref{lem:stage-ii-preliminary-event}]
The projection bound in the proof of \Cref{lem:tensor-product-tail} gives
$\norm{u^\top Y_t}_{L^p}\leq\sqrt{2p}\,\sigma$ for every $t$ and
$u\in\Sph^D$. H\"older's inequality and the triangle inequality yield
\begin{align*}
\norm{\mathcal M_t^{(p)}}_{\op}
&\leq\sup_{u_1,\ldots,u_p\in\Sph^D}
\prod_{a=1}^p\norm{u_a^\top Y_t}_{L^p}
\leq(2p)^{p/2}\sigma^p,\\
\kappa_{\min}\leq\kappa_k
&\leq\norm{\mathcal M_{\eta_k}^{(p)}}_{\op}
+\norm{\mathcal M_{\eta_k+1}^{(p)}}_{\op}
\leq2(2p)^{p/2}\sigma^p.
\end{align*}
Combining this with \Cref{ass:preliminary-signal} gives
$\Delta_{\min}\geq c(D+\log n)\geq256$ for sufficiently large $n$.
On $\mathcal E_{\mathrm{unif}}$, the proof of
\Cref{thm:preliminary-localization} gives $\mathcal E_{\mathrm{pre}}$.
Hence \Cref{lem:coarse-local-geometry} places each fitting interval
$(a_j,b_j]$ inside $(\eta_j,\eta_{j+1}]$, with
$b_j-a_j\geq7\Delta_{\min}/16$.

The simultaneous interval bound \eqref{eq:uniform-moment-event} therefore
applies to these data-selected intervals and gives, for every
$j=0,\ldots,K$,
\begin{align*}
\norm{\overline{\mathcal M}_{(a_j,b_j]}^{(p)}
      -\mathcal M_{\eta_j+1}^{(p)}}_{\op}
\leq C_1\sigma^p\left\{
\sqrt{\frac{D+\log n}{b_j-a_j}}
+\frac{(D+\log n)^{p/2+1/\gamma}}{b_j-a_j}
\right\}
\leq\epsilon_n,
\end{align*}
where the last inequality uses the length bound and the choice of $C$.
The two adjacent moment bounds give $\mathcal E_{\mathrm{mom},k}$ for
every $k$. The window inclusions in \Cref{lem:coarse-local-geometry}
then give $\mathcal E_{\mathrm{ref},k}$. This proves the event inclusion
in \eqref{eq:stage-ii-simultaneous-preliminary-event}. The probability
bound follows from \Cref{thm:cusum}.
\end{proof}

\begin{lemma}[Direction accuracy for local refinement]
\label{lem:stage-iii-direction-rate}
Suppose \Cref{ass:local-refinement-separation,ass:local-refinement-signal}
hold, and use $\epsilon_n$ defined in  \eqref{eq:stage-ii-common-radius}. For all sufficiently large
$n$, on $\mathcal E_{\mathrm{ref},k}$,
\begin{align}
d_\pm(\widehat{\nu}_k,\nu_k)
\leq C_\nu\frac{\epsilon_n}{\kappa_k} \quad \text{with} \quad C_\nu&=\frac{2p}{c_0}.
\label{eq:stage-iii-direction-rate}
\end{align}
Furthermore,
\begin{align}
\frac{\epsilon_n\sqrt{D+1}}{\kappa_{\min}} \to 0.
\label{eq:refinement-improved-direction-rate}
\end{align}
\end{lemma}

\begin{proof}[Proof of \Cref{lem:stage-iii-direction-rate}]
Since $D+1\leq2D$, \Cref{ass:local-refinement-signal} gives
\begin{align*}
\frac{\epsilon_n\sqrt{D+1}}{\kappa_{\min}}
&\leq C\left\{
\sqrt{\frac{\sigma^{2p}D(D+\log n)}
{\Delta_{\min}\kappa_{\min}^2}}
+\frac{\sigma^p\sqrt D\,(D+\log n)^{p/2+1/\gamma}}
{\Delta_{\min}\kappa_{\min}}
\right\}\longrightarrow0.
\end{align*}
Since $\epsilon_n/\kappa_{\min}\to0$ and $p,c_0,r_0$ are fixed,
for all sufficiently large $n$,
\begin{align*}
\max_{1\leq k\leq K}\frac{\epsilon_n}{\kappa_k}
=\frac{\epsilon_n}{\kappa_{\min}}
\leq\min\left\{\frac18,\frac{c_0r_0^2}{8},\frac{c_0}{32p^2}\right\}.
\end{align*}
This verifies \eqref{eq:stage-ii-smallness} in
\Cref{thm:stage-ii-direction} for every change.
Apply that theorem to the adjacent population and fitted tensors defined
above, with maximizing directions $\nu_k$ and $\widehat\nu_k$.
Assumption~\ref{ass:local-refinement-separation} supplies the separation
condition, and $\mathcal E_{\mathrm{ref},k}\subseteq\mathcal E_{\mathrm{mom},k}$
gives \eqref{eq:stage-iii-direction-rate}.
\end{proof}

\subsection{Proof of the refinement guarantee}

\begin{proof}[Proof of \Cref{thm:local-refinement}]
\noindent\textbf{Step 1: Preliminary accuracy and notation.}
Since $D\geq1$, \Cref{ass:local-refinement-signal} gives
\begin{align*}
\frac{\Delta_{\min}\kappa_{\min}^2}
{\sigma^{2p}(D+\log n)}
&\geq\frac{\Delta_{\min}\kappa_{\min}^2}
{\sigma^{2p}D(D+\log n)}\to\infty,\\
\frac{\Delta_{\min}\kappa_{\min}}
{\sigma^p(D+\log n)^{p/2+1/\gamma}}
&\geq\frac{\Delta_{\min}\kappa_{\min}}
{\sigma^p\sqrt D\,(D+\log n)^{p/2+1/\gamma}}
\to\infty.
\end{align*}
Thus \Cref{ass:preliminary-signal} holds for all sufficiently large $n$.
By \Cref{lem:stage-ii-preliminary-event},
\begin{align*}
\Pp(\mathcal E_{\mathrm{unif}}^c)&\leq n^{-5},
&\mathcal E_{\mathrm{unif}}
&\subseteq\{\widetilde K=K\}\cap\bigcap_{k=1}^K\mathcal E_{\mathrm{ref},k},
\end{align*}
where $\mathcal E_{\mathrm{ref},k}$ is defined in \eqref{eq:refinement-event}.
All moment and probability bounds below are unconditional; the deterministic
comparisons are on $\mathcal E_{\mathrm{unif}}$.

Fix $k$ and suppress its index in
$\eta,\widehat\eta,\nu,\widehat{\nu},\Theta,\widehat\Theta,\kappa,s,e,\widehat Q$.
Use the population and fitted tensors
$\mathcal M_\ell^{(p)},\widehat{\mathcal M}_\ell^{(p)}$, $\ell\in\{L,R\}$,
from Section~\ref{app:refinement-estimation}.
On $\mathcal E_{\mathrm{unif}}$, \Cref{lem:stage-iii-direction-rate} gives,
after sign alignment,
\begin{align}
\norm{\widehat\nu-\nu}_2&\leq\rho_{n,k},
\qquad \rho_{n,k}=C_\nu\epsilon_n/\kappa_k.
\end{align}
By \eqref{eq:refinement-improved-direction-rate}, take $n$ sufficiently
large that $\rho_{n,k}\sqrt{D+1}\leq1$ and
$\epsilon_n\leq\kappa_k/64$ for every $k$.
For integers $r\geq0$, define the centered forward and backward sums
\begin{align*}
S_r^+&=\sum_{t=\eta+1}^{\eta+r}
\{Y_t^{\otimes p}-\E Y_t^{\otimes p}\},
&S_r^-&=\sum_{t=\eta-r+1}^{\eta}
\{Y_t^{\otimes p}-\E Y_t^{\otimes p}\}.
\end{align*}
Here $S_0^\pm=0$, and centered summands outside $\{1,\ldots,n\}$
are interpreted as zero. For $M\geq1$, define
$r_M=\lceil M\sigma^{2p}/\kappa^2\rceil$.

\medskip
\noindent\textbf{Step 2: Localization error implies a boundary crossing.}
On $\mathcal E_{\mathrm{unif}}$, every $\mathcal E_{\mathrm{ref},k}$ holds.
For $\ell\in\{L,R\}$, set
\begin{align*}
\mu_\ell&=\ip{\mathcal M_\ell^{(p)}}{\widehat{\nu}^{\otimes p}}_{\F},
&\widehat\mu_\ell
&=\ip{\widehat{\mathcal M}_\ell^{(p)}}{\widehat{\nu}^{\otimes p}}_{\F}.
\end{align*}
The moment-error bounds imply
$\abs{\widehat\mu_\ell-\mu_\ell}\leq\epsilon_n$ and
\begin{align*}
\norm{\widehat\Theta-\Theta}_{\op}
&\leq\norm{\widehat{\mathcal M}_L^{(p)}-\mathcal M_L^{(p)}}_{\op}
+\norm{\widehat{\mathcal M}_R^{(p)}-\mathcal M_R^{(p)}}_{\op}
\leq2\epsilon_n.
\end{align*}
Exact maximization at $\widehat\nu$ and $\nu$ gives
\begin{align*}
\abs{\mu_L-\mu_R}
&=\abs{\ip{\Theta}{\widehat\nu^{\otimes p}}_{\F}}\\
&\geq\abs{\ip{\widehat\Theta}{\widehat\nu^{\otimes p}}_{\F}}
-2\epsilon_n\\
&\geq\abs{\ip{\widehat\Theta}{\nu^{\otimes p}}_{\F}}
-2\epsilon_n\\
&\geq\abs{\ip{\Theta}{\nu^{\otimes p}}_{\F}}
-4\epsilon_n=\kappa-4\epsilon_n.
\end{align*}
The triangle inequalities then give
\begin{align*}
\abs{\widehat\mu_L-\mu_R}
&\geq\abs{\mu_L-\mu_R}-\abs{\widehat\mu_L-\mu_L}
\geq\kappa-5\epsilon_n,\\
\abs{\widehat\mu_L-\widehat\mu_R}
&\leq\norm{\widehat\Theta}_{\op}
\leq\norm{\Theta}_{\op}+\norm{\widehat\Theta-\Theta}_{\op}
\leq\kappa+2\epsilon_n.
\end{align*}
Recall the refinement criterion from
Section~\ref{sec:change-point-local-refinement}. With
$Z_i=(\widehat\nu^\top Y_i)^p$,
\begin{align*}
\widehat Q(t)
&=\sum_{i=s+1}^{t}(Z_i-\widehat\mu_L)^2
+\sum_{i=t+1}^{e}(Z_i-\widehat\mu_R)^2,
\qquad s<t\leq e-1.
\end{align*}
We detail the case $\widehat\eta-\eta\geq0$.
For an integer $r\geq1$ with $s<\eta+r\leq e-1$,
$\E Y_i^{\otimes p}=\mathcal M_R^{(p)}$ for $\eta<i\leq\eta+r$, so
\begin{align*}
\widehat Q(\eta+r)-\widehat Q(\eta)
&=\sum_{i=\eta+1}^{\eta+r}
\left\{(Z_i-\widehat\mu_L)^2-(Z_i-\widehat\mu_R)^2\right\}\\
&=r\left\{(\widehat\mu_L-\mu_R)^2-(\widehat\mu_R-\mu_R)^2\right\}
-2(\widehat\mu_L-\widehat\mu_R)
\sum_{i=\eta+1}^{\eta+r}(Z_i-\mu_R)\\
&=r\left\{(\widehat\mu_L-\mu_R)^2-(\widehat\mu_R-\mu_R)^2\right\}
-2(\widehat\mu_L-\widehat\mu_R)
\ip{S_r^+}{\widehat\nu^{\otimes p}}_{\F}.
\end{align*}
Suppose $\kappa^2(\widehat\eta-\eta)/\sigma^{2p}\geq M$ and set
$r=\widehat\eta-\eta$. Then $r\geq r_M\geq1$, since $r$ is an integer.
The minimizing property of $\widehat\eta$ and the preceding bounds imply
\begin{align*}
0&\geq\widehat Q(\widehat\eta)-\widehat Q(\eta)\\
&\geq r\{(\kappa-5\epsilon_n)^2-\epsilon_n^2\}
-2(\kappa+2\epsilon_n)
\abs{\ip{S_r^+}{\widehat\nu^{\otimes p}}_{\F}}.
\end{align*}
Consequently,
\begin{align*}
\abs{\ip{S_r^+}{\widehat\nu^{\otimes p}}_{\F}}
&\geq\frac{r\{(\kappa-5\epsilon_n)^2-\epsilon_n^2\}}
{2(\kappa+2\epsilon_n)}
>\frac{r\kappa}{4},
\end{align*}
where the last inequality uses $\epsilon_n\leq\kappa/64$.
For $\widehat\eta<\eta$, the same argument uses $r=\eta-\widehat\eta$,
interchanges $L$ and $R$, and replaces $S_r^+$ by $S_r^-$.
Taking $h=\widehat\nu-\nu$, with $\norm h_2\leq\rho_{n,k}$ on
$\mathcal E_{\mathrm{unif}}$, gives the event inclusion
\begin{align*}
&\mathcal E_{\mathrm{unif}}\cap
\left\{\frac{\kappa^2}{\sigma^{2p}}\abs{\widehat\eta-\eta}\geq M\right\}\\
&\quad\subseteq
\left\{\exists\,r\geq r_M:
\sup_{\norm h_2\leq\rho_{n,k}}
\abs{\ip{S_r^+}{(\nu+h)^{\otimes p}}_{\F}}>\frac{r\kappa}{4}\right\}\\
&\qquad\cup
\left\{\exists\,r\geq r_M:
\sup_{\norm h_2\leq\rho_{n,k}}
\abs{\ip{S_r^-}{(\nu+h)^{\otimes p}}_{\F}}>\frac{r\kappa}{4}\right\}.
\end{align*}

\medskip
\noindent\textbf{Step 3: Probability of a boundary crossing.}
For deterministic $\nu,\rho_{n,k}$, the triangle and Minkowski inequalities
and \eqref{eq:improved-fixed-maximal}--\eqref{eq:improved-local-maximal}
give, for either sign and every $m\geq1$,
\begin{align*}
\left\|\sup_{\norm h_2\leq\rho_{n,k}}\max_{1\leq r\leq m}
\abs{\ip{S_r^\pm}{(\nu+h)^{\otimes p}}_{\F}}\right\|_{L^4}
&\leq\left\|\max_{1\leq r\leq m}
\abs{\ip{S_r^\pm}{\nu^{\otimes p}}_{\F}}\right\|_{L^4}\\
&\quad+\left\|\sup_{\norm h_2\leq\rho_{n,k}}\max_{1\leq r\leq m}
\abs{\ip{S_r^\pm}{(\nu+h)^{\otimes p}-\nu^{\otimes p}}_{\F}}
\right\|_{L^4}\\
&\leq C\sigma^p\sqrt m
\left[1+\sum_{j=1}^p\{\rho_{n,k}\sqrt{D+1}\}^{j}\right]\\
&\leq C(1+p)\sigma^p\sqrt m.
\end{align*}
The last inequality uses $\rho_{n,k}\sqrt{D+1}\leq1$.
Since $p$ is fixed, taking fourth powers gives
\begin{align}
\E\left[
\sup_{\norm h_2\leq\rho_{n,k}}\max_{1\leq r\leq m}
\abs{\ip{S_r^\pm}{(\nu+h)^{\otimes p}}_{\F}}^4\right]
&\leq Cm^2\sigma^{4p}.
\label{eq:improved-full-projection-maximal}
\end{align}
Applying Markov's inequality to the fourth power in
\eqref{eq:improved-full-projection-maximal}, with $m=2^{q+1}r_M$, gives
for every integer $q\geq0$,
\begin{align*}
\Pp\left(
\sup_{\norm h_2\leq\rho_{n,k}}\max_{1\leq r\leq2^{q+1}r_M}
\abs{\ip{S_r^\pm}{(\nu+h)^{\otimes p}}_{\F}}
>\tfrac14 2^q r_M\kappa\right)
&\leq\frac{C\sigma^{4p}(2^{q+1}r_M)^2}{(2^q r_M\kappa/4)^4}\\
&\leq\frac{C\sigma^{4p}}{2^{2q}r_M^2\kappa^4}\\
&\leq C(2^qM)^{-2},
\end{align*}
where the last inequality uses $r_M\geq M\sigma^{2p}/\kappa^2$.

For an integer $r\geq r_M$ with
$\sup_{\norm h_2\leq\rho_{n,k}}
\abs{\ip{S_r^\pm}{(\nu+h)^{\otimes p}}_{\F}}>r\kappa/4$,
choose $q\geq0$ with $2^q r_M\leq r<2^{q+1}r_M$. Then
\begin{align*}
&\sup_{\norm h_2\leq\rho_{n,k}}\max_{1\leq u\leq2^{q+1}r_M}
\abs{\ip{S_u^\pm}{(\nu+h)^{\otimes p}}_{\F}}\\
&\quad\geq\sup_{\norm h_2\leq\rho_{n,k}}
\abs{\ip{S_r^\pm}{(\nu+h)^{\otimes p}}_{\F}}
>\frac{r\kappa}{4}\geq\frac{2^q r_M\kappa}{4}.
\end{align*}
Taking the union over $q$ and applying the preceding bound gives
\begin{align}
&\Pp\left(\exists\,r\geq r_M:
\sup_{\norm h_2\leq\rho_{n,k}}
\abs{\ip{S_r^\pm}{(\nu+h)^{\otimes p}}_{\F}}
>\tfrac14r\kappa\right)
\nonumber\\
&\quad\leq\sum_{q=0}^\infty\Pp\left(
\sup_{\norm h_2\leq\rho_{n,k}}\max_{1\leq r\leq2^{q+1}r_M}
\abs{\ip{S_r^\pm}{(\nu+h)^{\otimes p}}_{\F}}
>\tfrac14 2^q r_M\kappa\right)
\nonumber\\
&\quad\leq CM^{-2}\sum_{q=0}^\infty4^{-q}\leq CM^{-2}.
\label{eq:improved-linear-boundary}
\end{align}

\medskip
\noindent\textbf{Step 4: Final union bound.}
The two events in the Step~2 inclusion each have probability at most
$CM^{-2}$ by \eqref{eq:improved-linear-boundary}. Their union therefore gives,
for every $k$,
\begin{align*}
\Pp\left(\mathcal E_{\mathrm{unif}}\cap
\left\{\frac{\kappa_k^2}{\sigma^{2p}}\abs{\widehat\eta_k-\eta_k}\geq M\right\}\right)
&\leq2CM^{-2}.
\end{align*}
Absorb the factor $2$ into $C$.
Since $\widetilde K=K$ on $\mathcal E_{\mathrm{unif}}$, a union bound
over $k$ now gives, for all sufficiently large $n$ and every $M\geq1$,
\begin{align}
&\Pp\left(
\{\widetilde K\ne K\}\cup
\left\{\max_{1\leq k\leq K}
\frac{\kappa_k^2}{\sigma^{2p}}\abs{\widehat\eta_k-\eta_k}\geq M\right\}
\right)
\nonumber\\
&\qquad\leq\Pp(\mathcal E_{\mathrm{unif}}^c)
+\sum_{k=1}^K\Pp\left(\mathcal E_{\mathrm{unif}}\cap
\left\{\frac{\kappa_k^2}{\sigma^{2p}}
\abs{\widehat\eta_k-\eta_k}\geq M\right\}\right)
\nonumber\\
&\qquad\leq n^{-5}+CKM^{-2}.
\end{align}
As $K$ is fixed, this proves
\begin{align}
\max_{1\leq k\leq K}
\frac{\kappa_k^2}{\sigma^{2p}}\abs{\widehat\eta_k-\eta_k}
&=O_{\Pp}(1).
\end{align}
\end{proof}

\section[Proof for minimax lower bound]{Proof for \Cref{sec:minimax-lower-bound}}
\label{app:minimax-lower-bound}

\begin{proof}[Proof of \Cref{thm:minimax-localization-lower}]
We first construct two regime laws with identical moments below order $p$
and an order-$p$ jump of size $\kappa$. We then place this jump at two
locations in $[\Delta,n-\Delta]$ and bound the error of any estimator by
Le Cam's two-point method. All constants depend only on $p$.

\medskip
\noindent\textbf{Step 1: Construct the left and right regime laws.}
We take $X=(Z,0,\ldots,0)^\top\in\R^D$ in both regimes and denote
its laws before and after the change point by $F_L$ and $F_R$, respectively.
We specify these laws through two Gaussian mixtures for $Z$.
Set
\begin{align*}
a&=\frac{\sigma}{p+3},
&c_0&=\frac{p!}{2(p+3)^p},
&\theta&=\frac{\kappa}{2p!a^p}.
\end{align*}
Then $0<\kappa\leq c_0\sigma^p$ implies $0<\theta\leq1/4$.
For each component $j=0,\ldots,p$, define
\begin{align*}
\pi_j&=2^{-p}\binom pj,
&s_j&=(-1)^{p-j},\\
w_{L,j}&=\pi_j(1+\theta s_j),
&w_{R,j}&=\pi_j(1-\theta s_j).
\end{align*}
These are probability weights: since $\sum_{j=0}^p\pi_j=1$ and
$\sum_{j=0}^p\pi_js_j=2^{-p}(1-1)^p=0$,
\begin{align*}
\sum_{j=0}^p w_{\ell,j}&=1,
&w_{\ell,j}&\geq\frac34\pi_j>0,
\qquad \ell\in\{L,R\}.
\end{align*}
The scalar distributions defining $F_L$ and $F_R$ are
\begin{align*}
\text{left regime:}\quad Z&\sim
\sum_{j=0}^p w_{L,j}\,N\!\left(a(2j-p),a^2\right),\\
\text{right regime:}\quad Z&\sim
\sum_{j=0}^p w_{R,j}\,N\!\left(a(2j-p),a^2\right).
\end{align*}
Both regimes share the same $p+1$ Gaussian components; their weights differ
by the sign of the perturbation $\theta\pi_js_j$.

Equivalently, for $\ell\in\{L,R\}$, let $J\in\{0,\ldots,p\}$ be the
component index with \mbox{$\Pp(J=j)=w_{\ell,j}$}.
Draw $G\sim N(0,1)$ independently of $J$ and set $Z=a(2J-p+G)$.
The augmented vector is
$Y=(1,X^\top)^\top=(1,Z,0,\ldots,0)^\top\in\R^{D+1}$.

\medskip
\noindent\textbf{Step 2: Verify the tail and moment conditions.}
For either regime, the Gaussian even-moment formula and monotonicity of
$L^q$ norms give, for $q\geq1$ and $m=\lceil q/2\rceil$,
\begin{align*}
\norm G_{L^q}
&\leq\norm G_{L^{2m}}
=\{(2m-1)!!\}^{1/(2m)}
\leq\sqrt{2m}\leq2\sqrt q,
\end{align*}
where $(2m-1)!!\leq(2m)^m$ and $2m\leq q+2\leq4q$.
Since $|2J-p|\leq p$, Minkowski's inequality gives
\begin{align*}
q^{-1/2}\norm Z_{L^q}
&\leq a\left(\frac{p}{\sqrt q}+2\right)
\leq a(p+2)\leq\sigma.
\end{align*}
Thus $\norm Z_{\psi_2}\leq\sigma$, and under both $F_L$ and $F_R$,
\begin{align}
\sup_{v\in\Sph^{D-1}}\norm{v^\top X}_{\psi_2}
&=\sup_{v\in\Sph^{D-1}}|v_1|\norm Z_{\psi_2}
\leq\sigma.
\label{eq:lower-proof-subgaussian}
\end{align}

For the moment comparison, let $\E_L$ and $\E_R$ denote expectations
when $X\sim F_L$ and $X\sim F_R$, respectively; in particular, $Z$ has
the left and right scalar mixture distributions from Step 1.
Independence of $J$ and $G$ and the binomial theorem yield
\begin{align*}
\E_L e^{tZ}-\E_R e^{tZ}
&=e^{a^2t^2/2}\sum_{j=0}^p
(w_{L,j}-w_{R,j})e^{at(2j-p)}\\
&=2\theta e^{a^2t^2/2}\,2^{-p}e^{-apt}(e^{2at}-1)^p\\
&=2\theta e^{a^2t^2/2}\sinh(at)^p\\
&=2\theta a^pt^p+O(t^{p+2})\qquad(t\to0).
\end{align*}
Since finite Gaussian mixtures have analytic moment generating functions,
the coefficient of $t^r$ is $(\E_L Z^r-\E_R Z^r)/r!$.
Comparing coefficients therefore gives
\begin{align*}
\E_L Z^r-\E_R Z^r&=0\quad(0\leq r<p),
&\E_L Z^p-\E_R Z^p&=2p!a^p\theta=\kappa.
\end{align*}
Write $\mathcal M_\ell^{(p)}=\E_\ell Y^{\otimes p}$ for
$\ell\in\{L,R\}$. Since $Y_1=1$, $Y_2=Z$, and $Y_i=0$ for $i>2$,
each potentially nonzero tensor entry equals $\E_\ell Z^r$, where $r$
counts the indices equal to $2$. The entries therefore agree between
regimes except at $(2,\ldots,2)$, where the difference is $\kappa$.
Consequently,
\begin{align}
\Theta&=\mathcal M_L^{(p)}-\mathcal M_R^{(p)}
=\kappa e_2^{\otimes p},
&\norm\Theta_{\op}
&=\kappa\sup_{u_1,\ldots,u_p\in\Sph^D}
\prod_{j=1}^p|(u_j)_2|=\kappa.
\label{eq:lower-proof-jump}
\end{align}
Here $e_2=(0,1,0,\ldots,0)^\top\in\R^{D+1}$ is the second canonical
basis vector, corresponding to the coordinate $Z$ in $Y$.
The same argument shows that all moment tensors of order less than $p$ agree.

The maximizing direction is $\nu=e_2$, up to sign. For $u\in\Sph^D$,
\begin{align}
\kappa-\abs{\ip{\Theta}{u^{\otimes p}}_{\F}}
&=\kappa(1-|u_2|^p)
\geq\kappa(1-|u_2|)
=\frac{\kappa}{2}d_\pm(u,e_2)^2.
\label{eq:lower-proof-separation}
\end{align}
Thus \Cref{ass:local-refinement-separation} holds with curvature constant
$1/2$ and radius $r_0=1$.

\medskip
\noindent\textbf{Step 3: Bound the divergence between the regimes.}
Let $\phi_j$ be the density of $N(a(2j-p),a^2)$, and let
$f_\ell(z)=\sum_{j=0}^p w_{\ell,j}\phi_j(z)$ be the scalar density of
$Z$ in regime $\ell\in\{L,R\}$. For positive numbers $u_j,v_j$, the
log-sum inequality \citep[Theorem~2.7.1]{CoverThomas2006} states that
\begin{align*}
\left(\sum_{j=0}^p u_j\right)
\log\frac{\sum_{j=0}^p u_j}{\sum_{j=0}^p v_j}
&\leq\sum_{j=0}^p u_j\log\frac{u_j}{v_j}.
\end{align*}
It follows from convexity of $x\mapsto x\log x$, applied with weights
$v_j/\sum_{i=0}^p v_i$ to the values $u_j/v_j$.
Apply this inequality at each $z$ with
$u_j=w_{L,j}\phi_j(z)$ and $v_j=w_{R,j}\phi_j(z)$.
The component densities cancel inside the logarithm and integrate to one.
Also, $X=(Z,0,\ldots,0)^\top$ determines $Z$ through its first coordinate,
so embedding $Z$ in $X$ preserves KL divergence. Integration therefore gives
\begin{align*}
\operatorname{KL}(F_L,F_R)
&=\int_{\R}f_L(z)\log\frac{f_L(z)}{f_R(z)}\,dz\\
&\leq\sum_{j=0}^p w_{L,j}\log\frac{w_{L,j}}{w_{R,j}}\\
&=(1+\theta)\log\frac{1+\theta}{1-\theta}\sum_{j:s_j=1}\pi_j
 +(1-\theta)\log\frac{1-\theta}{1+\theta}\sum_{j:s_j=-1}\pi_j\\
&=\frac{1+\theta}{2}\log\frac{1+\theta}{1-\theta}
+\frac{1-\theta}{2}\log\frac{1-\theta}{1+\theta}\\
&=\theta\log\frac{1+\theta}{1-\theta}\leq4\theta^2.
\end{align*}
Here $\sum_{j:s_j=1}\pi_j=\sum_{j:s_j=-1}\pi_j=1/2$, because
$\sum_j\pi_j=1$ and $\sum_j\pi_js_j=0$.
The final inequality follows from
\begin{align*}
\log\frac{1+\theta}{1-\theta}
&=2\int_0^\theta\frac{du}{1-u^2}\leq4\theta.
\end{align*}
Interchanging $L$ and $R$ gives the same bound, so
\begin{align}
\max\{\operatorname{KL}(F_L,F_R),\operatorname{KL}(F_R,F_L)\}
&\leq4\theta^2
=\frac{(p+3)^{2p}}{(p!)^2}\frac{\kappa^2}{\sigma^{2p}}.
\label{eq:lower-proof-kl-scale}
\end{align}

\medskip
\noindent\textbf{Step 4: Construct the two change-point laws.}
The integer $h$ will be the distance, measured in observation indices,
between the two alternative true change-point locations. Set
\begin{align}
\beta_p&=\frac{(p!)^2}{2(p+3)^{2p}}=2c_0^2,
&h&=\left\lfloor\beta_p\frac{\sigma^{2p}}{\kappa^2}\right\rfloor.
\label{eq:lower-proof-shift}
\end{align}
Since $\kappa\leq c_0\sigma^p$ and $\beta_p=2c_0^2$, we have
$\beta_p\sigma^{2p}/\kappa^2\geq2$. Applying
$x/2\leq\lfloor x\rfloor\leq x$ for $x\geq2$ gives
\begin{align*}
\frac{\beta_p}{2}\frac{\sigma^{2p}}{\kappa^2}
&\leq h\leq\beta_p\frac{\sigma^{2p}}{\kappa^2}.
\end{align*}
The first condition in \Cref{ass:local-refinement-signal} gives
\begin{align*}
\frac{\Delta\kappa^2}{\sigma^{2p}}
&\geq\frac{\Delta\kappa^2}{\sigma^{2p}D(D+\log n)}
\longrightarrow\infty.
\end{align*}
Hence $h=o(\Delta)$, so $1\leq h<\Delta$ for all sufficiently large $n$.
Choose two possible true change-point locations,
\begin{align*}
\eta^{(0)}&=\Delta,
&\eta^{(1)}&=\Delta+h.
\end{align*}
For $i\in\{0,1\}$, let $P_{i,n}$ be the joint law of $n$ independent
observations with law $F_L$ up to $\eta^{(i)}$ and $F_R$ afterward:
\begin{align*}
P_{0,n}&=F_L^{\otimes\Delta}\otimes F_R^{\otimes(n-\Delta)},\\
P_{1,n}&=F_L^{\otimes(\Delta+h)}\otimes F_R^{\otimes(n-\Delta-h)}.
\end{align*}
Because $\Delta\leq n/4$,
\begin{align*}
\Delta\leq\eta^{(0)}<\eta^{(1)}<2\Delta\leq n-\Delta.
\end{align*}
Together with \eqref{eq:lower-proof-subgaussian} and
\eqref{eq:lower-proof-jump}, this proves
$P_{0,n},P_{1,n}\in\mathcal P(n,D,\Delta,\sigma,\kappa)$ for all
sufficiently large $n$.

Only indices $\Delta+1,\ldots,\Delta+h$ have different marginals:
$F_R$ under $P_{0,n}$ and $F_L$ under $P_{1,n}$.
By independence and \eqref{eq:lower-proof-kl-scale}--\eqref{eq:lower-proof-shift},
\begin{align}
\operatorname{KL}(P_{0,n},P_{1,n})
&=h\operatorname{KL}(F_R,F_L)\nonumber\\
&\leq h\frac{(p+3)^{2p}}{(p!)^2}\frac{\kappa^2}{\sigma^{2p}}
\leq\frac12.
\label{eq:lower-proof-sample-kl}
\end{align}
With $\operatorname{TV}(P,Q)=\sup_A|P(A)-Q(A)|$, Pinsker's inequality and
\eqref{eq:lower-proof-sample-kl} give
\begin{align*}
\operatorname{TV}(P_{0,n},P_{1,n})
&\leq\sqrt{\frac{\operatorname{KL}(P_{0,n},P_{1,n})}{2}}
\leq\frac12.
\end{align*}

\medskip
\noindent\textbf{Step 5: Convert testing error into localization risk.}
For any estimator $\check\eta$, define
\begin{align*}
\psi&=\mathbf 1\{\check\eta\geq\Delta+h/2\},
\end{align*}
where $\psi=i$ selects $P_{i,n}$. Since the change points are $h$ apart,
\begin{align*}
\{\psi=1\}&\subseteq\bigl\{\lvert\check\eta-\eta^{(0)}\rvert\geq h/2\bigr\},
&\{\psi=0\}&\subseteq\bigl\{\lvert\check\eta-\eta^{(1)}\rvert\geq h/2\bigr\}.
\end{align*}
By the definition of total variation,
\begin{align*}
P_{0,n}(\psi=1)+P_{1,n}(\psi=0)
&=1-\{P_{1,n}(\psi=1)-P_{0,n}(\psi=1)\}\\
&\geq1-\operatorname{TV}(P_{0,n},P_{1,n})\geq\frac12.
\end{align*}
Thus
\begin{align*}
\max_{i=0,1}P_{i,n}\bigl(\lvert\check\eta-\eta^{(i)}\rvert\geq h/2\bigr)
&\geq\frac12\bigl\{P_{0,n}(\psi=1)+P_{1,n}(\psi=0)\bigr\}
\geq\frac14.
\end{align*}
Using $\E W\geq t\Pp(W\geq t)$ for $W\geq0$, we obtain
\begin{align}
\sup_{P\in\mathcal P(n,D,\Delta,\sigma,\kappa)}
\E_P\abs{\check\eta-\eta}
&\geq\max_{i=0,1}\E_{P_{i,n}}\lvert\check\eta-\eta^{(i)}\rvert\nonumber\\
&\geq\frac h2\max_{i=0,1}
P_{i,n}\bigl(\lvert\check\eta-\eta^{(i)}\rvert\geq h/2\bigr)\nonumber\\
&\geq\frac h8
\geq\frac{\beta_p}{16}\frac{\sigma^{2p}}{\kappa^2}.
\label{eq:lower-proof-risk}
\end{align}
Taking the infimum over $\check\eta$ in \eqref{eq:lower-proof-risk} proves
\eqref{eq:lower-risk-bound}
for all sufficiently large $n$, with
\begin{align*}
c_1&=\frac{\beta_p}{16}=\frac{(p!)^2}{32(p+3)^{2p}}.
\end{align*}
Finally, the constructed laws satisfy \Cref{ass:data-conditions} with
$\gamma=\infty$, the separation condition by
\eqref{eq:lower-proof-separation}, and the signal conditions because their
minimum spacings are at least $\Delta$.
Since $h/2\geq(\beta_p/4)\sigma^{2p}/\kappa^2$, every estimator has
probability at least $1/4$ of an error at least
$(\beta_p/4)\sigma^{2p}/\kappa^2$ under at least one alternative.
This matches the localization scale in \Cref{thm:local-refinement}.
\end{proof}

\section{Proofs for Section~\ref{sec:procedure-inference}}
\label{app:inference-proofs}

Fix $k\in\{1,\ldots,K\}$ and write
$\eta=\eta_k$, $\kappa=\kappa_k$, and $(s,e]=(s_k,e_k]$.
The population and fitted directions remain $\nu_k$ and $\widehat\nu_k$.
Let $\mathcal M_L^{(p)}=\mathcal M_\eta^{(p)}$ and
$\mathcal M_R^{(p)}=\mathcal M_{\eta+1}^{(p)}$, with fitted counterparts
$\widehat{\mathcal M}_L^{(p)},\widehat{\mathcal M}_R^{(p)}$ from
Section~\ref{app:refinement-estimation}. Define
\begin{align*}
\mu_\ell&=\ip{\mathcal M_\ell^{(p)}}{\nu_k^{\otimes p}}_{\F},
\quad \ell\in\{L,R\},
&\delta&=\mu_L-\mu_R,
\end{align*}
so $|\delta|=\kappa$. Constants in this appendix may also depend on the
fixed $\sigma$. Recall the uniform empirical-moment event from
\eqref{eq:uniform-moment-event}:
\begin{align*}
\mathcal E_{\mathrm{unif}}
&=\left\{
\sup_{\abs I\geq1}
\frac{\norm{\overline{\mathcal M}_I^{(p)}-
\E\overline{\mathcal M}_I^{(p)}}_{\op}}
{\sqrt{(D+\log n)/\abs I}
+(D+\log n)^{p/2+1/\gamma}/\abs I}
\leq C_1\sigma^p
\right\}.
\end{align*}
Here $I$ ranges over all nonempty index intervals in $(0,n]$, and
$C_1$ is the constant from \Cref{thm:cusum}.
Its complement has probability at most $n^{-5}$ by
\Cref{lem:stage-ii-preliminary-event}.
All moment and probability bounds below use the original law;
comparisons involving fitted windows are on $\mathcal E_{\mathrm{unif}}$.

Let $\epsilon_n$ be the deterministic radius in
\eqref{eq:stage-ii-common-radius}, and set
$\rho_{n,k}=C_\nu\epsilon_n/\kappa_k$.
By \eqref{eq:refinement-improved-direction-rate} and
\Cref{ass:local-refinement-signal},
\begin{align}
\rho_{n,k}\sqrt{D+1}&\to0,
&\frac{\sigma^{2p}}{\kappa_k^2}&=o(\Delta_{\min}).
\label{eq:inference-deterministic-smallness}
\end{align}
Since $C_\nu>0$ is fixed, the first limit also gives
$\epsilon_n/\kappa_k=\rho_{n,k}/C_\nu
\leq\rho_{n,k}\sqrt{D+1}/C_\nu\to0$.
These statements hold uniformly over the fixed set of changes. The moment
bound $\kappa_k\leq C_p\sigma^p$ also implies $\Delta_{\min}\to\infty$.
On $\mathcal E_{\mathrm{unif}}$, choose the sign of $\widehat\nu_k$ so that
$\|\widehat\nu_k-\nu_k\|_2\leq\rho_{n,k}$.

\subsection{Uniform oracle approximation}
\label{app:inference-oracle}

For integer candidates $q$ with $s<q\leq e-1$, define
\begin{align*}
Q_k^*(q)
&=\sum_{\substack{t\in\mathbb Z\\s<t\leq q}}
\{(\nu_k^\top Y_t)^p-\mu_L\}^2
+\sum_{\substack{t\in\mathbb Z\\q<t\leq e}}
\{(\nu_k^\top Y_t)^p-\mu_R\}^2.
\end{align*}
Define the centered sums within the adjacent true regimes by
\begin{align*}
S_r^+&=\sum_{t=\eta+1}^{\eta+r}
\{Y_t^{\otimes p}-\E Y_t^{\otimes p}\},
\qquad 0\leq r\leq\eta_{k+1}-\eta,\\
S_r^-&=\sum_{t=\eta-r+1}^{\eta}
\{Y_t^{\otimes p}-\E Y_t^{\otimes p}\},
\qquad 0\leq r\leq\eta-\eta_{k-1},
\end{align*}
for integer $r$, with $S_0^\pm=0$.
For larger $r$, extend the centered summands by zero beyond the respective
adjacent regime.
Expansion at admissible candidates gives
\begin{align}
Q_k^*(\eta+r)-Q_k^*(\eta)
&=r\kappa^2-2\delta\ip{S_r^+}{\nu_k^{\otimes p}}_{\F},
\nonumber\\
Q_k^*(\eta-r)-Q_k^*(\eta)
&=r\kappa^2+2\delta\ip{S_r^-}{\nu_k^{\otimes p}}_{\F}.
\label{eq:inference-oracle-increments}
\end{align}
\begin{lemma}[Uniform oracle equivalence]
\label{lem:uniform-oracle-equivalence}
Consider the moment-change model \eqref{eq:piecewise-moment-model}, with $p\geq2$ and $K$ fixed. Suppose
\Cref{ass:data-conditions,ass:local-refinement-separation,ass:local-refinement-signal}
hold. Use the estimator and tuning specified in \Cref{thm:local-refinement}.
Then, for each $k\in\{1,\ldots,K\}$ and every fixed $M,z>0$,
\begin{align}
\Pp\Bigg(\mathcal E_{\mathrm{unif}}\cap
\Bigg\{\frac1{\sigma^{2p}}
\max_{\substack{r\in\mathbb Z\\|r|\leq M\sigma^{2p}/\kappa^2}}
\Big|\{\widehat Q_k(\eta+r)-\widehat Q_k(\eta)\}
-\{Q_k^*(\eta+r)-Q_k^*(\eta)\}\Big|>z\Bigg\}\Bigg)
&\to 0.
\label{eq:inference-oracle-equivalence}
\end{align}
Here $\mathcal E_{\mathrm{unif}}$ is the event defined in
\eqref{eq:uniform-moment-event}.
\end{lemma}

\begin{proof}
On $\mathcal E_{\mathrm{unif}}$, \eqref{eq:refinement-window-margin} and
\eqref{eq:inference-deterministic-smallness} give
$M\sigma^{2p}/\kappa^2+1<\min\{\eta-s,e-\eta\}$ for all sufficiently
large $n$. Thus every $\eta+r$ in the maximum in
\eqref{eq:inference-oracle-equivalence} lies in $(s,e-1]$, where both
criteria are defined.

\medskip
\noindent\textbf{Step 1: Projected jump error.}
On $\mathcal E_{\mathrm{unif}}$, write
$\widehat\delta=\ip{\widehat\Theta_k}{\widehat\nu_k^{\otimes p}}_{\F}$.
Since $\|\widehat\nu_k\|_2=1$, the tensor bounds in
\Cref{lem:stage-ii-preliminary-event} imply, for $\ell\in\{L,R\}$,
\begin{align*}
\left|\ip{\widehat{\mathcal M}_\ell^{(p)}-\mathcal M_\ell^{(p)}}
{\widehat\nu_k^{\otimes p}}_{\F}\right|
\leq\|\widehat{\mathcal M}_\ell^{(p)}-\mathcal M_\ell^{(p)}\|_{\op}
\leq\epsilon_n,\\
\|\widehat\Theta_k-\Theta_k\|_{\op}
\leq\|\widehat{\mathcal M}_L^{(p)}-\mathcal M_L^{(p)}\|_{\op}
+\|\widehat{\mathcal M}_R^{(p)}-\mathcal M_R^{(p)}\|_{\op}
\leq2\epsilon_n.
\end{align*}
For the signed jump error, add and subtract
$\ip{\Theta_k}{\widehat\nu_k^{\otimes p}}_{\F}$ and apply
\eqref{eq:stage-ii-lemma-lipschitz}:
\begin{align*}
|\widehat\delta-\delta|
&\leq\left|\ip{\widehat\Theta_k-\Theta_k}{\widehat\nu_k^{\otimes p}}_{\F}\right|
+\left|\ip{\Theta_k}{\widehat\nu_k^{\otimes p}-\nu_k^{\otimes p}}_{\F}\right|\\
&\leq\|\widehat\Theta_k-\Theta_k\|_{\op}
+p\|\Theta_k\|_{\op}\|\widehat\nu_k-\nu_k\|_2\\
&\leq2\epsilon_n+p\kappa\rho_{n,k}
=(2+pC_\nu)\epsilon_n,
\end{align*}
where $\|\Theta_k\|_{\op}=\kappa$ and
$\rho_{n,k}=C_\nu\epsilon_n/\kappa$.
Finally, the unit norm of $\widehat\nu_k$ and the operator-norm
triangle inequality give
\begin{align*}
|\widehat\delta|
\leq\|\widehat\Theta_k\|_{\op} \leq\|\Theta_k\|_{\op}+\|\widehat\Theta_k-\Theta_k\|_{\op}
\leq\kappa+2\epsilon_n.
\end{align*}
Since $\delta^2=\kappa^2$, it follows that
\begin{align*}
|\widehat\delta^2-\kappa^2|
&=|\widehat\delta-\delta|\,|\widehat\delta+\delta|
\leq C(\kappa\epsilon_n+\epsilon_n^2).
\end{align*}

\medskip
\noindent\textbf{Step 2: Uniform comparison.}
Expanding the fitted criterion and collecting the inner products gives,
for positive integer shifts within its search region,
\begin{align*}
\widehat Q_k(\eta+r)-\widehat Q_k(\eta)
&=r\widehat\delta^2+2\widehat\delta
\ip{r(\widehat{\mathcal M}_R^{(p)}-\mathcal M_R^{(p)})-S_r^+}
{\widehat\nu_k^{\otimes p}}_{\F},\\
\widehat Q_k(\eta-r)-\widehat Q_k(\eta)
&=r\widehat\delta^2-2\widehat\delta
\ip{r(\widehat{\mathcal M}_L^{(p)}-\mathcal M_L^{(p)})-S_r^-}
{\widehat\nu_k^{\otimes p}}_{\F}.
\end{align*}
Subtracting \eqref{eq:inference-oracle-increments} and using Step~1
yields, on $\mathcal E_{\mathrm{unif}}$ and for either sign,
\begin{align*}
\Big|\{\widehat Q_k(\eta\pm r)-\widehat Q_k(\eta)\}
-\{Q_k^*(\eta\pm r)-Q_k^*(\eta)\}\Big| &\leq Cr(\kappa\epsilon_n+\epsilon_n^2)
+C\epsilon_n\left|\ip{S_r^\pm}{\nu_k^{\otimes p}}_{\F}\right|\\
&\qquad+C(\kappa+\epsilon_n)
\left|\ip{S_r^\pm}{\widehat\nu_k^{\otimes p}-\nu_k^{\otimes p}}_{\F}\right|.
\end{align*}
The difference is zero at $r=0$.

Set $m=\lceil M\sigma^{2p}/\kappa^2\rceil$.
On $\mathcal E_{\mathrm{unif}}$, the fitted direction lies in the
deterministic ball $\|\widehat\nu_k-\nu_k\|_2\leq\rho_{n,k}$.
We therefore bound the last term uniformly over this ball before
applying \Cref{lem:local-direction-maximal}. Since the sums use
deterministic blocks, the resulting moment bounds are unconditional:
\begin{align*}
\left\|\max_{1\leq r\leq m}
\left|\ip{S_r^\pm}{\nu_k^{\otimes p}}_{\F}\right|\right\|_{L^4}
&\leq C\sigma^p\sqrt m,
\end{align*}
and
\begin{align*}
&\left\|\sup_{\|h\|_2\leq\rho_{n,k}}\max_{1\leq r\leq m}
\left|\ip{S_r^\pm}{(\nu_k+h)^{\otimes p}-\nu_k^{\otimes p}}_{\F}\right|
\right\|_{L^4} \leq C\sigma^p\sqrt m
\sum_{j=1}^p\{\rho_{n,k}\sqrt{D+1}\}^j.
\end{align*}
Since $\kappa\leq C_p\sigma^p$, we have
$m\leq C_M\sigma^{2p}/\kappa^2$.
Taking maxima over $|r|\leq M\sigma^{2p}/\kappa^2$ in the comparison
and applying these bounds, the resulting upper bound, divided by
$\sigma^{2p}$, has $L^4$ norm at most
\begin{align*}
C_M\left\{
\frac{\epsilon_n}{\kappa}
+\frac{\epsilon_n^2}{\kappa^2}
+\left(1+\frac{\epsilon_n}{\kappa}\right)
\sum_{j=1}^p\{\rho_{n,k}\sqrt{D+1}\}^j\right\}
\to 0
\end{align*}
by \eqref{eq:inference-deterministic-smallness}.
Markov's inequality proves \eqref{eq:inference-oracle-equivalence}.
\end{proof}

\subsection{Nonvanishing-jump limit}

\begin{proof}[Proof of \Cref{thm:nonvanishing-localization-limit}]
\noindent\textbf{Step 1: Projected moments and oracle increments.}
For every fixed $q\geq1$, the augmented projection bound in the proof of
\Cref{lem:tensor-product-tail} gives, with fixed $\sigma<\infty$,
\begin{align*}
\sup_{1\leq t\leq n}\|\nu_k^\top Y_t\|_{L^q}
&\leq C_q\sigma<\infty.
\end{align*}
In particular, the $p$th powers are uniformly integrable, even when $D$
grows. The projected finite-dimensional convergence in
\Cref{ass:nonvanishing-observation-limit} therefore yields
\begin{align*}
\mu_L&\to\E\xi_{k,0}^p,
&\mu_R&\to\E\xi_{k,1}^p,
&\delta&\to\rho_k,
\qquad |\rho_k|=\kappa_k^*.
\end{align*}
Because $\Delta_{\min}\to\infty$, every fixed local index eventually lies
in its corresponding true regime. The moment model then gives
$\E\xi_{k,t}^p=\E\xi_{k,0}^p$ for $t\leq0$ and
$\E\xi_{k,t}^p=\E\xi_{k,1}^p$ for $t\geq1$.
Write $\zeta_{k,t}=\xi_{k,t}^p-\E\xi_{k,t}^p$.

Define the finite-sample oracle contrast by
\begin{align*}
\mathcal L_{n,k}(r)
&=\begin{cases}
|r|\kappa^2+2\delta\ip{S_{|r|}^-}{\nu_k^{\otimes p}}_{\F},&r<0,\\[4pt]
0,&r=0,\\[4pt]
r\kappa^2-2\delta\ip{S_r^+}{\nu_k^{\otimes p}}_{\F},&r>0,
\end{cases}
\qquad r\in\mathbb Z.
\end{align*}
On $\mathcal E_{\mathrm{unif}}$, \eqref{eq:inference-oracle-increments}
gives $\mathcal L_{n,k}(r)=Q_k^*(\eta+r)-Q_k^*(\eta)$ whenever
$s<\eta+r\leq e-1$.

Fix an integer $H\geq1$ independently of $n$.
By \Cref{ass:nonvanishing-observation-limit} and the deterministic
mean and coefficient limits above, continuous mapping through
coordinatewise $p$th powers and finite sums gives jointly
\begin{align*}
\left(\ip{S_r^+}{\nu_k^{\otimes p}}_{\F},
\ip{S_r^-}{\nu_k^{\otimes p}}_{\F}\right)_{r=1}^{H}
&\xrightarrow{d}
\left(\sum_{t=1}^{r}\zeta_{k,t},
\sum_{t=-r+1}^{0}\zeta_{k,t}\right)_{r=1}^{H}
\end{align*}
and
\begin{align}
\bigl(\mathcal L_{n,k}(r)\bigr)_{|r|\leq H}
&\xrightarrow{d}\bigl(\mathcal L_k(r)\bigr)_{|r|\leq H}
\qquad\text{in }\R^{2H+1},
\label{eq:nonvanishing-oracle-finite-limit}
\end{align}
where $\mathcal L_k$ is defined in \eqref{eq:nonvanishing-random-walk}.

\medskip
\noindent\textbf{Step 2: A finite, unique limiting minimizer.}
Set $T_r^+=\sum_{t=1}^r\zeta_{k,t}$ and
$T_r^-=\sum_{t=-r+1}^{0}\zeta_{k,t}$ for integers $r\geq1$.
Fix an integer $m\geq1$. For all sufficiently large $n$, the
forward block $(Y_{\eta+1},\ldots,Y_{\eta+m})$ and the reverse block
$(Y_\eta,\ldots,Y_{\eta-m+1})$ lie in the adjacent true regimes.
Apply \eqref{eq:improved-fixed-maximal} to these deterministic blocks
with the deterministic unit direction $\nu_k$.
Step~1 also gives joint convergence of
$(\ip{S_r^\pm}{\nu_k^{\otimes p}}_{\F})_{r=1}^m$
to $(T_r^\pm)_{r=1}^m$.
For $A>0$, the function
$x\mapsto\min\{A,\max_{1\leq r\leq m}|x_r|^4\}$ is bounded and
continuous on $\R^m$, so
\begin{align*}
\E\min\left\{A,\max_{1\leq r\leq m}|T_r^\pm|^4\right\}
&=\lim_{n\to\infty}\E\min\left\{A,\max_{1\leq r\leq m}
\left|\ip{S_r^\pm}{\nu_k^{\otimes p}}_{\F}\right|^4\right\} \leq C\sigma^{4p}m^2.
\end{align*}
The last inequality follows by removing the truncation and taking
fourth powers in \eqref{eq:improved-fixed-maximal}.
Letting $A\uparrow\infty$ and applying monotone convergence yields
\begin{align*}
\E\max_{1\leq r\leq m}|T_r^\pm|^4
&\leq C\sigma^{4p}m^2.
\end{align*}
For each $b>0$, Markov's inequality implies
\begin{align*}
\sum_{j=0}^{\infty}
\Pp\left(\max_{1\leq r\leq2^j}|T_r^\pm|>b2^j\right)
&\leq C\sigma^{4p}b^{-4}\sum_{j=0}^{\infty}2^{-2j}<\infty.
\end{align*}
By the first Borel--Cantelli lemma, for each positive rational $b$
and either sign, almost surely
\begin{align*}
\max_{1\leq s\leq2^j}|T_s^\pm|
&\leq b2^j\qquad\text{for all sufficiently large }j.
\end{align*}
The choices of $b$ and sign are countable, so these statements hold
simultaneously on one event of probability one. On this event, fix
such a $b$ and, for each integer $r\geq1$, choose $j$ with
$2^{j-1}<r\leq2^j$. For all sufficiently large $r$,
\begin{align*}
\frac{|T_r^\pm|}{r}
&\leq\frac{\max_{1\leq s\leq2^j}|T_s^\pm|}{2^{j-1}}
\leq2b.
\end{align*}
Taking $\limsup_{r\to\infty}$ and then letting $b\downarrow0$ through
positive rational values proves $T_r^\pm/r\to0$ almost surely. Hence
\begin{align*}
\frac{\mathcal L_k(r)}{|r|}&\longrightarrow(\kappa_k^*)^2>0
\qquad\text{as }|r|\to\infty.
\end{align*}
The positive limit implies $\mathcal L_k(r)\to+\infty$ almost surely.
Since $\mathcal L_k(0)=0$, the set of integers with
$\mathcal L_k(r)\leq0$ is nonempty and finite. Minimizing over this set
therefore gives the global minimum.

To rule out ties and prove uniqueness of the minimizer, fix integers
$r<s$. Then
\begin{align*}
\mathcal L_k(s)-\mathcal L_k(r)
&=(|s|-|r|)(\kappa_k^*)^2
-2\rho_k\sum_{t=r+1}^{s}
\{\xi_{k,t}^p-\E\xi_{k,t}^p\}.
\end{align*}
Since $\rho_k\ne0$, the difference is a nonconstant polynomial with a
Lebesgue-null zero set. The joint-density condition in
\Cref{ass:nonvanishing-observation-limit} therefore rules out ties for
each $r<s$, and a countable union gives uniqueness. Write
$r_*=\operatorname*{arg\,min}_{r\in\mathbb Z}\mathcal L_k(r)$.

\medskip
\noindent\textbf{Step 3: Argmin convergence.}
Since $\Pp(\mathcal E_{\mathrm{unif}}^c)\to0$, restriction to
$\mathcal E_{\mathrm{unif}}$ preserves weak limits.
On this event, \eqref{eq:refinement-window-margin} and
$\Delta_{\min}\to\infty$ ensure that every fixed finite set of shifts
is eventually allowed. With $\sigma$ fixed and $\kappa\to\kappa_k^*>0$,
\Cref{lem:uniform-oracle-equivalence} and
\eqref{eq:nonvanishing-oracle-finite-limit} therefore give weak convergence of
\mbox{$r\mapsto\widehat Q_k(\eta+r)-\widehat Q_k(\eta)$} to $\mathcal L_k$
in the supremum norm on every compact subset of $\mathbb Z$, since
these subsets are finite.

The limit paths are continuous on the discrete space $\mathbb Z$, and
Step~2 gives a unique finite, hence tight, minimizer $r_*$.
The error $\widehat\eta_k-\eta_k$ exactly minimizes the centered criterion
over its candidate set and is tight by \Cref{thm:local-refinement}.
Extending the criterion outside this set by its minimum plus one
preserves both the minimizers and the compact limits above.
Applying the argmax theorem
\citep[Theorem~3.2.2]{vanDerVaartWellner1996} to the negative criterion gives
\begin{align*}
\widehat\eta_k-\eta_k
&\xrightarrow{d}r_*
=\operatorname*{arg\,min}_{r\in\mathbb Z}\mathcal L_k(r).
\end{align*}
\end{proof}
\subsection{Vanishing-jump limit}

Recall from \eqref{eq:vanishing-projected-noise} that the centered
projected noise is
\begin{align*}
Z_{k,t}
&=(\nu_k^\top Y_t)^p-\E\bigl\{(\nu_k^\top Y_t)^p\bigr\}
=\ip{Y_t^{\otimes p}-\E Y_t^{\otimes p}}{\nu_k^{\otimes p}}_{\F}.
\end{align*}
Write $\omega_L=\omega_{k,L}$ and $\omega_R=\omega_{k,R}$ for its
positive limiting long-run standard deviations before and after
$\eta_k$, respectively, as specified in
\Cref{ass:vanishing-partial-sum-variance}. We use the extension in
\Cref{ass:vanishing-projected-process} to define $Z_{k,t}$ beyond the
observed adjacent segments.

Recall also the signed jump in the projected $p$th moment:
\begin{align*}
\delta&=\mu_L-\mu_R
=\E\{(\nu_k^\top Y_\eta)^p\}
-\E\{(\nu_k^\top Y_{\eta+1})^p\},
\qquad |\delta|=\kappa.
\end{align*}

\begin{lemma}[Two-sided projected invariance principle]
\label{lem:inference-fclt}
Under the conditions of \Cref{thm:vanishing-localization-limit}, define
\begin{align*}
S_{L,n}(t)&=\kappa\operatorname{sgn}(\delta)
\sum_{j=1}^{\lfloor t/\kappa^2\rfloor}Z_{k,\eta-j+1},
& S_{R,n}(t)&=\kappa\operatorname{sgn}(\delta)
\sum_{j=1}^{\lfloor t/\kappa^2\rfloor}Z_{k,\eta+j},
\qquad t\geq0.
\end{align*}
For every fixed $T>0$,
\begin{align*}
(S_{L,n},S_{R,n})
&\xrightarrow{d}(\omega_LB_{k,L},\omega_RB_{k,R})
\qquad\text{in }\ell^\infty([0,T])^2,
\end{align*}
where $\ell^\infty([0,T])$ is the space of bounded functions equipped
with the supremum norm, and $B_{k,L},B_{k,R}$ are independent standard
Brownian motions.
\end{lemma}

\begin{proof}[Proof of \Cref{lem:inference-fclt}]
Put $N=\kappa^{-2}\to\infty$ and
\begin{align*}
\xi_{L,n,j}&=\operatorname{sgn}(\delta)Z_{k,\eta-j+1},
&\xi_{R,n,j}&=\operatorname{sgn}(\delta)Z_{k,\eta+j},
\qquad j\geq1,
\end{align*}
so $S_{\ell,n}(t)=N^{-1/2}\sum_{j=1}^{\lfloor Nt\rfloor}\xi_{\ell,n,j}$.
For fixed $T>0$, it suffices to prove joint finite-dimensional convergence
to $(\omega_LB_{k,L},\omega_RB_{k,R})$ and stochastic equicontinuity
of both components on $[0,T]$.

\medskip
\noindent\textbf{Step 1: Moment and variance bounds.}
The extensions in \Cref{ass:vanishing-projected-process} make each
row $(\xi_{\ell,n,j})_{j\geq1}$ centered, stationary, and
$\alpha$-mixing. With $\sigma$ fixed, \Cref{lem:tensor-product-tail}
gives $\sup_{\ell,n,j}\norm{\xi_{\ell,n,j}}_{L^q}\leq C_q$ for every
fixed $q\geq1$. Davydov's inequality \citep{Davydov1968} yields
\begin{align}
\abs{\operatorname{Cov}(\xi_{\ell,n,j},\xi_{\ell,n,j+h})}
&\leq C\alpha_X(h)^{1/2}\leq Ce^{-ch^\gamma},
\qquad h\geq1,\quad\gamma<\infty.
\label{eq:inference-covariance-bound}
\end{align}
This bound also holds across the change, with separation measured in
original time. Under independence, all covariances at positive lags
vanish. Moreover, \Cref{lem:mixing-fourth-maximal} gives
\begin{align}
\E\max_{0\leq r\leq m}
\abs{\sum_{j=1}^r\xi_{\ell,n,a+j}}^4
&\leq Cm^2,
\qquad a\geq0,\quad m\geq1.
\label{eq:inference-fourth-moment}
\end{align}
For $\ell\in\{L,R\}$, write
\begin{align*}
c_{\ell,n}(h)&=\operatorname{Cov}
(\xi_{\ell,n,1},\xi_{\ell,n,1+h}),\qquad h\geq0,
&v_{\ell,n}^2&=c_{\ell,n}(0)+2\sum_{h=1}^{\infty}c_{\ell,n}(h).
\end{align*}
The covariance bound gives $\sup_n\sum_{h\geq1}h|c_{\ell,n}(h)|<\infty$,
so stationarity implies
\begin{align}
\left|\frac1m\operatorname{Var}
\left(\sum_{j=1}^m\xi_{\ell,n,j}\right)-v_{\ell,n}^2\right|
&\leq\frac2m\sum_{h=1}^{m-1}h|c_{\ell,n}(h)|
+2\sum_{h=m}^{\infty}|c_{\ell,n}(h)|
\leq\frac Cm.
\label{eq:inference-stationary-variance}
\end{align}
Since multiplication by $\operatorname{sgn}(\delta)$ preserves variance,
\Cref{ass:vanishing-partial-sum-variance} with $m=\lfloor N\rfloor$
and \eqref{eq:inference-stationary-variance} give
\begin{align}
v_{L,n}^2&\longrightarrow\omega_L^2,
&v_{R,n}^2&\longrightarrow\omega_R^2.
\label{eq:inference-lrv}
\end{align}

\medskip
\noindent\textbf{Step 2: Joint finite-dimensional convergence.}
Fix $0=t_0<t_1<\cdots<t_J\leq T$ and real coefficients
$a_{\ell,j}$, not all zero. Set
\begin{align*}
W_n&=\sum_{\ell\in\{L,R\}}\sum_{j=1}^J a_{\ell,j}
\{S_{\ell,n}(t_j)-S_{\ell,n}(t_{j-1})\}.
\end{align*}
Two disjoint intervals in original time, possibly on opposite sides
of the change, have at most $h$ pairs of indices with separation $h$.
The covariance of their sums is therefore bounded by
\mbox{$C\sum_{h\geq1}he^{-ch}<\infty$}.
Together with \eqref{eq:inference-stationary-variance} and
\eqref{eq:inference-lrv}, this gives
\begin{align*}
\operatorname{Var}(W_n)
&\longrightarrow\sum_{\ell\in\{L,R\}}\omega_\ell^2
\sum_{j=1}^J a_{\ell,j}^2(t_j-t_{j-1})>0.
\end{align*}
For $\lfloor Nt_{j-1}\rfloor<i\leq\lfloor Nt_j\rfloor$, define the
weighted scores in original time order by
\begin{align*}
U_{n,1-i}&=a_{L,j}\xi_{L,n,i},
&U_{n,i}&=a_{R,j}\xi_{R,n,i},\\
\sqrt N W_n&=\sum_{r=1-\lfloor Nt_J\rfloor}^{\lfloor Nt_J\rfloor}U_{n,r},
&\E U_{n,r}&=0.
\end{align*}
The weights preserve the uniform $\alpha$-mixing and covariance
bounds in Step~1. Since $J$ is fixed, \eqref{eq:inference-fourth-moment}
and Minkowski's inequality give
\begin{align*}
\left\|\sum_{r=a+1}^{a+b}U_{n,r}\right\|_{L^4}
&\leq C\sqrt b,
\qquad\text{uniformly over admissible blocks.}
\end{align*}
These bounds and the positive variance limit verify the
$\alpha$-mixing triangular-array CLT of
\citet[Corollary~2.1]{Withers1981}.
Hence
\begin{align*}
\frac{W_n}{\sqrt{\operatorname{Var}(W_n)}}
&\xrightarrow{d}\mathcal N(0,1).
\end{align*}
The variance limit and the Cram\'er--Wold device give the joint Gaussian
finite-dimensional limit $(\omega_LB_{k,L},\omega_RB_{k,R})$.
Its left and right components satisfy
\begin{align*}
\operatorname{Cov}(\omega_LB_{k,L}(s),\omega_RB_{k,R}(t))
&=0,\qquad s,t\in[0,T].
\end{align*}
Joint Gaussianity therefore implies that $B_{k,L}$ and $B_{k,R}$
are independent.

\medskip
\noindent\textbf{Step 3: Stochastic equicontinuity.}
For every $\varepsilon>0$, we verify
\begin{align}
\lim_{h\downarrow0}\limsup_{n\to\infty}
\Pp\left(\max_{\ell\in\{L,R\}}
\sup_{\substack{s,t\in[0,T]\\|t-s|\leq h}}
|S_{\ell,n}(t)-S_{\ell,n}(s)|>\varepsilon\right)&=0.
\label{eq:inference-equicontinuity}
\end{align}
Fix $0<h\leq T$. If $0\leq s\leq t\leq T$ and $t-s\leq h$, then
$j=\lfloor s/h\rfloor$ gives
$s,t\in[jh,(j+2)h]\cap[0,T]$.
Within this window, each increment is the difference of two prefix
sums starting at index $\lfloor Njh\rfloor+1$. Since
$\lfloor Nt\rfloor-\lfloor Njh\rfloor\leq\lceil2Nh\rceil+1$,
\begin{align*}
\sup_{s,t\in[jh,(j+2)h]\cap[0,T]}
|S_{\ell,n}(t)-S_{\ell,n}(s)|
&\leq\frac2{\sqrt N}
\max_{0\leq r\leq\lceil2Nh\rceil+1}
\left|\sum_{i=1}^{r}\xi_{\ell,n,\lfloor Njh\rfloor+i}\right|.
\end{align*}
A union bound over $\ell\in\{L,R\}$ and
$j=0,\ldots,\lfloor T/h\rfloor$, followed by Markov's inequality
and \eqref{eq:inference-fourth-moment}, now gives
\begin{align*}
&\Pp\left(\max_{\ell\in\{L,R\}}
\sup_{\substack{s,t\in[0,T]\\|t-s|\leq h}}
|S_{\ell,n}(t)-S_{\ell,n}(s)|>\varepsilon\right)\\
&\quad\leq\sum_{\ell\in\{L,R\}}\sum_{j=0}^{\lfloor T/h\rfloor}
\Pp\left(\max_{0\leq r\leq\lceil2Nh\rceil+1}
\left|\sum_{i=1}^{r}\xi_{\ell,n,\lfloor Njh\rfloor+i}\right|
>\frac{\varepsilon\sqrt N}{2}\right)\\
&\quad\leq\frac{C}{\varepsilon^4N^2}
\left(\left\lfloor\frac Th\right\rfloor+1\right)
(\lceil2Nh\rceil+1)^2
\leq C\varepsilon^{-4}(T/h+1)(h+N^{-1})^2.
\end{align*}
Taking $\limsup_{n\to\infty}$ and then $h\downarrow0$ proves
\eqref{eq:inference-equicontinuity}, since
\begin{align*}
\lim_{h\downarrow0}C\varepsilon^{-4}(T/h+1)h^2
&=\lim_{h\downarrow0}C\varepsilon^{-4}(Th+h^2)=0.
\end{align*}

For each $n$, the pair $(S_{L,n},S_{R,n})$ is a continuous function
of a finite random vector into its deterministic step-function
subspace, hence is Borel measurable in $\ell^\infty([0,T])^2$.
Because $S_{L,n}(0)=S_{R,n}(0)=0$ and $[0,T]$ is compact,
\eqref{eq:inference-equicontinuity} gives asymptotic tightness, with
every subsequential limit supported on $C([0,T])^2$.
Step~2 identifies its finite-dimensional distributions as those of
$(\omega_LB_{k,L},\omega_RB_{k,R})$. Continuous functions are
determined by their values at rational times, so every subsequential
limit has this same law. This proves the claimed weak convergence
in the supremum norm.
\end{proof}

\begin{proof}[Proof of \Cref{thm:vanishing-localization-limit}]
Since $\Pp(\mathcal E_{\mathrm{unif}}^c)\to0$, restriction to
$\mathcal E_{\mathrm{unif}}$ preserves weak limits. We work on this event.

\medskip
\noindent\textbf{Step 1: Convergence on compact intervals.}
Denote by $r_n(u)$ the rounded shift
$\operatorname{sgn}(u)\lfloor |u|/\kappa^2\rfloor$.
The domain of $G_n$ is
$\{u\in\R:s<\eta+r_n(u)\leq e-1\}$, corresponding to the integer
candidates in \eqref{eq:local-refined-estimate}. On this set, define
\begin{align*}
G_n(u)&=\widehat Q_k\{\eta+r_n(u)\}-\widehat Q_k(\eta).
\end{align*}
With $\sigma$ fixed, \eqref{eq:inference-deterministic-smallness} gives
$\kappa^{-2}=o(\Delta_{\min})$. Hence, for each fixed $T>0$,
\eqref{eq:refinement-window-margin} implies, eventually,
\begin{align*}
\sup_{|u|\leq T}|r_n(u)|+1
&\leq T/\kappa^2+1
<\frac45\Delta_{\min}
\leq\min\{\eta-s,e-\eta\}.
\end{align*}
Thus $s<\eta+r_n(u)\leq e-1$ for every $|u|\leq T$, so $G_n$ is
defined throughout $[-T,T]$. For these candidates,
\eqref{eq:inference-oracle-increments} gives
\begin{align*}
Q_k^*\{\eta+r_n(u)\}-Q_k^*(\eta)
&=\kappa^2|r_n(u)|+
\begin{cases}
2S_{L,n}(-u),&u\leq0,\\
-2S_{R,n}(u),&u>0.
\end{cases}
\end{align*}
Since $0\leq x-\lfloor x\rfloor<1$ for $x\geq0$,
\begin{align*}
0\leq |u|-\kappa^2|r_n(u)|
&=\kappa^2\left(\frac{|u|}{\kappa^2}
-\left\lfloor\frac{|u|}{\kappa^2}\right\rfloor\right)
<\kappa^2\longrightarrow0.
\end{align*}
Thus the drift converges uniformly to $|u|$, and
\Cref{lem:inference-fclt} gives the oracle limit from
\eqref{eq:vanishing-brownian-process}, namely
\begin{align*}
\mathcal G_k(u)&=|u|+2\begin{cases}
\omega_LB_{k,L}(-u),&u\leq0,\\
\omega_RB_{k,R}(u),&u>0,
\end{cases}
\end{align*}
where $B_{k,L}$ and $B_{k,R}$ are independent standard Brownian
motions. The right-hand minus sign in the oracle increment is
absorbed into $B_{k,R}$ without changing this joint law.
Choose a fixed $M\geq T/\sigma^{2p}$. Since $\sigma$ is fixed,
\Cref{lem:uniform-oracle-equivalence} gives
\begin{align*}
&\sup_{|u|\leq T}
\left|G_n(u)-\bigl[Q_k^*\{\eta+r_n(u)\}-Q_k^*(\eta)\bigr]\right|\\
&\quad\leq\max_{\substack{r\in\mathbb Z\\|r|\leq M\sigma^{2p}/\kappa^2}}
\left|\bigl[\widehat Q_k(\eta+r)-\widehat Q_k(\eta)\bigr]
-\bigl[Q_k^*(\eta+r)-Q_k^*(\eta)\bigr]\right|
=o_{\Pp}(1).
\end{align*}
Slutsky's theorem therefore yields convergence in the supremum norm:
\begin{align*}
G_n&\xrightarrow{d}\mathcal G_k
\qquad\text{in }\ell^\infty([-T,T]).
\end{align*}

\medskip
\noindent\textbf{Step 2: A unique finite limiting minimizer.}
For $\ell\in\{L,R\}$, $\varepsilon>0$, and $j\geq0$, Doob's
maximal inequality applied to $B_{k,\ell}^2$ gives
\begin{align*}
\Pp\left(\sup_{2^j\leq t\leq2^{j+1}}
\frac{|B_{k,\ell}(t)|}{t}>\varepsilon\right)
&\leq\frac{\E B_{k,\ell}(2^{j+1})^2}{\varepsilon^2 2^{2j}}
=\frac{2^{1-j}}{\varepsilon^2}.
\end{align*}
The bound is summable in $j$. Borel--Cantelli, applied to both sides
and all positive rational $\varepsilon$, therefore yields
$B_{k,\ell}(t)/t\to0$ almost surely as $t\to\infty$.
Since $\omega_L,\omega_R<\infty$, it follows that
\begin{align*}
\frac{\mathcal G_k(u)}{|u|}&\longrightarrow1
\qquad\text{almost surely as }|u|\to\infty.
\end{align*}
Thus $\mathcal G_k(u)\geq|u|/2$ for all sufficiently large $|u|$,
so $\mathcal G_k(u)\to+\infty$. Because $\mathcal G_k(0)=0$ and its
paths are continuous, its global minimum is attained on a compact
interval. For $u<v$,
\begin{align*}
\operatorname{Var}\{\mathcal G_k(v)-\mathcal G_k(u)\}
&=\begin{cases}
4\omega_L^2(v-u),&u<v\leq0,\\
4\omega_R^2(v-u),&0\leq u<v,\\
4\{\omega_L^2|u|+\omega_R^2v\},&u<0<v,
\end{cases}
\end{align*}
which is strictly positive. Lemma~2.6 of \citet{KimPollard1990},
applied to the continuous Gaussian process $-\mathcal G_k$ on $\R$,
therefore gives almost sure uniqueness.

\medskip
\noindent\textbf{Step 3: Argmin convergence.}
The error $\kappa^2(\widehat\eta_k-\eta_k)$ exactly minimizes $G_n$
over its domain, since $r_n(\kappa^2r)=r$ for every integer $r$,
and is tight by \Cref{thm:local-refinement} with $\sigma$ fixed.
Extending $G_n$ outside its domain by its minimum plus one preserves
both the minimizers and the compact limits in Step~1.
The limit paths are continuous, and Step~2 gives a unique finite,
hence tight, minimizer. Applying the argmax theorem
\citep[Theorem~3.2.2]{vanDerVaartWellner1996} to the negative criterion gives
\begin{align*}
\kappa_k^2(\widehat\eta_k-\eta_k)
&\xrightarrow{d}\operatorname*{arg\,min}_{u\in\R}\mathcal G_k(u).
\end{align*}
\end{proof}
\subsection{Variance estimation and confidence intervals}

We first prove consistency of the Bartlett estimators
\citep{NeweyWest1987} defined in \eqref{eq:ci-left-long-run-variance}
and \eqref{eq:ci-right-long-run-variance}.
The proof controls both the selected interval and the
estimated direction uniformly, because they are computed from the same
observations as the variance estimates.

\begin{lemma}[Consistency of the projected long-run variances]
\label{lem:inference-hac}
Under the conditions of \Cref{thm:vanishing-localization-limit},
\begin{align*}
\widehat\omega_{k,L}^2&\xrightarrow{\Pp}\omega_{k,L}^2,
&\widehat\omega_{k,R}^2&\xrightarrow{\Pp}\omega_{k,R}^2,
\end{align*}
where the estimators are defined in
\eqref{eq:ci-left-long-run-variance} and
\eqref{eq:ci-right-long-run-variance}.
\end{lemma}

\begin{proof}[Proof of \Cref{lem:inference-hac}]
Fix $k$. We prove consistency of $\widehat\omega_{k,L}^2$;
the proof for $\widehat\omega_{k,R}^2$ is identical.
Relabel the left segment $(\eta_{k-1},\eta_k]$ as $\{1,\ldots,N\}$,
where $N=\Delta_{k-1}$. For deterministic $v\in\R^{D+1}$,
write $z(v)=(z_t(v))_{t=1}^N$, where
\begin{align*}
z_t(v)&=(v^\top Y_t)^p-\E\{(v^\top Y_t)^p\},
&Z_k&=z(\nu_k),
&\gamma_{n,h}&=\operatorname{Cov}(Z_{k,0},Z_{k,h}),
\end{align*}
using the stationary regime law to define $\gamma_{n,h}$.
Throughout, $I=(a,c]\subseteq\{1,\ldots,N\}$ has $|I|\geq7N/16$,
and $\bar z_I(v)=|I|^{-1}\sum_{t\in I}z_t(v)$.
On $\mathcal E_{\mathrm{unif}}$, the relabelled left fitting interval is
\begin{align*}
\widehat I&=(a_{k-1}-\eta_{k-1},\,b_{k-1}-\eta_{k-1}]
\subseteq\{1,\ldots,N\},
&|\widehat I|&\geq7N/16,
\end{align*}
by \eqref{eq:refinement-fitting-length}.
For $x\in\R^N$, extended by zero, define
\begin{align*}
H_N(x)&=\frac1{N(\ell_n+1)}
\sum_{r=1}^{N+\ell_n}\left(\sum_{t=r-\ell_n}^{r}x_t\right)^2.
\end{align*}
Then $H_N^{1/2}$ is a seminorm, and
\begin{align*}
\frac N{|I|}H_N(x\mathbf1_I)
&=\frac1{|I|}\sum_{\substack{s,t\in I\\|s-t|\leq\ell_n}}
\left(1-\frac{|s-t|}{\ell_n+1}\right)x_sx_t.
\end{align*}
Since stationarity gives
\begin{align*}
z_t(v)-\bar z_I(v)
&=(v^\top Y_t)^p-|I|^{-1}\sum_{s\in I}(v^\top Y_s)^p,
\end{align*}
the estimator in \eqref{eq:ci-left-long-run-variance} is exactly
\begin{align*}
\widehat\omega_{k,L}^2
&=\frac N{|\widehat I|}
H_N\bigl([z(\widehat\nu_k)-\bar z_{\widehat I}(\widehat\nu_k)]
\mathbf1_{\widehat I}\bigr)
\qquad\text{on }\mathcal E_{\mathrm{unif}}.
\end{align*}
Its nonnegative square root is $\widehat\omega_{k,L}$.
Step~1 proves consistency with the population direction and centering;
Step~2 bounds the effect of the fitted direction and empirical centering
in this identity.

\noindent\textbf{Step 1: Consistency for the true projected noise.}
The signal condition and $\kappa_{\min}\leq C\sigma^p$ give
$\ell_n\to\infty$ and $\ell_n^2/N\to0$.
For the oracle form $H_N(Z_k\mathbf1_I)$, define the squared
full-window sums
\begin{align*}
V_r&=\frac1{\ell_n+1}
\left(\sum_{t=r-\ell_n}^{r}Z_{k,t}\right)^2,
\qquad \ell_n+1\leq r\leq N.
\end{align*}
Integrating \eqref{eq:dependent-subweibull-sum-tail} at order $16$
gives $\sup_r\|V_r\|_{L^8}\leq C$.
Within each residue class modulo $\ell_n+1$, the windows are disjoint
and their lag-$h$ separation is at least $h$.
Applying \Cref{lem:mixing-fourth-maximal} to each class and summing gives
\begin{align}
\left\|\max_{\ell_n+1\leq s\leq N}
\left|\sum_{r=\ell_n+1}^{s}(V_r-\E V_r)\right|\right\|_{L^4}
&\leq C\sqrt{N(\ell_n+1)}.
\label{eq:inference-window-maximal}
\end{align}
Only the at most $2\ell_n$ windows crossing the endpoints of $I$
separate $H_N(Z_k\mathbf1_I)$ from its full-window sum. Hence
\Cref{lem:mixing-fourth-maximal} also gives
\begin{align}
&\left\|\sup_I\left|H_N(Z_k\mathbf1_I)
-\frac1N\sum_{r=a+\ell_n+1}^{c}V_r\right|\right\|_{L^2}\notag\\
&\quad\leq\frac2N\left\|
\max_{1\leq r\leq N+\ell_n}
\max_{\substack{J\subseteq[r-\ell_n,r]\cap[1,N]\\J\text{ an interval}}}
\left|\sum_{t\in J}Z_{k,t}\right|^2\right\|_{L^2}
\leq C\frac{\ell_n+1}{\sqrt N}.
\label{eq:inference-window-boundary}
\end{align}
Here an interval sum is a difference of two prefix sums, and
$\|\max_r A_r\|_{L^2}\leq(\sum_r\|A_r\|_{L^2}^2)^{1/2}$.
The covariance bound in \eqref{eq:inference-covariance-bound} yields
\begin{align*}
\E V_r
&=\sum_{|h|\leq\ell_n}
\left(1-\frac{|h|}{\ell_n+1}\right)\gamma_{n,h}
=v_{L,n}^2+O((\ell_n+1)^{-1}),
&v_{L,n}^2&\longrightarrow\omega_{k,L}^2,
\end{align*}
where \eqref{eq:inference-lrv} identifies the covariance sum.
Since the full-window sum contains $|I|-\ell_n$ terms,
\eqref{eq:inference-window-maximal}--\eqref{eq:inference-window-boundary}
imply
\begin{align}
\left\|\sup_I
\left|\frac N{|I|}H_N(Z_k\mathbf1_I)-v_{L,n}^2\right|\right\|_{L^2}
&\leq C\left\{\frac1{\ell_n+1}
+\frac{\ell_n+1}{\sqrt N}\right\}=o(1).
\label{eq:inference-bartlett-uniform}
\end{align}
Markov's inequality proves uniform oracle consistency.

\noindent\textbf{Step 2: Estimated direction and sample centering.}
Apply \Cref{lem:local-direction-maximal} within each window and on the
whole segment, respectively, to obtain
\begin{align*}
\E\sup_{I,\,\|h\|_2\leq\rho_{n,k}}
H_N\bigl([z(\nu_k+h)-z(\nu_k)]\mathbf1_I\bigr)
&\leq C\left\{\sum_{j=1}^p
(\rho_{n,k}\sqrt{D+1})^j\right\}^2=o(1),\\
\left\|\sup_{I,\,\|h\|_2\leq\rho_{n,k}}
|\bar z_I(\nu_k+h)|\right\|_{L^4}
&\leq\frac C{\sqrt N}(1+\rho_{n,k}\sqrt{D+1})^p.
\end{align*}
Write
$x_{I,h}=[z(\nu_k+h)-\bar z_I(\nu_k+h)]\mathbf1_I$;
all suprema over $(I,h)$ below use the intervals above and
$\|h\|_2\leq\rho_{n,k}$.
Expanding the squares in $H_N(\mathbf1_I)$, a pair $s,t\in I$
appears in $\ell_n+1-|s-t|$ windows when $|s-t|\leq\ell_n$,
and in none otherwise. Therefore,
\begin{align*}
H_N(\mathbf1_I)
&=\frac1N\sum_{s\in I}
\sum_{\substack{t\in I\\|t-s|\leq\ell_n}}
\left(1-\frac{|t-s|}{\ell_n+1}\right)\\
&\leq\frac{|I|}{N}
\left\{1+2\sum_{j=1}^{\ell_n}
\left(1-\frac{j}{\ell_n+1}\right)\right\}
=\frac{|I|}{N}(\ell_n+1)\leq\ell_n+1.
\end{align*}
For each fixed $s$, there is one term with $t=s$ and at most two
terms with $|t-s|=j\geq1$; adding any missing nonnegative weights
gives the first inequality. The equality uses
$\sum_{j=1}^{\ell_n}(1-j/(\ell_n+1))=\ell_n/2$,
and the last inequality uses $|I|\leq N$.
The difference between the centered vectors is
\begin{align*}
x_{I,h}-Z_k\mathbf1_I
&=[z(\nu_k+h)-Z_k]\mathbf1_I
-\bar z_I(\nu_k+h)\mathbf1_I.
\end{align*}
Apply the triangle inequality for $H_N^{1/2}$ and
$(a+b)^2\leq2a^2+2b^2$. Using $\|\cdot\|_{L^2}\leq\|\cdot\|_{L^4}$
for the sample-mean supremum, the two moment bounds above give
\begin{align*}
&\E\sup_{I,h}H_N(x_{I,h}-Z_k\mathbf1_I)\\
&\quad\leq2\E\sup_{I,h}
H_N([z(\nu_k+h)-Z_k]\mathbf1_I)
+2(\ell_n+1)
\left\|\sup_{I,h}|\bar z_I(\nu_k+h)|\right\|_{L^4}^{2}\\
&\quad\leq C\left\{\sum_{j=1}^p(\rho_{n,k}\sqrt{D+1})^j\right\}^{2}
+C\frac{\ell_n+1}{N}(1+\rho_{n,k}\sqrt{D+1})^{2p}
=o(1).
\end{align*}
The last equality uses $\rho_{n,k}\sqrt{D+1}\to0$ and
$\ell_n/N\to0$. Markov's inequality therefore yields, for every
$\varepsilon>0$,
\begin{align*}
\Pp\left\{\sup_{I,h}H_N(x_{I,h}-Z_k\mathbf1_I)^{1/2}
>\varepsilon\right\}
&\leq\varepsilon^{-2}\E\sup_{I,h}
H_N(x_{I,h}-Z_k\mathbf1_I)\longrightarrow0.
\end{align*}

By \eqref{eq:inference-bartlett-uniform} and
$v_{L,n}^2\to\omega_{k,L}^2$,
$\sup_I H_N(Z_k\mathbf1_I)=O_{\Pp}(1)$.
The reverse triangle inequality and
\mbox{$|a^2-b^2|=|a-b|\,|a+b|$} give, for any $x,y\in\R^N$,
\begin{align*}
|H_N(x)-H_N(y)|
&\leq\{2H_N(y)^{1/2}+H_N(x-y)^{1/2}\}
H_N(x-y)^{1/2}.
\end{align*}
Taking $x=x_{I,h}$ and $y=Z_k\mathbf1_I$ in this bound gives
\begin{align*}
&\sup_{I,h}|H_N(x_{I,h})-H_N(Z_k\mathbf1_I)|\\
&\quad\leq\left\{2\sup_I H_N(Z_k\mathbf1_I)^{1/2}
+\sup_{I,h}H_N(x_{I,h}-Z_k\mathbf1_I)^{1/2}\right\}\\
&\qquad\times\sup_{I,h}H_N(x_{I,h}-Z_k\mathbf1_I)^{1/2}
=\{O_{\Pp}(1)+o_{\Pp}(1)\}\,o_{\Pp}(1)=o_{\Pp}(1).
\end{align*}
Finally, $N/|I|\leq16/7$ and \eqref{eq:inference-bartlett-uniform} imply
\begin{align*}
&\sup_{I,h}\left|\frac N{|I|}H_N(x_{I,h})-v_{L,n}^2\right|\\
&\quad\leq\frac{16}{7}\sup_{I,h}
|H_N(x_{I,h})-H_N(Z_k\mathbf1_I)|
+\sup_I\left|\frac N{|I|}H_N(Z_k\mathbf1_I)-v_{L,n}^2\right|
=o_{\Pp}(1).
\end{align*}

A sign change of $\widehat\nu_k$ leaves the Bartlett estimate unchanged.
On $\mathcal E_{\mathrm{unif}}$, its sign-aligned version satisfies
$\|\widehat\nu_k-\nu_k\|_2\leq\rho_{n,k}$ by
\Cref{lem:stage-iii-direction-rate}.
Evaluating the uniform bounds at $I=\widehat I$ and using the estimator
identity above therefore gives
\begin{align*}
\widehat\omega_{k,L}^2-v_{L,n}^2&\xrightarrow{\Pp}0,
&v_{L,n}^2&\longrightarrow\omega_{k,L}^2,
\end{align*}
since $\Pp(\mathcal E_{\mathrm{unif}}^c)\to0$.
\end{proof}

\begin{proof}[Proof of \Cref{thm:ci-coverage}]
\noindent\textbf{Step 1: Consistent scales.}
Fix $k$ and write
$\widehat{\boldsymbol{\vartheta}}=(\widehat\omega_{k,L},\widehat\omega_{k,R})$.
On $\mathcal E_{\mathrm{unif}}$, \eqref{eq:estimated-leading-direction}
gives, for all sufficiently large $n$,
\begin{align*}
|\widehat\kappa_k-\kappa_k|
&\leq\|\widehat\Theta_k-\Theta_k\|_{\op}
\leq2\epsilon_n=o(\kappa_k),
&\widehat\kappa_k&\geq\kappa_k/2>0.
\end{align*}
Together with \Cref{lem:inference-hac} and
$\Pp(\mathcal E_{\mathrm{unif}}^c)\leq n^{-5}$, this yields
\begin{align*}
\widehat\kappa_k/\kappa_k&\xrightarrow{\Pp}1,
&\widehat{\boldsymbol{\vartheta}}
&\xrightarrow{\Pp}(\omega_{k,L},\omega_{k,R})\in(0,\infty)^2.
\end{align*}
The corresponding failure probability satisfies
\begin{align*}
&\Pp\bigl((\mathcal E_{\mathrm{unif}}\cap
\{\widehat{\boldsymbol{\vartheta}}\in(0,\infty)^2\})^c\bigr)\\
&\quad\leq\Pp(\mathcal E_{\mathrm{unif}}^c)
+\sum_{\ell\in\{L,R\}}\Pp\left(\mathcal E_{\mathrm{unif}}\cap
\left\{|\widehat\omega_{k,\ell}-\omega_{k,\ell}|
\geq\omega_{k,\ell}/2\right\}\right)=o(1).
\end{align*}
For fitted quantities below, restrict to
$\mathcal E_{\mathrm{unif}}\cap\{\widehat{\boldsymbol{\vartheta}}\in(0,\infty)^2\}$;
its probability tends to one, so the weak limits are unchanged.

\medskip
\noindent\textbf{Step 2: Continuity of the Brownian quantiles.}
Use one pair of independent standard Brownian motions $B_{k,L},B_{k,R}$
to define, for $\boldsymbol{\vartheta}=(\vartheta_L,\vartheta_R)\in(0,\infty)^2$,
\begin{align*}
\mathcal H_{\boldsymbol{\vartheta}}(-t)&=t+2\vartheta_LB_{k,L}(t),
&\mathcal H_{\boldsymbol{\vartheta}}(t)&=t+2\vartheta_RB_{k,R}(t),\qquad t\geq0.
\end{align*}
At the population scales this equals $\mathcal G_k$ in
\eqref{eq:vanishing-brownian-process}.
For each fixed $\boldsymbol{\vartheta}$, define
\begin{align*}
U(\boldsymbol{\vartheta})
&=\operatorname*{arg\,min}_{u\in\R}\mathcal H_{\boldsymbol{\vartheta}}(u).
\end{align*}
This minimizer is almost surely finite and unique by the proof of
\Cref{thm:vanishing-localization-limit}. Its CDF and quantiles are
\begin{align*}
F_{\boldsymbol{\vartheta}}(x)&=\Pp\{U(\boldsymbol{\vartheta})\leq x\},
&q_a(\boldsymbol{\vartheta})&=\inf\{x:F_{\boldsymbol{\vartheta}}(x)\geq a\},\qquad a\in(0,1).
\end{align*}
The classical joint law of the depth and location of a drifted Brownian
minimum \citep{AbramsonEvans2014} gives independent minimum depths on
the two sides, with exponential rates $(2\vartheta_L^2)^{-1}$ and
$(2\vartheta_R^2)^{-1}$. Both depths are positive almost surely, so
$\Pp\{U(\boldsymbol{\vartheta})=0\}=0$.
For $t>0$, independence of the two halves yields the density
\begin{align*}
f_{\boldsymbol{\vartheta}}(t)
&=\int_0^\infty
\frac{m\exp\{-(m+t)^2/(8\vartheta_R^2t)\}}
{4\vartheta_R^3\sqrt{2\pi t^3}}
\left(1-e^{-m/(2\vartheta_L^2)}\right)\,dm>0.
\end{align*}
Here the fraction is the joint density of the right minimum's depth $m$
and location $t$; the second factor is the probability that the left
minimum has smaller depth. Exchanging $L$ and $R$ gives
$f_{\boldsymbol{\vartheta}}(-t)$.
Thus $U(\boldsymbol{\vartheta})$ has no atoms, and
\begin{align*}
F_{\boldsymbol{\vartheta}}(y)-F_{\boldsymbol{\vartheta}}(x)
&=\int_x^y f_{\boldsymbol{\vartheta}}(t)\,dt>0,\qquad x<y.
\end{align*}
Hence $F_{\boldsymbol{\vartheta}}$ is continuous and strictly increasing.

For $\boldsymbol{\vartheta}_j\to\boldsymbol{\vartheta}$, the common Brownian paths give,
almost surely, for every $T>0$,
\begin{align*}
\sup_{|u|\leq T}|\mathcal H_{\boldsymbol{\vartheta}_j}(u)
-\mathcal H_{\boldsymbol{\vartheta}}(u)|
&\leq2\|\boldsymbol{\vartheta}_j-\boldsymbol{\vartheta}\|_\infty
\max_{\ell\in\{L,R\}}\sup_{0\leq t\leq T}|B_{k,\ell}(t)|
\longrightarrow0.
\end{align*}
The coefficient vectors are bounded, and $B_{k,\ell}(t)/t\to0$
almost surely by the proof of \Cref{thm:vanishing-localization-limit}.
Thus, for some finite $C$ and almost surely finite $T_0$,
\begin{align*}
\mathcal H_{\boldsymbol d}(\pm t)
&\geq t-2C\max_{\ell\in\{L,R\}}|B_{k,\ell}(t)|\geq t/2,
\quad t\geq T_0,\quad
\boldsymbol d\in\{\boldsymbol{\vartheta},\boldsymbol{\vartheta}_1,\boldsymbol{\vartheta}_2,\ldots\}.
\end{align*}
Since $\mathcal H_{\boldsymbol d}(0)=0$, all minimizers lie in
$[-T_0,T_0]$. Continuity and uniqueness give, for every $\varepsilon>0$,
\begin{align*}
\inf_{\substack{|u|\leq T_0\\|u-U(\boldsymbol{\vartheta})|\geq\varepsilon}}
\{\mathcal H_{\boldsymbol{\vartheta}}(u)
-\mathcal H_{\boldsymbol{\vartheta}}(U(\boldsymbol{\vartheta}))\}&>0.
\end{align*}
Uniform convergence therefore implies $U(\boldsymbol{\vartheta}_j)\to
U(\boldsymbol{\vartheta})$ almost surely and
$F_{\boldsymbol{\vartheta}_j}(x)\to F_{\boldsymbol{\vartheta}}(x)$ for every $x$.
For $a\in(0,1)$ and $\varepsilon>0$, CDF convergence and strict increase
give brackets proving continuity of $q_a$:
\begin{align*}
F_{\boldsymbol{\vartheta}}(q_a(\boldsymbol{\vartheta})-\varepsilon)
&<a<F_{\boldsymbol{\vartheta}}(q_a(\boldsymbol{\vartheta})+\varepsilon),\\
q_a(\boldsymbol{\vartheta})-\varepsilon
&<q_a(\boldsymbol{\vartheta}_j)\leq q_a(\boldsymbol{\vartheta})+\varepsilon
\qquad\text{for all sufficiently large }j.
\end{align*}

\medskip
\noindent\textbf{Step 3: Coverage.}
Take $\boldsymbol{\vartheta}=(\omega_{k,L},\omega_{k,R})$.
The Brownian motions in \eqref{eq:ci-simulated-process} are independent
of the observations, so Steps~1--2 give
$\widehat q_{k,a}=q_a(\widehat{\boldsymbol{\vartheta}})
\xrightarrow{\Pp}q_a(\boldsymbol{\vartheta})$ for
$a\in\{\alpha/2,1-\alpha/2\}$.
Together with \Cref{thm:vanishing-localization-limit}, Slutsky's theorem yields
\begin{align*}
\bigl(\widehat\kappa_k^2(\widehat\eta_k-\eta_k),
\widehat q_{k,\alpha/2},\widehat q_{k,1-\alpha/2}\bigr)
&\xrightarrow{d}\bigl(U(\boldsymbol{\vartheta}),
q_{\alpha/2}(\boldsymbol{\vartheta}),q_{1-\alpha/2}(\boldsymbol{\vartheta})\bigr).
\end{align*}
On the restriction in Step~1, $\widehat\kappa_k>0$, so
\eqref{eq:ci-confidence-interval} gives
\begin{align*}
\eta_k\in\mathcal J_{k,\alpha}
&\quad\Longleftrightarrow\quad
\widehat q_{k,\alpha/2}
\leq\widehat\kappa_k^2(\widehat\eta_k-\eta_k)
\leq\widehat q_{k,1-\alpha/2}.
\end{align*}
Since $\Pp\{U(\boldsymbol{\vartheta})=q_a(\boldsymbol{\vartheta})\}=0$
and $F_{\boldsymbol{\vartheta}}(q_a(\boldsymbol{\vartheta}))=a$,
the limiting event has boundary probability zero. Therefore,
\begin{align*}
&\Pp\{\mathcal E_{\mathrm{unif}},\ \widehat{\boldsymbol{\vartheta}}\in(0,\infty)^2,\
\eta_k\in\mathcal J_{k,\alpha}\}\\
&\qquad\longrightarrow
F_{\boldsymbol{\vartheta}}(q_{1-\alpha/2}(\boldsymbol{\vartheta}))
-F_{\boldsymbol{\vartheta}}(q_{\alpha/2}(\boldsymbol{\vartheta}))=1-\alpha.
\end{align*}
The excluded probability is $o(1)$ by Step~1, proving the claim.
\end{proof}

\section{Auxiliary results}

\subsection{Uniform empirical-moment and tensor CUSUM bounds}

The following maximal deviation bound is a direct consequence of
Theorem~1 and Remarks~1--3 of \citet{MerlevedePeligradRio2011}.

\begin{proposition}[Sub-Weibull deviation under generalized-exponential
$\alpha$-mixing]
\label{prop:subweibull-sum}
Let $N\geq1$, and let $\{V_i\}_{i = 1}^N$ be centered random variables.  Suppose
the $\alpha$-mixing coefficients
\(\alpha_V(\ell)\) of $\{V_i\}_{i = 1}^N$, defined as in
\eqref{eq:alpha-mixing}, satisfy
\begin{align*}
\alpha_V(\ell)
&\leq
b_0\exp(-b_1\ell^{\gamma_1}),
\qquad 1\leq \ell<N,
\end{align*}
for some \(b_0,b_1>0\) and \(\gamma_1\in(0,\infty]\). When
$\gamma_1=\infty$, this bound means that the variables are independent.
Let
$\gamma_2\in(0,\infty]$ and
define the uniform tail envelope by
\begin{align*}
\sigma_{V,\gamma_2}
&=
\begin{cases}
\displaystyle
\max_{1\leq t\leq N}\norm{V_t}_{\psi_{\gamma_2}},
&\gamma_2<\infty,\\[1ex]
\displaystyle
\max_{1\leq t\leq N}\norm{V_t}_{\infty},
&\gamma_2=\infty,
\end{cases}
<\infty.
\end{align*}
Assume that
\begin{align*}
\bar\gamma
&=
\left(
\frac1{\gamma_1}+\frac1{\gamma_2}
\right)^{-1}
<1,
\end{align*}
where $1/\infty=0$; the boundary case
$(\gamma_1,\gamma_2)=(\infty,1)$ is also allowed.
Then, for every \(x\geq1\), with
probability at least \(1-e^{-x}\),
\begin{align}
\max_{1\leq m\leq N}
\abs{\sum_{t=1}^m V_t}
&\leq
C\sigma_{V,\gamma_2}
\left\{
\sqrt{Nx}
+
(x+\log N)^{1/\gamma_1+1/\gamma_2}
\right\}.
\label{eq:dependent-subweibull-sum-tail}
\end{align}
Here $C>0$ depends only on $\gamma_1,\gamma_2,b_0,b_1$; when
$\gamma_2=\infty$, the dependence on $\gamma_2$ is omitted, and when
$\gamma_1=\infty$, the dependence on $\gamma_1,b_0,b_1$ is omitted.
\end{proposition}

\begin{proof}[Proof of \Cref{prop:subweibull-sum}]
Theorem~1 of \citet{MerlevedePeligradRio2011} is stated for $N\geq4$.
To cover the remaining cases, note that, for $1\leq N\leq3$,
the tail envelope $\sigma_{V,\gamma_2}$ and a union bound over at most
three variables give the stated inequality.
Hence assume $N\geq4$ below.

If $\gamma_1=\infty$, $V_t$ are independent and
$\gamma_2\leq1$. The independent Bernstein bounds in equations
(1.2)--(1.3) of \citet{MerlevedePeligradRio2011}, integrated over the
tail and combined with Doob's maximal inequality for the partial-sum
martingale, give for every $q\geq2$
\[
\left\|\max_{1\leq m\leq N}
\left|\sum_{t=1}^m V_t\right|\right\|_{L^q}
\leq C\sigma_{V,\gamma_2}
\left\{\sqrt{Nq}+(q+\log N)^{1/\gamma_2}\right\}.
\]
Taking $q=2\vee x$ and applying Markov's inequality proves
\eqref{eq:dependent-subweibull-sum-tail}, with $1/\gamma_1=0$.
Hence assume $\gamma_1<\infty$ below.

If $\sigma_{V,\gamma_2}=0$, all $V_t$ vanish almost surely and the
conclusion is immediate. Otherwise, normalize by $\sigma_{V,\gamma_2}$
and extend the sequence by zero outside $\{1,\ldots,N\}$.
For $\gamma_2<\infty$, the moment definition in
Section~\ref{sec:notation} gives
\begin{align*}
\max_{1\leq t\leq N}
\left\|V_t/\sigma_{V,\gamma_2}\right\|_{L^q}
&\leq q^{1/\gamma_2},\qquad q\geq1.
\end{align*}
For $z\geq1$, Markov's inequality with $q=z^{\gamma_2}$ yields
\begin{align*}
\max_{1\leq t\leq N}
\Pp\bigl(|V_t|>e\,\sigma_{V,\gamma_2}z\bigr)
&\leq\left(\frac{q^{1/\gamma_2}}{ez}\right)^q
=e^{-z^{\gamma_2}}.
\end{align*}
For $0\leq z<1$, use the trivial bound by one. Thus, for all $z\geq0$,
\begin{align*}
\max_{1\leq t\leq N}
\Pp\bigl(|V_t|>e\,\sigma_{V,\gamma_2}z\bigr)
&\leq\exp(1-z^{\gamma_2}),
\end{align*}
which is the required uniform generalized exponential tail bound after a fixed
rescaling. If $\gamma_2=\infty$, then
$|V_t|/\sigma_{V,\infty}\leq1$ almost surely, so the normalized tail
probability is zero for $z\geq1$.

Theorem~1 of \citet{MerlevedePeligradRio2011} is stated
under $\tau$-mixing. As explained in their Remark~1, its dependence
condition also follows from $\alpha$-mixing: the comparison between
the two coefficients, together with the uniformly bounded second
moments of the normalized variables, gives
$\tau(\ell)\leq C\alpha_V(\ell)^{1/2}\leq C\exp(-c\ell^{\gamma_1})$.
Thus the required $\tau$-mixing decay holds with the same exponent
$\gamma_1$, after adjusting constants. The zero extension preserves
these bounds and has $\alpha_V(\ell)=0$ for $\ell\geq N$.

We next bound the variance proxy in Theorem~1 of
\citet{MerlevedePeligradRio2011}.
Let $Q(v)$ be the envelope of the upper quantile functions of the
absolute normalized variables. The tail condition gives
$Q(v)\leq C\{\log(e/v)\}^{1/\gamma_2}$, with $1/\infty=0$.
By Remark~3 of \citet{MerlevedePeligradRio2011}, the $\alpha$-mixing
bound therefore yields
\begin{align*}
V \leq C\left\{1+
\sum_{\ell\geq1}\int_0^{2\alpha_V(\ell)}Q(v)^2\,dv\right\} \leq C\left\{1+\sum_{\ell\geq1}
\alpha_V(\ell)^{1/2}\right\}
\leq C^{\prime}.
\end{align*}
Thus Theorem~1 of the cited paper applies with
$\bar\gamma=(1/\gamma_1+1/\gamma_2)^{-1}$.

For $x\geq1$, substitute
\[
u=C_*\left\{\sqrt{Nx}+
(x+\log N)^{1/\bar\gamma}\right\}
\]
into that theorem. For a sufficiently large $C_*$, its first
probability term is at most $e^{-x}/3$ because
$u^{\bar\gamma}\geq C(x+\log N)$; the second and third are each at
most $e^{-x}/3$ because $V\leq C^{\prime}$ and $u^2/N\geq Cx$.
Rescaling by $\sigma_{V,\gamma_2}$ and using
$1/\bar\gamma=1/\gamma_1+1/\gamma_2$ proves
\eqref{eq:dependent-subweibull-sum-tail}.
\end{proof}

The concentration arguments below use
Proposition~\ref{prop:subweibull-sum}. We record the tensor-product tail
and product-net reductions needed in the later proofs.

\begin{lemma}[Tensor-product tails]
\label{lem:tensor-product-tail}
Under Assumption~\ref{ass:data-conditions}, fix
$u_1,\ldots,u_p\in\Sph^D$. The variables
\begin{align*}
V_t=\prod_{a=1}^pu_a^\top Y_t-
\E\prod_{a=1}^pu_a^\top Y_t
\end{align*}
are centered and satisfy
$\sup_t\norm{V_t}_{\psi_{2/p}}\leq C_p\sigma^p$.
Their $\alpha$-mixing coefficients are at most
$b_0\exp(-b_1\ell^\gamma)$ for $1\leq\ell<n$.
\end{lemma}

\begin{proof}[Proof of \Cref{lem:tensor-product-tail}]
Write $u_a=(c_a,w_a^\top)^\top\in\Sph^D$, with $w_a\in\R^D$.
Assumption~\ref{ass:data-conditions}, the triangle inequality, and
Cauchy--Schwarz give
\begin{align*}
\norm{u_a^\top Y_t}_{\psi_2}
&\leq |c_a|+\sigma\norm{w_a}_2
\leq\sqrt{1+\sigma^2}\leq\sqrt2\,\sigma,
\end{align*}
where $c_a^2+\norm{w_a}_2^2=1$ and $\sigma\geq1$.
The moment definition of the $\psi_2$ norm in
Section~\ref{sec:notation} therefore gives, for $u\geq1$,
\begin{align*}
\norm{u_a^\top Y_t}_{L^{pu}}
&\leq\sqrt{pu}\,\norm{u_a^\top Y_t}_{\psi_2}
\leq\sigma\sqrt{2pu}.
\end{align*}
See also \citet[Proposition~2.5.2(ii)]{Vershynin2018} for this moment
characterization of sub-Gaussian variables.
H\"older's inequality now yields
\begin{align*}
\norm{\prod_{a=1}^pu_a^\top Y_t}_{L^u}
&\leq(\sigma\sqrt{2pu})^p.
\end{align*}
Centering costs at most a factor two.  The mixing assertion follows because
measurable transformations cannot increase $\alpha$-mixing coefficients.
\end{proof}

For $I=(s,e]$, recall the empirical moment tensor
$\overline{\mathcal M}_I^{(p)}$ and the CUSUM statistic
$\mathcal T_I(q)$ defined in
Section~\ref{sec:preliminary-seeded-detection}.
\begin{lemma}[Fixed interval-direction concentration]
\label{lem:fixed-direction}
Let $I = (s,e] \subseteq (0, n]$ be a fixed nonempty interval, and let
$p \ge 2$ be fixed. Under Assumption~\ref{ass:data-conditions}, fix
\(u_1,\ldots,u_p\in\Sph^D\).
Then, for every \(x\geq 1\), with
probability at least \(1-e^{-x}\),
\begin{align}
\abs{
\ip{
\overline{\mathcal M}_I^{(p)}-\E\overline{\mathcal M}_I^{(p)}
}{
 u_1\otimes\cdots\otimes u_p
}_{\F}
}
\leq
C\sigma^p
\left\{
\sqrt{\frac{x}{|I|}}
+
\frac{(x+\log |I|)^{p/2+1/\gamma}}{|I|}
\right\}.
\label{eq:fixed-direction-tail}
\end{align}
\end{lemma}

\begin{proof}[Proof of Lemma~\ref{lem:fixed-direction}]
Set
\begin{align*}
V_t
&=
\prod_{j=1}^p\ip{Y_t}{u_j}
-\E\prod_{j=1}^p\ip{Y_t}{u_j}.
\end{align*}
Lemma~\ref{lem:tensor-product-tail} gives
$\max_{t\in I}\norm{V_t}_{\psi_{2/p}}\leq C_p\sigma^p$ and bounds the
$\alpha$-mixing coefficients of $\{V_t\}_{t\in I}$ by the bound in
Assumption~\ref{ass:data-conditions}.
Moreover,
\begin{align}
\ip{
\overline{\mathcal M}_I^{(p)}-\E\overline{\mathcal M}_I^{(p)}
}{
 u_1\otimes\cdots\otimes u_p
}_{\F}
&=
\frac{1}{\abs I}\sum_{t\in I}V_t.
\label{eq:direction-as-scalar-sum}
\end{align}
After reindexing $I$, apply Proposition~\ref{prop:subweibull-sum} with
\begin{align*}
\gamma_1&=\gamma,
&
\gamma_2&=2/p,
&
N&=\abs I,
&
\sigma_{V,\gamma_2}
&=\max_{t\in I}\norm{V_t}_{\psi_{\gamma_2}}
\leq C_p\sigma^p.
\end{align*}
Since $p\geq2$ and $\gamma\in[1,\infty)$,
$1/\gamma_1+1/\gamma_2=p/2+1/\gamma > 1$. Note that the temporally independent case, i.e.~$\gamma = \infty$, is also included
in the proposition. Therefore the case $m=N$ of
Proposition~\ref{prop:subweibull-sum} and
\eqref{eq:direction-as-scalar-sum}, after division by $N$, give
\eqref{eq:fixed-direction-tail}.
\end{proof}

\begin{lemma}[Approximating the tensor operator norm by a product net]
\label{lem:tensor-net}
Let \(\mathcal{M}\in(\R^{D+1})^{\otimes p}\), let
\(\mathcal N_\varepsilon\subset\Sph^D\) be a
finite Euclidean \(\varepsilon\)-net, and
suppose \(p\varepsilon<1\). Then
\begin{align}
\norm{\mathcal{M}}_{\op}
&\leq
\frac{1}{1-p\varepsilon}
\max_{v_1,\ldots,v_p\in\mathcal N_\varepsilon}
\abs{\ip{\mathcal{M}}{v_1\otimes\cdots\otimes v_p}_{\F}}.
\label{eq:tensor-net-bound}
\end{align}
Moreover, such a net exists with cardinality at most
$(1+2/\varepsilon)^{D+1}$.
\end{lemma}

\begin{proof}[Proof of Lemma~\ref{lem:tensor-net}]
Fix \(u_1,\ldots,u_p\in\Sph^D\). For each \(j\), choose
\(v_j\in\mathcal N_\varepsilon\) satisfying
\begin{align}
\norm{u_j-v_j}_2
&\leq
\varepsilon.
\label{eq:net-approximation-each-mode}
\end{align}
A telescoping expansion gives
\begin{align}
&\ip{\mathcal{M}}{u_1\otimes\cdots\otimes u_p}_{\F}
-
\ip{\mathcal{M}}{v_1\otimes\cdots\otimes v_p}_{\F}
\nonumber\\
&\qquad=
\sum_{j=1}^p
\ip{\mathcal{M}}{
 v_1\otimes\cdots\otimes v_{j-1}
 \otimes(u_j-v_j)
 \otimes u_{j+1}\otimes\cdots\otimes u_p
}_{\F}.
\label{eq:tensor-telescoping}
\end{align}
By multilinearity and the definition of \(\norm{\mathcal{M}}_{\op}\),
\(\norm{u_i}_2=\norm{v_i}_2=1\), and
\eqref{eq:net-approximation-each-mode}, every summand in
\eqref{eq:tensor-telescoping} has absolute value at most
\begin{align*}
\norm{\mathcal{M}}_{\op}
\prod_{i<j}\norm{v_i}_2
\norm{u_j-v_j}_2
\prod_{i>j}\norm{u_i}_2
&\leq
\varepsilon\norm{\mathcal{M}}_{\op}.
\end{align*}
Hence
\begin{align}
\abs{\ip{\mathcal{M}}{u_1\otimes\cdots\otimes u_p}_{\F}}
&\leq
\abs{\ip{\mathcal{M}}{v_1\otimes\cdots\otimes v_p}_{\F}}
+
p\varepsilon\norm{\mathcal{M}}_{\op}.
\label{eq:tensor-net-intermediate}
\end{align}
Taking the supremum over \(u_1,\ldots,u_p\in\Sph^D\) in
\eqref{eq:tensor-net-intermediate} and rearranging proves
\eqref{eq:tensor-net-bound}.  The cardinality assertion follows from the standard covering number bound
\citep[see e.g.][Corollary~4.2.13]{Vershynin2018}.
\end{proof}

\begin{theorem}[Uniform empirical-moment and tensor CUSUM bounds]
\label{thm:cusum}
Let $n\geq2$. Suppose Assumption~\ref{ass:data-conditions} holds and $p\geq2$ is fixed.
There exist constants $C_1,C_2>0$, depending only on
$p,\gamma,b_0,b_1$, such that the following holds. With $I=(s,e]\subseteq(0,n]$
ranging over all nonempty index intervals, define events
\begin{align}
\mathcal E_{\mathrm{unif}}
&=\left\{
\sup_{\abs I\geq1}
\frac{\norm{\overline{\mathcal M}_I^{(p)}-
\E\overline{\mathcal M}_I^{(p)}}_{\op}}
{\sqrt{(D+\log n)/\abs I}
+(D+\log n)^{p/2+1/\gamma}/\abs I}
\leq C_1\sigma^p
\right\},
\label{eq:uniform-moment-event}\\
\mathcal E_{\mathrm{cusum}}
&=\left\{
\sup_{\substack{I=(s,e]\\1\leq h\leq\abs I/2}}
\max_{s+h\leq q\leq e-h}
\frac{\norm{\mathcal T_I(q)-\E\mathcal T_I(q)}_{\op}}
{\sqrt{D+\log n}
+(D+\log n)^{p/2+1/\gamma}/\sqrt h}
\leq C_2\sigma^p
\right\}.
\label{eq:uniform-cusum-event}
\end{align}
Then $\mathcal E_{\mathrm{unif}}\subseteq\mathcal E_{\mathrm{cusum}}$ and
\begin{align*}
\Pp(\mathcal E_{\mathrm{unif}}\cap\mathcal E_{\mathrm{cusum}})
=\Pp(\mathcal E_{\mathrm{unif}})
\geq1-n^{-5}.
\end{align*}
\end{theorem}

\begin{proof}[Proof of Theorem~\ref{thm:cusum}]
Set
$\nu=p/2+1/\gamma$ and $\varepsilon=(2p)^{-1}$, and choose a Euclidean
$\varepsilon$-net $\mathcal N_\varepsilon\subset\Sph^D$ with
$\abs{\mathcal N_\varepsilon}\leq(1+4p)^{D+1}$.  Apply
Lemma~\ref{lem:fixed-direction} with
\begin{align*}
y_0=7\log n+p(D+1)\log(1+4p)
\end{align*}
to every nonempty integer interval and every element of
$\mathcal N_\varepsilon^p$.  There are fewer than
$n^2(1+4p)^{p(D+1)}$ interval--direction pairs, so a union bound makes their total
failure probability at most $n^{-5}$.  On the complementary event,
Lemma~\ref{lem:tensor-net} gives, simultaneously for all nonempty intervals,
\begin{align*}
\norm{\overline{\mathcal M}_I^{(p)}-\E\overline{\mathcal M}_I^{(p)}}_{\op}
&\leq C\sigma^p\left\{
\sqrt{\frac{y_0}{\abs I}}
+\frac{(y_0+\log\abs I)^\nu}{\abs I}
\right\}\\
&\leq C\sigma^p\left\{
\sqrt{\frac{D+\log n}{\abs I}}
+\frac{(D+\log n)^\nu}{\abs I}
\right\}.
\end{align*}
The second inequality uses $\log\abs I\leq\log n$, $D+1\leq2D$,
and fixed $p$.
Thus, after enlarging $C_1$ if necessary,
$\Pp(\mathcal E_{\mathrm{unif}})\geq1-n^{-5}$.

Work on $\mathcal E_{\mathrm{unif}}$ and set $L=D+\log n$.  Fix an admissible
interval $I=(s,e]$, an integer $h$, and
$q\in\{s+h,\ldots,e-h\}$.  Put
\begin{align*}
\ell&=q-s,
&
r&=e-q.
\end{align*}
The CUSUM definition gives the exact identity
\begin{align*}
\mathcal T_I(q)-\E\mathcal T_I(q)
&=\sqrt{\frac{\ell r}{\ell+r}}
\left\{(\overline{\mathcal M}_{(s,q]}^{(p)}
-\E\overline{\mathcal M}_{(s,q]}^{(p)})-(\overline{\mathcal M}_{(q,e]}^{(p)}
-\E\overline{\mathcal M}_{(q,e]}^{(p)})\right\}.
\end{align*}
Because $\ell,r\geq h\geq1$, both child intervals are nonempty, so the
defining bound for $\mathcal E_{\mathrm{unif}}$ applies to each.
The triangle inequality therefore gives
\begin{align*}
\norm{\mathcal T_I(q)-\E\mathcal T_I(q)}_{\op}
&\leq C\sigma^p\sqrt{\frac{\ell r}{\ell+r}}
\left\{
\sqrt{\frac L\ell}+\sqrt{\frac Lr}
+\frac{L^\nu}{\ell}+\frac{L^\nu}{r}
\right\}.
\end{align*}
The Gaussian coefficients satisfy
\begin{align*}
\sqrt{\frac{\ell r}{\ell+r}}
\left(\frac1{\sqrt\ell}+\frac1{\sqrt r}\right)
&=
\sqrt{\frac r{\ell+r}}+\sqrt{\frac\ell{\ell+r}}
\leq\sqrt2,
\end{align*}
and the large-deviation coefficients satisfy
\begin{align*}
\sqrt{\frac{\ell r}{\ell+r}}
\left(\frac1\ell+\frac1r\right)
&\leq\frac1{\sqrt\ell}+\frac1{\sqrt r}
\leq\frac2{\sqrt h}.
\end{align*}
Since $I$, $h$, and $q$ were arbitrary, enlarging $C_2$ if necessary gives
$\mathcal E_{\mathrm{unif}}\subseteq\mathcal E_{\mathrm{cusum}}$.  Combining this
inclusion with the probability bound for $\mathcal E_{\mathrm{unif}}$ completes the
proof.
\end{proof}

\subsection{Results on seeded intervals}

\begin{lemma}[Size of the seeded interval collection]
\label{lem:seeded-collection-size}
The seeded interval collection \(\mathcal I\) in \eqref{eq:preliminary-seeded-collection} satisfies
\begin{align}
\abs{\mathcal I}
\leq \frac n4 \quad \text{and} \quad
\sum_{I\in\mathcal I}\abs I \leq Cn\log n.
\end{align}
\end{lemma}

\begin{proof}[Proof of Lemma~\ref{lem:seeded-collection-size}]
Recall the definitions given in \Cref{sec:preliminary-seeded-detection} including the dyadic interval scale set $\mathcal L_n$ and the seeded interval $\mathcal I_{\ell}$ at scale $\ell$. At a dyadic interval scale \(\ell\in\mathcal L_n\), we have
\begin{align*}
\abs{\mathcal I_\ell}
&\leq
\left\lfloor\frac{2(n-\ell)}{\ell}\right\rfloor+2
\leq \frac{2n}{\ell}.
\end{align*}
Therefore
\begin{align*}
\abs{\mathcal I}
&\leq
\sum_{\ell\in\mathcal L_n}\frac{2n}{\ell}
\leq
2n\sum_{r=4}^{\infty}2^{-r}
=\frac n4,
\end{align*}
and
\begin{align*}
\sum_{I\in\mathcal I}\abs I
&\leq
\sum_{\ell\in\mathcal L_n}\ell\abs{\mathcal I_\ell}
\leq 2n\abs{\mathcal L_n}
\leq Cn\log n.
\end{align*}
Removing duplicate intervals can only decrease both quantities.
\end{proof}

\begin{lemma}[Isolating seeded interval]
\label{lem:isolating-seeded-interval}
If $\Delta_{\min}\geq256$, then, for every $k=1,\ldots,K$, there exists
\(I_k^*=(a_k^*,b_k^*]\in\mathcal I\) that contains \(\eta_k\) and no other
change point and satisfies
\begin{align}
\frac{\Delta_{\min}}{32}
<\abs{I_k^*}
\leq\frac{\Delta_{\min}}{16},
\label{eq:isolating-seeded-length}
\end{align}
and
\begin{align}
\min\{\eta_k-a_k^*,b_k^*-\eta_k\}
\geq\frac{\abs{I_k^*}}4.
\label{eq:isolating-seeded-margins}
\end{align}
\end{lemma}

\begin{proof}[Proof of Lemma~\ref{lem:isolating-seeded-interval}]
\noindent\textbf{Step 1: Choice of scale.}
Let $\ell$ be the largest dyadic length satisfying
$\ell\leq\Delta_{\min}/16$.
Because $256\leq\Delta_{\min}\leq n$, we have $16\leq\ell\leq n$, so
$\ell\in\mathcal L_n$.  The next larger dyadic length satisfies
$2\ell>\Delta_{\min}/16$, and therefore
\begin{align*}
\frac{\Delta_{\min}}{32}<\ell\leq\frac{\Delta_{\min}}{16}.
\end{align*}
This proves \eqref{eq:isolating-seeded-length} once an interval of length
$\ell$ is selected.

\noindent\textbf{Step 2: Choice of interval.}
For $j\in\mathbb Z$, the central half of the grid interval
$(j\ell/2,j\ell/2+\ell]$ is
\begin{align*}
H_j
&=\left(j\ell/2+\ell/4,\,j\ell/2+3\ell/4\right].
\end{align*}
The intervals $\{H_j:j\in\mathbb Z\}$ cover the real line, so choose $j$ such
that $\eta_k\in H_j$, and set
$a_k^*=j\ell/2$, $b_k^*=a_k^*+\ell$, and $I_k^*=(a_k^*,b_k^*]$.
It remains to check that this grid interval belongs to the finite collection
$\mathcal I_\ell$.  Since
$\eta_k\geq\Delta_{k-1}\geq\Delta_{\min}$ and
$n-\eta_k\geq\Delta_k\geq\Delta_{\min}$, while
$\Delta_{\min}\geq16\ell$, central-half membership
gives
\begin{align*}
a_k^*
\geq\eta_k-3\ell/4\geq0,
\quad \text{and} \quad
b_k^*
\leq\eta_k+3\ell/4\leq n.
\end{align*}
Thus $j\geq0$ and $j\leq\lfloor2(n-\ell)/\ell\rfloor$, which proves
$I_k^*\in\mathcal I_\ell\subseteq\mathcal I$.

\noindent\textbf{Step 3: Margins and isolation.}
The inclusion $\eta_k\in H_j$ yields
\begin{align*}
\eta_k-a_k^*
\geq\ell/4,
\quad \text{and} \quad
b_k^*-\eta_k
\geq\ell/4.
\end{align*}
This proves \eqref{eq:isolating-seeded-margins}.  Finally,
$\ell<\Delta_{\min}$ and $I_k^*$ contains $\eta_k$, whereas the adjacent
changes are at distances $\Delta_{k-1}\geq\Delta_{\min}$ and
$\Delta_k\geq\Delta_{\min}$ from $\eta_k$.  Hence
$I_k^*$ lies strictly between $\eta_{k-1}$ and $\eta_{k+1}$ and contains no
change other than $\eta_k$.
\end{proof}

\subsection{Estimating the most significant
projection direction}
\label{app:direction-setup}

We use the distance $d_\pm$ defined in Section~\ref{sec:notation}.
Throughout this subsection, we consider order-$p$ moment tensors and
omit their superscript $(p)$.

We impose the following eigengap type of assumption such that the leading projection direction is identifiable.
\begin{assumption}
\label{ass:local-refinement-separation}
There exist fixed constants $c_0,r_0>0$ such that, for every
$k=1,\ldots,K$ and $0<\rho\leq r_0$,
\begin{align*}
\kappa_k
-\sup_{\substack{u\in\Sph^D\\d_\pm(u,\nu_k)\geq\rho}}
\abs{\ip{\Theta_k}{u^{\otimes p}}_{\F}}
&\geq c_0\kappa_k\rho^2.
\end{align*}
\end{assumption}

Similar assumptions have been considered in the study of tensor
eigenvectors and best rank-one approximations by \citet{KoldaMayo2011},
\href{https://www.stat.uchicago.edu/~lekheng/work/nonneg.pdf}{Qi, Comon and Lim (2016)},
and \href{https://doi.org/10.1137/17M1133312}{Jaffe, Weiss and Nadler (2018)}.
\Cref{lem:cp-direction-separation} shows that
\Cref{ass:local-refinement-separation} holds for tensors admitting an
orthogonal symmetric CP decomposition with a suitable eigengap assumption.

\begin{lemma}
\label{lem:cp-direction-separation}
Suppose $2\leq R\leq D+1$ and the symmetric tensor
$\mathcal T\in(\R^{D+1})^{\otimes p}$ admits the orthogonal
symmetric CP decomposition
\begin{align*}
\mathcal T
&=
\sum_{j=1}^{R}\lambda_j \phi_j^{\otimes p},
\end{align*}
where $\phi_1,\ldots,\phi_R\in\R^{D+1}$ are orthonormal.
Define the strictly positive weight gap by
\begin{align}
\Delta
&=
\abs{\lambda_1}
-
\max_{2\leq j\leq R}\abs{\lambda_j}
>0.
\label{eq:cp-weight-gap}
\end{align}
Then $\phi_1$ is the unique maximizer, up to sign, of
\begin{align*}
u&\longmapsto\abs{\ip{\mathcal T}{u^{\otimes p}}_{\F}},
\qquad u\in\Sph^D.
\end{align*}
Moreover, \Cref{ass:local-refinement-separation} holds with
$r_0=1$ and $c_0=\Delta/(2\abs{\lambda_1})$.
\end{lemma} 
\begin{proof}[Proof of \Cref{lem:cp-direction-separation}]
For $v\in\Sph^D$, orthonormality and $p\geq2$ give
\begin{align*}
\sum_{j=1}^R\abs{\phi_j^\top v}^p
&\leq\sum_{j=1}^R\abs{\phi_j^\top v}^2\leq1.
\end{align*}
Consequently, generalized H\"older's inequality gives, for unit vectors
$v_1,\ldots,v_p$,
\begin{align*}
\abs{\ip{\mathcal T}{v_1\otimes\cdots\otimes v_p}_{\F}}
&\leq\abs{\lambda_1}
\prod_{\ell=1}^p
\left(\sum_{j=1}^R\abs{\phi_j^\top v_\ell}^p\right)^{1/p}
\leq\abs{\lambda_1}.
\end{align*}
Equality is attained at $v_1=\cdots=v_p=\phi_1$, so
$\norm{\mathcal T}_{\op}=\abs{\lambda_1}$.
For any $v\in\Sph^D$, the weight gap in \eqref{eq:cp-weight-gap} yields
\begin{align*}
\abs{\ip{\mathcal T}{v^{\otimes p}}_{\F}}
&\leq\abs{\lambda_1}\abs{\phi_1^\top v}^p
 +(\abs{\lambda_1}-\Delta)\sum_{j=2}^R\abs{\phi_j^\top v}^p\\
&\leq\abs{\lambda_1}-\Delta(1-\abs{\phi_1^\top v}^2).
\end{align*}
Since $d_\pm(v,\phi_1)^2=2(1-\abs{\phi_1^\top v})$, it follows that
\begin{align*}
\norm{\mathcal T}_{\op}
-\abs{\ip{\mathcal T}{v^{\otimes p}}_{\F}}
&\geq\frac\Delta2 d_\pm(v,\phi_1)^2.
\end{align*}
The right side is positive unless $v\in\{\phi_1,-\phi_1\}$.
Taking the infimum over $d_\pm(v,\phi_1)\geq\rho$ proves the separation
condition with $r_0=1$ and $c_0=\Delta/(2\abs{\lambda_1})$.
\end{proof}

The following deterministic perturbation bound applies to symmetric
matrices and to symmetric tensors of higher order.

\begin{lemma}[Perturbation of a maximizing direction]
\label{lem:stage-ii-perturbation}
Let $p\geq2$ and $D\geq1$, and let
$\mathcal A,\widehat{\mathcal A}\in(\R^{D+1})^{\otimes p}$ be real symmetric
tensors.  Let
\begin{align*}
\kappa&=\norm{\mathcal A}_{\op}>0 \quad \text{and} \quad
\nu \in\argmaxsmall_{u\in\Sph^D}
\abs{\ip{\mathcal A}{u^{\otimes p}}_{\F}}.
\end{align*}
Suppose that $c_0,r_0>0$ satisfy, for every $0<\rho\leq r_0$,
\begin{align}
\kappa
-\sup_{\substack{u\in\Sph^D\\d_\pm(u,\nu)\geq\rho}}
\abs{\ip{\mathcal A}{u^{\otimes p}}_{\F}}
&\geq c_0\kappa\rho^2.
\label{eq:stage-ii-lemma-separation}
\end{align}
Let $\widehat\nu$ be any exact maximizer of
$\abs{\ip{\widehat{\mathcal A}}{u^{\otimes p}}_{\F}}$ over
$u\in\Sph^D$, and define
$\varepsilon =\norm{\widehat{\mathcal A}-\mathcal A}_{\op}$.
If
\begin{align}
\frac{\varepsilon}{\kappa}
&\leq\min\left\{\frac14,\frac{c_0r_0^2}{4},\frac{c_0}{16p^2}\right\},
\label{eq:stage-ii-lemma-smallness}
\end{align}
then
\begin{align*}
d_\pm(\widehat\nu,\nu)
&\leq\frac{p\varepsilon}{c_0\kappa}.
\end{align*}
\end{lemma}

\begin{proof}[Proof of \Cref{lem:stage-ii-perturbation}]
For a real symmetric tensor, the operator norm equals the maximum absolute
diagonal contraction by Banach's identity
\citep[Eq.~(5.1)]{FriedlandLim2018}.
For every $v\in\Sph^D$, the definition of $\varepsilon$ gives
\begin{align}
\abs{\ip{\widehat{\mathcal A}-\mathcal A}{v^{\otimes p}}_{\F}}
&\leq\varepsilon.
\label{eq:stage-ii-lemma-diagonal-noise}
\end{align}
We write $d=d_\pm(\widehat\nu,\nu)$.

\medskip
\noindent\textbf{Step 1: localization.}
Comparison of the empirical objective at $\nu$ and $\widehat\nu$ yields
\begin{align*}
\kappa
-\abs{\ip{\mathcal A}{\widehat\nu^{\otimes p}}_{\F}}
\leq
\abs{\ip{\widehat{\mathcal A}}{\nu^{\otimes p}}_{\F}}
-\abs{\ip{\widehat{\mathcal A}}{\widehat\nu^{\otimes p}}_{\F}}
+2\varepsilon
\leq2\varepsilon.
\end{align*}
If $d>0$, separation at $\rho=\min\{d,r_0\}$ gives
$c_0\kappa\min\{d,r_0\}^2\leq2\varepsilon$.
Condition~\eqref{eq:stage-ii-lemma-smallness} therefore forces $d<r_0$
and $d^2\leq2\varepsilon/(c_0\kappa)\leq1/(8p^2)$.
The case $d=0$ is immediate.
Consequently,
\begin{align}
d&<\min\{r_0,1/(2p)\}.
\label{eq:stage-ii-lemma-localization}
\end{align}

\medskip
\noindent\textbf{Step 2: sign alignment.}
Choose $\widetilde\nu\in\{\widehat\nu,-\widehat\nu\}$ so that
$\norm{\widetilde\nu-\nu}_2=d$.  Its absolute population and empirical
objectives agree with those of $\widehat\nu$ by homogeneity.  For any
order-$p$ tensor $\mathcal B$ and unit vectors $u,v$, telescoping gives
\begin{align}
\abs{\ip{\mathcal B}{u^{\otimes p}-v^{\otimes p}}_{\F}}
=\left|\sum_{a=1}^p
\ip{\mathcal B}{u^{\otimes(a-1)}\otimes(u-v)
\otimes v^{\otimes(p-a)}}_{\F}\right|\leq p\norm{\mathcal B}_{\op}\norm{u-v}_2.
\label{eq:stage-ii-lemma-lipschitz}
\end{align}
It follows from \eqref{eq:stage-ii-lemma-localization} that
\begin{align*}
\abs{\ip{\mathcal A}{\widetilde\nu^{\otimes p}-\nu^{\otimes p}}_{\F}}
&\leq p\kappa d<\kappa/2.
\end{align*}
Let $s=\operatorname{sign}\{\ip{\mathcal A}{\nu^{\otimes p}}_{\F}\}$.
Then $s\in\{-1,1\}$ and
\begin{align*}
s\ip{\mathcal A}{\widetilde\nu^{\otimes p}}_{\F}
&>\kappa/2,
&
s\ip{\mathcal A}{\nu^{\otimes p}}_{\F}
&=\kappa.
\end{align*}
By \eqref{eq:stage-ii-lemma-diagonal-noise} and
$\varepsilon\leq\kappa/4$, the corresponding empirical
contractions satisfy
\begin{align*}
s\ip{\widehat{\mathcal A}}{\widetilde\nu^{\otimes p}}_{\F}
&>\kappa/4,
&
s\ip{\widehat{\mathcal A}}{\nu^{\otimes p}}_{\F}
&\geq3\kappa/4.
\end{align*}
All four contractions therefore have sign $s$.

\medskip
\noindent\textbf{Step 3: the refined objective comparison.}
Exact maximization and the common empirical sign imply
\begin{align*}
s\ip{\widehat{\mathcal A}}{\widetilde\nu^{\otimes p}}_{\F}
&\geq
s\ip{\widehat{\mathcal A}}{\nu^{\otimes p}}_{\F}.
\end{align*}
Using the
common population sign, we obtain
\begin{align*}
\kappa
-\abs{\ip{\mathcal A}{\widetilde\nu^{\otimes p}}_{\F}}
&=s\ip{\mathcal A}{\nu^{\otimes p}-\widetilde\nu^{\otimes p}}_{\F}
\\
&\leq s\ip{\widehat{\mathcal A}-\mathcal A}{\widetilde\nu^{\otimes p}-\nu^{\otimes p}}_{\F}
\\
&\leq p\varepsilon d,
\end{align*}
where the last inequality uses \eqref{eq:stage-ii-lemma-lipschitz}.

\medskip
\noindent\textbf{Step 4: the direction bound.}
If $d=0$, there is nothing to prove.  Otherwise $0<d<r_0$, and applying
\eqref{eq:stage-ii-lemma-separation} at $\rho=d$ to $\widetilde\nu$ gives
\begin{align*}
c_0\kappa d^2
&\leq\kappa
-\abs{\ip{\mathcal A}{\widetilde\nu^{\otimes p}}_{\F}}
\leq p\varepsilon d.
\end{align*}
Dividing by $c_0\kappa d$ proves
$d\leq p\varepsilon/(c_0\kappa)$.
\end{proof}

\begin{theorem}[Accuracy of the estimated direction]
\label{thm:stage-ii-direction}
Let $p\geq2$ and $D\geq1$ be integers. Let $\mathcal M_L,\mathcal M_R$ be real
symmetric order-$p$ tensors on $\R^{D+1}$, with symmetric estimators
$\widehat{\mathcal M}_L,\widehat{\mathcal M}_R$. Define
\begin{align*}
\Theta&=\mathcal M_L-\mathcal M_R,
&\widehat\Theta&=\widehat{\mathcal M}_L-\widehat{\mathcal M}_R,
&\kappa&=\norm{\Theta}_{\op}>0,
\end{align*}
and choose exact maximizing directions on the unit sphere $\Sph^D$
\begin{align*}
\nu&\in\argmaxsmall_{u\in\Sph^D}
\abs{\ip{\Theta}{u^{\otimes p}}_{\F}},
&\widehat\nu&\in\argmaxsmall_{u\in\Sph^D}
\abs{\ip{\widehat\Theta}{u^{\otimes p}}_{\F}}.
\end{align*}
For $u,v\in\Sph^D$, let
$d_\pm(u,v)=\min\{\norm{u-v}_2,\norm{u+v}_2\}$.
Suppose $c_0,r_0>0$ satisfy, for every $0<\rho\leq r_0$,
\begin{align*}
\kappa-\sup_{\substack{u\in\Sph^D\\d_\pm(u,\nu)\geq\rho}}
\abs{\ip{\Theta}{u^{\otimes p}}_{\F}}
&\geq c_0\kappa\rho^2.
\end{align*}
For a deterministic $\epsilon_n>0$, define
\begin{align}
\mathcal E_{\mathrm{mom}}
&=\left\{\max_{\ell\in\{L,R\}}
\norm{\widehat{\mathcal M}_\ell-\mathcal M_\ell}_{\op}
\leq\epsilon_n\right\}.
\label{eq:stage-ii-event}
\end{align}
If
\begin{align}
\frac{\epsilon_n}{\kappa}
&\leq\min\left\{\frac18,\frac{c_0r_0^2}{8},\frac{c_0}{32p^2}\right\},
\label{eq:stage-ii-smallness}
\end{align}
then, on $\mathcal E_{\mathrm{mom}}$, every such $\widehat\nu$ satisfies
\begin{align}
d_\pm(\widehat\nu,\nu)
&\leq\frac{2p\epsilon_n}{c_0\kappa}.
\label{eq:stage-ii-direction-rate}
\end{align}
\end{theorem}

\begin{proof}[Proof of \Cref{thm:stage-ii-direction}]
On $\mathcal E_{\mathrm{mom}}$,
\begin{align*}
\norm{\widehat\Theta-\Theta}_{\op}
&\leq\sum_{\ell\in\{L,R\}}
\norm{\widehat{\mathcal M}_\ell-\mathcal M_\ell}_{\op}
\leq2\epsilon_n.
\end{align*}
Apply \Cref{lem:stage-ii-perturbation} with
$\mathcal A=\Theta$, $\widehat{\mathcal A}=\widehat\Theta$, and
$\varepsilon\leq2\epsilon_n$.
Condition~\eqref{eq:stage-ii-smallness} implies
\eqref{eq:stage-ii-lemma-smallness}, giving
\eqref{eq:stage-ii-direction-rate}.
\end{proof}

\subsection{Maximal bounds for projected moment sums}

\begin{lemma}[Fourth-moment maximal bound]
\label{lem:mixing-fourth-maximal}
Let $V_1,\ldots,V_m$ be centered real random variables with
$\max_t\norm{V_t}_{L^8}\leq B$. Suppose their $\alpha$-mixing coefficients
satisfy the bound in \Cref{ass:data-conditions}, including independence
when $\gamma=\infty$. Then
\begin{align}
\left\|\max_{0\leq r\leq m}
\abs{\sum_{t=1}^rV_t}\right\|_{L^4}
&\leq CB\sqrt m,
\end{align}
where $C > 0$ depends only on the fixed mixing parameters.
\end{lemma}

\begin{proof}[Proof of \Cref{lem:mixing-fourth-maximal}]
\noindent\textbf{Step 1: Fourth moments of interval sums.}
For ordered indices $t_1\leq t_2\leq t_3\leq t_4$, let
$d=\max\{t_2-t_1,t_4-t_3\}$. If $d=t_2-t_1\geq1$, centering,
Davydov's covariance inequality \citep{Davydov1968} with exponents
$8$ and $8/3$, and H\"older's inequality give
\begin{align*}
\abs{\E(V_{t_1}V_{t_2}V_{t_3}V_{t_4})}
&=\abs{\operatorname{Cov}(V_{t_1},V_{t_2}V_{t_3}V_{t_4})}
\leq CB^4\alpha_V(d)^{1/2}.
\end{align*}
If $d=t_4-t_3\geq1$, isolate $V_{t_4}$ instead. For $d=0$,
H\"older's inequality bounds the product moment by $B^4$. Hence, for
finite $\gamma$,
\begin{align*}
\abs{\E(V_{t_1}V_{t_2}V_{t_3}V_{t_4})}
&\leq CB^4\exp(-c d^\gamma).
\end{align*}
For a consecutive interval $I$ of $\ell$ indices,
\begin{align}
\E\abs{\sum_{t\in I}V_t}^4
&\leq24\sum_{\substack{t_1,\ldots,t_4\in I\\
t_1\leq t_2\leq t_3\leq t_4}}
\abs{\E(V_{t_1}V_{t_2}V_{t_3}V_{t_4})} \leq CB^4\ell^2\sum_{d=0}^\infty(2d+1)e^{-cd^\gamma}
\leq CB^4\ell^2.
\label{eq:improved-mixing-block-fourth}
\end{align}
The first inequality uses the fourth-power expansion, the triangle
inequality, and at most $4!$ permutations per ordered quadruple. The second uses the preceding moment
bound, at most $\ell^2$ choices of $t_1$ and $t_3-t_2$ for each outer-gap
pair $(a,b)=(t_2-t_1,t_4-t_3)$, and $2d+1$ pairs with
$\max\{a,b\}=d$. The last inequality follows from convergence of the
series for finite $\gamma\geq1$.
Under independence, a product
expectation can be nonzero only if $t_1=t_2$ and $t_3=t_4$.
There are at most $\ell^2$ such ordered quadruples, each with absolute
product expectation at most $B^4$, so the same bound holds.

\medskip
\noindent\textbf{Step 2: Dyadic decomposition.}
Choose $J\geq0$ with $m\leq2^J<2m$, and set $V_t=0$ for
$m<t\leq2^J$. For every interval $I\subseteq\{1,\ldots,2^J\}$,
\begin{align*}
\sum_{t\in I}V_t&=\sum_{t\in I\cap\{1,\ldots,m\}}V_t.
\end{align*}
The intersection is empty or an interval of length at most $|I|$;
hence \eqref{eq:improved-mixing-block-fourth} also holds for the extended
sequence.

For $q=0,\ldots,J$ and $b=0,\ldots,2^q-1$, define
\begin{align*}
I_{q,b}&=\{b2^{J-q}+1,\ldots,(b+1)2^{J-q}\}.
\end{align*}
Every prefix $\{1,\ldots,r\}$ is a disjoint union of at most one
such interval at each level $q$. Thus
\begin{align*}
\max_{0\leq r\leq m}\abs{\sum_{t=1}^rV_t}
&\leq\sum_{q=0}^J\max_{0\leq b<2^q}
\abs{\sum_{t\in I_{q,b}}V_t}.
\end{align*}
Minkowski's inequality, followed by bounding a maximum of fourth
powers by their sum, gives
\begin{align*}
\left\|\max_{0\leq r\leq m}\abs{\sum_{t=1}^rV_t}\right\|_{L^4}
&\leq\sum_{q=0}^J
\left\{\sum_{b=0}^{2^q-1}\E\abs{\sum_{t\in I_{q,b}}V_t}^4\right\}^{1/4}\\
&\leq CB\,2^{J/2}\sum_{q=0}^J2^{-q/4}
\leq CB\sqrt m.
\end{align*}
The second inequality uses \eqref{eq:improved-mixing-block-fourth}
and $|I_{q,b}|=2^{J-q}$; the last uses $2^J<2m$ and
$\sum_{q=0}^\infty2^{-q/4}<\infty$.
\end{proof}

The next lemma controls projected moment sums uniformly over directions
near the population direction, using forward and reverse sums for candidate
locations to the right and left of a true change, respectively.

\begin{lemma}[Maximal bounds near a fixed direction]
\label{lem:local-direction-maximal}
Suppose \Cref{ass:data-conditions} holds and $p\geq2$ is fixed.
Fix deterministic integers $0\leq a<b\leq n$ and take either the forward
list $(Y_{a+1},\ldots,Y_b)$ or the reverse list
$(Y_b,\ldots,Y_{a+1})$. Within this lemma, relabel the chosen list as
$Y_1,\ldots,Y_{b-a}$ and define
\begin{align*}
S_r&=\sum_{t=1}^r\{Y_t^{\otimes p}-\E Y_t^{\otimes p}\},
\qquad 1\leq r\leq m.
\end{align*}
Zero centered tensor summands may be included at either end, with
$m\geq1$ denoting the total number of terms.
For every fixed $\nu\in\Sph^D$ and $\rho\geq0$,
\begin{align}
\left\|\max_{1\leq r\leq m}
\abs{\ip{S_r}{\nu^{\otimes p}}_{\F}}\right\|_{L^4}
&\leq C\sigma^p\sqrt m,
\label{eq:improved-fixed-maximal}\\
\left\|\sup_{\norm{h}_2\leq\rho}\max_{1\leq r\leq m}
\abs{\ip{S_r}{(\nu+h)^{\otimes p}-\nu^{\otimes p}}_{\F}}
\right\|_{L^4}
&\leq C\sigma^p\sqrt m
\sum_{j=1}^p\{\rho\sqrt{D+1}\}^{j}.
\label{eq:improved-local-maximal}
\end{align}
Here $C>0$ depends only on $p$ and the fixed mixing parameters.
\end{lemma}

\begin{proof}[Proof of \Cref{lem:local-direction-maximal}]
\textbf{Step 1: bounds for contracted tensor sums.}
For $j=1,\ldots,p$, define the order-$j$ tensor
$A_{r,j}\in(\R^{D+1})^{\otimes j}$ by
\begin{align*}
A_{r,j}
&=\sum_{t=1}^r\left\{
(\nu^\top Y_t)^{p-j}Y_t^{\otimes j}
-\E[(\nu^\top Y_t)^{p-j}Y_t^{\otimes j}]
\right\},
\end{align*}
with added summands set to zero, and let
$A_{r,0}=\ip{S_r}{\nu^{\otimes p}}_{\F}$.
Each coordinate summand is a centered product of $p$ unit-vector
projections, so its $L^8$ norm is at most $C_p\sigma^p$ by
\Cref{lem:tensor-product-tail}. The mixing bound is inherited under
coordinate maps and restriction to a block. Reversal uses
$\alpha(\mathcal A,\mathcal B)=\alpha(\mathcal B,\mathcal A)$;
added zeros generate trivial sigma-fields. Thus
\Cref{lem:mixing-fourth-maximal} yields
\begin{align*}
\left\|\max_{1\leq r\leq m}|(A_{r,j})_{\boldsymbol i}|\right\|_{L^4}
&\leq C\sigma^p\sqrt m
\end{align*}
for every coordinate $\boldsymbol i$. For $j=0$, the coordinate index
is omitted, giving \eqref{eq:improved-fixed-maximal}.
For $j\geq1$,
\begin{align}
\left\|\max_{1\leq r\leq m}\norm{A_{r,j}}_{\F}\right\|_{L^4}^{2}
&\leq\left\|\sum_{\boldsymbol i\in\{1,\ldots,D+1\}^j}
\max_{1\leq r\leq m}|(A_{r,j})_{\boldsymbol i}|^2\right\|_{L^2}
\nonumber\\
&\leq\sum_{\boldsymbol i\in\{1,\ldots,D+1\}^{j}}
\left\|\max_{1\leq r\leq m}|(A_{r,j})_{\boldsymbol i}|\right\|_{L^4}^{2}
\nonumber\\
&\leq Cm\sigma^{2p}(D+1)^j.
\label{eq:improved-frobenius-maximal}
\end{align}
The first inequality takes the maximum coordinatewise, the second is
Minkowski's inequality in $L^2$, and the last sums the coordinate
bounds over $(D+1)^j$ entries.

\noindent\textbf{Step 2: expansion around the fixed direction.}
Symmetry gives, for every $h\in\R^{D+1}$,
\begin{align}
\ip{S_r}{(\nu+h)^{\otimes p}-\nu^{\otimes p}}_{\F}
&=\sum_{j=1}^p\binom pj\ip{A_{r,j}}{h^{\otimes j}}_{\F}.
\end{align}
Since $\norm{h^{\otimes j}}_{\F}=\norm h_2^j$, Cauchy--Schwarz implies
\begin{align*}
\sup_{\norm h_2\leq\rho}\max_{1\leq r\leq m}
\abs{\ip{S_r}{(\nu+h)^{\otimes p}-\nu^{\otimes p}}_{\F}}
&\leq\sum_{j=1}^p\binom pj\rho^j
\max_{1\leq r\leq m}\norm{A_{r,j}}_{\F}.
\end{align*}
Taking $L^4$ norms gives
\begin{align*}
\left\|\sup_{\norm h_2\leq\rho}\max_{1\leq r\leq m}
\abs{\ip{S_r}{(\nu+h)^{\otimes p}-\nu^{\otimes p}}_{\F}}
\right\|_{L^4}
&\leq\sum_{j=1}^p\binom pj\rho^j
\left\|\max_{1\leq r\leq m}\norm{A_{r,j}}_{\F}\right\|_{L^4}\\
&\leq C\sigma^p\sqrt m\sum_{j=1}^p\binom pj
\{\rho\sqrt{D+1}\}^{j}\\
&\leq 2^pC\sigma^p\sqrt m\sum_{j=1}^p
\{\rho\sqrt{D+1}\}^{j}.
\end{align*}
The first inequality is Minkowski's inequality, the second follows from
\eqref{eq:improved-frobenius-maximal}, and the third uses
$\binom pj\leq2^p$. Since $p$ is fixed, absorbing $2^p$ into $C$
proves \eqref{eq:improved-local-maximal}.
\end{proof}

\end{document}